\documentclass[a4paper,12pt]{article}
\usepackage{amsmath,amsthm,amsfonts,amssymb,bm,mathrsfs}
\usepackage{mathtools}
\usepackage{graphicx}
\usepackage{xcolor}
\usepackage{float}
\usepackage[super,square,numbers,sort&compress]{natbib} 
\usepackage{fancyhdr}
\usepackage{chemarrow}
\usepackage{extarrows}
\usepackage{comment}
\usepackage{enumerate}
\usepackage{wrapfig}
\usepackage{bibentry,natbib}
\setcitestyle{authoryear,open={(},close={)},nosort}
\usepackage[margin=1.25in]{geometry}
\usepackage[T1]{fontenc}
\usepackage{footmisc}
\usepackage{sgame}
\usepackage{color}
\usepackage{setspace}
\usepackage{fancyhdr}
\usepackage{hhline}
\usepackage[super]{nth}
\usepackage[nice]{nicefrac}
\usepackage{multicol}
\usepackage{tikz}
\usetikzlibrary{decorations.pathreplacing}
\usetikzlibrary{arrows}
\tikzstyle{block}=[draw opacity=0.7,line width=1.4cm]
\usepackage{pgfplots}
\pgfplotsset{width=10cm,compat=1.9}
\usepgfplotslibrary{external}
\usepackage{amsmath}
\usepackage{amsfonts}
\usepackage{lipsum}
\usepackage{amssymb}
\usepackage{lmodern}
\usepackage{enumitem}
\usepackage{dsfont}
\providecommand{\U}[1]{\protect\rule{.1in}{.1in}}
\allowdisplaybreaks
\usetikzlibrary{matrix,arrows,decorations.pathmorphing}

\newtheorem{theorem}{Theorem}

\newtheorem{claim}{Claim}

\newtheorem{definition}{Definition}

\newtheorem{lemma}{Lemma}

\newtheorem{corollary}{Corollary}

\newtheorem{exam}{Example}

\newcommand{\bN}{\mathbb{N}}
\newcommand{\bR}{\mathbb{R}}
\newcommand{\bRp}{\mathbb{R}_+}

\newcommand{\cD}{\mathcal{D}}

\newcommand{\cP}{\mathcal{P}}

\newcommand{\lr}{\mathcal{L}^0(\mathbb{R})}
\newcommand{\lrp}{\mathcal{L}^0(\mathbb{R}_+)}

\newcommand{\lxone}{\mathcal{L}^0(X_1)}
\newcommand{\lxtwo}{\mathcal{L}^0(X_2)}
\newcommand{\lxi}{\mathcal{L}^0(X_i)}
\newcommand{\lxj}{\mathcal{L}^0(X_j)}
\newcommand{\lxmi}{\mathcal{L}^0(X_{-i})}
\newcommand{\cone}{\overline{c}_1}
\newcommand{\ctwo}{\overline{c}_2}
\newcommand{\ci}{\overline{c}_i}

\newcommand{\lci}{\underline{c}_i}

\newcommand{\lcone}{\underline{c}_1}
\newcommand{\lctwo}{\underline{c}_2}

\DeclareMathOperator{\supp}{supp}

\usepackage{hyperref}
\definecolor{Redish}{HTML}{890F0F}
\hypersetup{colorlinks=true,linkcolor=Redish,filecolor=Redish,urlcolor=Redish,citecolor=Redish,}
\newcommand*{\fullref}[1]{\hyperref[{#1}]{\ref{#1} \nameref{#1}}} 
\newcommand*{\figref}[1]{\hyperref[{#1}]{Figure \ref{#1}}}
\newcommand*{\secref}[1]{\hyperref[{#1}]{Section \ref{#1}}}
\newcommand*{\tabref}[1]{\hyperref[{#1}]{Table \ref{#1}}}
\newcommand*{\Equationref}[1]{\hyperref[{#1}]{Equation \ref{#1}}}
\newcommand*{\equationref}[1]{\hyperref[{#1}]{equation \ref{#1}}}
\newcommand*{\propref}[1]{\hyperref[{#1}]{Proposition \ref{#1}}}
\newcommand*{\defref}[1]{\hyperref[{#1}]{Definition \ref{#1}}}
\newcommand*{\lemmaref}[1]{\hyperref[{#1}]{Lemma \ref{#1}}}
\newcommand*{\thmref}[1]{\hyperref[{#1}]{Theorem \ref{#1}}}
\newcommand*{\cororef}[1]{\hyperref[{#1}]{Corollary \ref{#1}}}

\def\sym#1{\ifmmode^{#1}\else\(^{#1}\)\fi}

\usepackage{multicol}
\usepackage{multirow}

\newcommand{\up}[1]{\overline{{#1}}}
\newcommand{\dw}[1]{\underline{{#1}}}

\begin{document}
\baselineskip=20pt

\title{\Large \textsc{A Theory of Choice Bracketing under Risk}\thanks{ I am  deeply indebted to Faruk Gul, Wolfgang Pesendorfer and Pietro Ortoleva for their invaluable advice,   encouragement and support throughout this project.  I have greatly benefited from discussions with Xiaosheng Mu and Rui Tang. I also thank Roland B\'{e}nabou, 
Modibo Camara, Sylvain Chassang, Xiaoyu Cheng, Francesco Fabbri, Shaowei Ke, Shengwu Li, Alessandro Lizzeri, Dan McGee, Lasse Mononen, Evgenii Safonov, Ludvig Sinander, Jo\~{a}o Thereze, Can Urgan, Leeat Yariv and seminar participants at Princeton Microeconomic Theory Student Lunch Seminar for  helpful comments and discussions.}}

\author{ Mu Zhang \thanks{Department of Economics, Princeton University, \href{mailto:muz@princeton.edu}{muz@princeton.edu}}
}
\date{August 24, 2021}

\maketitle

\begin{abstract}
\baselineskip=20pt
Aggregating risks from multiple sources can be complex and demanding,  and decision makers usually adopt heuristics to simplify the evaluation process.  This paper axiomatizes two closed related and yet different heuristics, narrow bracketing and correlation neglect, by relaxing the independence axiom in the expected utility theory. The flexibility of our framework allows for applications in various economic problems. First, our model can explain the experimental evidence of narrow bracketing over monetary gambles. Second, when one source represents background risk, we can accommodate \cite{ecta2000calibration}'s critique and explain risk aversion over small gambles. Finally, when different sources represent consumptions in different periods, we unify three seemingly distinct models of time preferences and propose a novel model that simultaneously satisfies indifference to temporal resolution of uncertainty, separation of time and risk preferences, and recursivity in the domain of lotteries. As a direct application to macroeconomics and finance, we provide an alternative to \cite{ecta1989EZ} which avoids the unreasonably high timing premium discussed in \cite*{aer2014time_premium}.\bigskip

 
\end{abstract}

\newpage



\section{Introduction}\label{section_introduction}

Decision makers in the real world usually face multiple risky choice problems. For instance, an investor might need to take care of her investment accounts simultaneously in different financial markets, including stocks, bonds, and cryptocurrencies. One implicit assumption of the long-standing focus on single choice problems in economics is that agents can rationally aggregate and assess risks and consequences in multiple sources. However,  multi-source risk is naturally more complex and challenging than single-source risk  and decision makers usually adopt heuristics to simplify the evaluation process.  In this paper, we will focus on two such heuristics, narrow bracketing and correlation neglect.

 Narrow bracketing, formalized by \cite{MS1985mental} and \cite*{JRU1999choice}, describes the situation where a decision maker (DM) faced with multiple choice problems tends to choose an option in each decision without full regard to other decisions. This is a simplifying heuristic as searching for a local optimum is  less costly than searching for a global optimum. As the building block of many behavioral models, narrow bracketing helps to explain financial anomalies like the equity premium puzzle \citep{qje1995myopic} and the stock market participation puzzle \citep*{aer2006framing}. It can also make a complex model tractable by assuming that agents optimize each decision in isolation \citep*{BJW2020prospect}.   Besides direct experimental evidence (\citealp{S1981framing}, \citealp{aer2009narrow}, \citealp{EF2020revealNB}), narrow bracketing itself also serves as an important and implicit assumption in many experiments.  For example, show-up fees are ubiquitous but almost no experiments on choices under risk take them into account when estimating the subjects' risk attitudes.  Also, evidence on loss aversion is only valid when we ignore all wealth accrued from decisions outside of the laboratory.

The second heuristic,  correlation neglect, describes the tendency of agents to ignore the interdependence among different decision problems and treat them as if they are independent. This simplifies the decision process since only marginal distributions need to be considered. Experimental evidence of correlation neglect has been found in various economic scenarios including belief formation \citep{res2019correlation}, portfolio allocation \citep{sej2009correlation, EW2016correlation} and school choice \citep*{RST2020correlation}. Correlation neglect is also an important element in many behavioral models. For instance, \cite{aer2015overconfidence} uses it as the micro-foundation of overconfidence in political behavior.

Despite the popularity of these heuristics in behavioral and experimental economics, they have received little attention in the choice-theoretic literature. One possible reason is that  they are typically interpreted as  behavioral or ``irrational'' biases and supposed to deviate drastically from the standard framework for choices under risk, since, for example, narrow bracketing can induce a violation of first order stochastic dominance \citep{aer2009narrow}. Also,  in many applications, narrow bracketing and correlation neglect are confounded with each other, or other behavioral factors like loss aversion and reference dependence.  However, recent literature urges for a better understanding of these two heuristics. Take narrow bracketing as an example.   Empirically, through a novel revealed preference test,  \cite{EF2020revealNB} find that most subjects are best described as either narrow or broad bracketing, even if intermediate cases are allowed. They suggest that narrow bracketing might be better viewed as a ``heuristic'' instead of a ``bias'', and it may occur when agents ``are unaware of how to broadly bracket, or or are unaware that broad bracketing can lead to notably higher payoffs, or choose to employ  to simplify their decision-making''. Theoretically, based on an impossibility result, \cite*{mpst2020background} suggest that ``theories that do not account for narrow framing $\cdots$ cannot explain commonly observed choices among risky alternatives.''

To our best knowledge, this paper is the first to provide a choice-theoretic foundation for narrow bracketing and correlation neglect as simplifying heuristics under multi-source risks.  We consider the preference of a DM over lotteries of two-dimensional outcome profiles $\mathcal{P}$. An outcome profile can be interpreted as the consequences of two decision problems,  such as simultaneous monetary gambles, intertemporal choices and consumption choices involving multiple goods. We will call each dimension  a {\it source} of risk and  the marginal distribution in some source as a {\it marginal lottery}. We start with the benchmark where the preference admits an expected utility (EU) representation, which is characterized by the von Neumann–Morgenstern (vNM) independence axiom. Then we axiomatize narrow bracketing and correlation neglect by relaxing the independence axiom.

The axiomatic approach has the following advantages. First, it provides a unified framework to study and compare the heuristics. For instance, besides the expected utility benchmark, our characterization theorem allows for models with either only narrow bracketing, or only correlation neglect, or both of them. Then the differences of behaviors induced by those models can be exactly attributed to the two heuristics. Second, since the two heuristics can be captured by intuitive and simple deviations from the expected utility benchmark, we argue that narrow bracketing and correlation neglect are not more behavioral or ``irrational'' than other commonly used non-EU models in the literature like the certainty effect and the (cumulative) prospect theory. This suggests that they might deserve more attention in both theoretical and applied works. Actually, one can even replace EU with any desired non-EU model as the benchmark and study the interaction of narrow bracketing and correlation neglect with other behavioral factors. Finally, as we will show later,  the axiomatic framework is so flexible  that it can be applied to various choice domains.

 In our model, narrow bracketing and correlation neglect are closed related heuristics  but they differ in the following sense. A narrow bracketer acts if she can perceive the correlation among different sources in each lottery correctly, but she is optimizing in one or both sources in isolation. In contrast, a DM who ignores correlation understands how to aggregate outcomes and optimizes globally,  but she misperceives the interdependence of risks in difference sources.

Our main results are two representation theorems. We first assume Axiom Correlation Neglect, where the DM always ignores the correlation among two sources of risk. This is consistent with the experimental designs for testing narrow bracketing where risks in different decision problems are resolved independently (\citealp{aer2009narrow}, \citealp{EF2020revealNB}). We call a lottery with independent marginals  a {\it product lottery} and denote the set of all product lotteries as $\hat{\cP}\subset \cP$. Correlation neglect implies that we can just focus on the preference restricted to the set of product lotteries. We start with the  {\it EU with correlation neglect (EU-CN)} model where the DM does not narrowly bracket risks.
$$V^{EU-CN}(P)  = \sum_{x,y}w(x,y)P_1(x)P_2(y), ~\forall P\in {\cP}.$$

Here, $w$ represents the DM's preference over deterministic outcomes. The only difference between EU-CN and EU is that the DM uses the product of the marginal distributions in the calculation, instead of the correct joint distribution. As a result, she can rationally aggregate risks if they are independent across different sources, but she cannot appreciate the correlation.

By comparison, a DM exhibits narrow bracketing will evaluate the marginal lotteries in isolation by first taking the certainty equivalents of them and then evaluating the vector of certainty equivalents. The corresponding {\it Narrow Bracketing (NB)} representation is 
$$V^{NB}(P)  =w(CE_{v_1}(P_1),CE_{v_2}(P_2)), ~\forall P\in {\cP}.$$

We also consider the intermediate case of narrow bracketing where the DM only narrowly brackets risks in one source instead both. Consider the intertemporal interpretation where source $1$ represents today and source $2$ represents tomorrow and suppose that the DM only narrowly brackets tomorrow's risks. The model is called {\it backward induction bracketing with correlation neglect (BIB-CN)} since the DM adopts the following backward
induction evaluation process. First, she reduces tomorrow's risk to its certainty equivalent. Then, she evaluates today's risk using expected utility. 
$$V^{BIB-CN}(P)  = \sum_{x} w(x,CE_{v_{2}}(P_2))P_1(x), ~\forall P\in {\cP}.$$

Symmetrically, we can consider {\it forward induction bracketing with correlation neglect (FIB-CN)} representation where the DM only narrowly brackets today's risks:
$$V^{FIB-CN}(P)  = \sum_{y} w(CE_{v_{1}}(P_1),y)P_2(y), ~\forall P\in {\cP}.$$
Our characterization result also allows for piece-wise combinations of NB and BIB-CN or FIB-CN and we call them {\it generalized backward induction bracketing with correlation neglect (GBIB-CN)} and {\it generalized forward induction bracketing with correlation neglect (GFIB-CN)}. It is easy to see that NB, BIB-CN and FIB-CN will be special cases and hence we have three classes of models under correlation neglect: EU-CN where the DM exhibits broad bracketing, GBIB-CN and GFIB-CN, where the DM  narrowly brackets risks in at least one source.

We then provide an axiomatic foundation for those representations. Besides Axiom Correlation Neglect and standard axioms including weak order,  monotonicity and  continuity, we have Axiom Weak Independence, which contains two relaxations of the vNM independence axiom on the set of product lotteries. The first relaxation says that if we fix the same marginal lottery in one source, then the independence axiom should hold in the other source. This guarantees that for any fixed background risk, the DM's behavior is consistent with the EU benchmark. The second part states that the independence axiom should hold locally in a source $i\in\{1,2\}$ where the DM does not narrowly bracket risk, that is, where the conditional preference for marginal lotteries in source $i$ \textit{does depend} on marginal lottery, or background risk, in the other source. In this way, we maintain the vNM independence axiom as general as possible and interpret narrow bracketing as one intuitive violation of it. Whenever the DM does not exhibit narrow bracketing in some source, then the independence property should hold in that source at least locally. Our first main result states that  a preference satisfies Axiom Correlation Neglect, Axiom Weak Independence and standard axioms if and only if it admits a representation among the three classes of models: EU-CN, GBIB-CN and GFIB-CN.

Our second result discards the correlation neglect assumption and allows the DM to correctly perceive the interdependence of risks in different sources. We generalize previous models to {\it backward induction bracketing (BIB)} and  {\it forward induction bracketing (FIB)}. For instance, the BIB representation is 
$$V^{BIB}(P)  = \sum_{x} w(x,CE_{v_{2}}(P_{2|x}))P_1(x), \forall P\in {\cP},$$
where $P_1$ is the marginal lottery of $P$ in source $1$ and $P_{2|x}$ is the conditional lottery of $P$ in source $2$ given outcome $x$ in source $1$. In contrast to BIB-CN, now the DM correctly perceives the correlation and uses the conditional risk in source 2 instead of the marginal risk. This functional form combines the insights of narrow bracketing and backward induction and  will be our focus in the application to time preferences. We then weaken Axiom Correlation Neglect to Axiom Correlation Sensitivity, which states that independence holds if the correlation structures are not affected by the mixture. Our characterization results \thmref{thm_BIB}
 and \cororef{coro_FIB} extend \thmref{thm_cn} by allowing for representations EU, BIB and FIB.

 Unlike many decision theory papers that only involve a single representation, our results characterize several seemingly distinct and extreme functional forms. We interpret the distinction as an important and necessary feature of our framework, since our goal is to model choice bracketing and correlation neglect as simplifying and intuitive heuristics, while many ``intermediate'' functional forms in our setup would be either complicated or hard to interpret.  For instance, our model excludes the ``partial narrow bracketing'' representation in the literature (e.g., \citealp*{aer2006framing}, \citealp{aer2009narrow}, \citealp{EF2020revealNB}), which features a weighted average of the broad bracketing representation (EU) and NB.\footnote{The NB representation in those papers differs from ours. Please refer to Section \ref{section_money} for a detailed discussion.} We will justify this exclusion in two ways. Empirically, \cite{EF2020revealNB} document that very few (around 5\%) of their  subjects are best classified to partial narrow bracketing. This suggests that incorporating such intermediate cases might not necessarily have larger explanatory power in practice. Theoretically, computing the weighted average of the broad bracketing utility function and the narrow bracketing utility function is arguably more complex and involving than computing either of them. This contradicts with our interpretation that choice bracketing should simplify, rather than complexify, the evaluation process compared to the EU benchmark.

Now we discuss applications of our model to various economic scenarios. First,  two sources of risks represent simultaneous and independent monetary gambles, then our NB model can be used to explain the experimental evidence regarding choice bracketing in \cite{S1981framing}, \cite{aer2009narrow} and \cite{EF2020revealNB}. Second, When we interpret the risk in source 1  as the background risk and the risk in source 2 as the gamble at hand, then we can study \cite{ecta2000calibration}'s critique on EU and risk aversion over small gambles following \cite{ecta2000calibration}, \cite*{aer2006framing} and \cite*{mpst2020background}. When a DM admits a NB representation, then she will ignores the effect of the background risk and hence can exhibit reasonable risk aversion over small gambles without inducing unrealistic risk aversion over large gambles.

Finally, when outcomes in different sources represent consumptions in different periods, then our framework can be used to study time preferences.  For example, the EU model includes the standard expected discounted utility model and its generalization -- the Kihlstrom-Mirman model (\citealp{KM1974}, \citealp*{DGO2020SI}). The BIB model is the counterpart  of the history-independent models in \cite{KP1978} on the set of lotteries, which are originally defined on the set of temporal lotteries that allow for risks to be resolved in different periods. The NB model is essentially the Dynamic Ordinal Certainty Equivalent (DOCE) model studied in (\citealp{ecta1978DOCE}, \citealp{selden1978DOCE}, \citealp*{wp2020DOCE}). Although the above three models have been considered as distinct ones and studied separately in the literature,\footnote{See the discussion in  \cite{ecta1989EZ}.} our representation theorem provides a unified framework for them and shows that their distinctions can be attributed to narrow bracketing and correlation neglect. This is surprising ex-ante since our analysis is based solely on simplifying heuristics to deal with multi-source risk and contains no normative properties of time preferences. The result hence reveals a deep connection between the two economic issues.

Based on BIB, We propose a novel model called KM-BIB, which satisfies various desirable normative properties such as indifference to temporal resolution of uncertainty,  separation of time and risk preferences,  recursivity in the domain of lotteries, stationarity and discounted utility when there is no risk. We identify a new connection between the empirical evidence of narrow bracketing in experiments (\citealp{aer2009narrow}) and the theoretical difficulty to satisfy ordinal dominance in recursive preferences \citep*{ecta2017recursive}. Then we focus on a special case of KM-BIB, which is an alternative to the  CRRA-CES Epstein-Zin (EZ) model \citep{ecta1989EZ}.  The CRRA-CES EZ model is built on \cite{KP1978} and  has been commonly adopted to explain many long-standing financial puzzles including the equity premium puzzle \citep{jf2004LLR} due to its separation of time and risk preferences. However, using introspection, \cite*{aer2014time_premium} argue that the parameter values  in \cite{jf2004LLR} would predict absurdly high value of early resolution of uncertainty. In contrast, when extended naturally to the infinite horizon, our alternative model is not subject to this critique without dampening the separation of time and risk preferences.

\bigskip

\noindent {\bf Related Literature.} The most closely related work to ours is \cite{V2020bracketing}, where the author simultaneously and independently develops a choice-theoretic model for choice bracketing within the expected
utility framework. In the model, each DM is endowed with a broad preference and a narrow preference, both of which are EU and can be observed or identified in the experiment. She provides two axioms to connect the two preferences by interpreting narrow bracketing as deviations from broad bracketing by ignoring correlation and changing the EU index. Our paper differ from \cite{V2020bracketing} in three aspects. First, in our framework, a DM only has one preference and we try identify where she is subject to choice bracketing and/or correlation neglect from her choice data over lotteries. Second, \cite{V2020bracketing} maintains the EU paradigm, while we interpret choice bracketing and correlation neglect as deviations from the EU benchmark. Actually, in the case with two sources, her representation of narrow bracketing lies in the intersection of our EU and NB representations. Finally, we allow for more general forms of choice bracketing and separation between choice bracketing and correlation neglect, while \cite{V2020bracketing} regards correlation neglect as an ingredient of choice bracketing.

Our paper is also related to the growing literature on explaining narrow bracketing with other factors. The first strand of literature assumes DM's limited attention to price or preference shocks in her consumption behavior. For instance, \cite{qje2020attention} show that an rationally inattentive consumer with imperfect information about the shocks would exhibit mental budgeting and naiver diversification. \cite{res2020thinking} proposes a theory of narrow thinking where the DM makes each decision with imperfect information of other decision problems. As a result, the optimization problem is equivalent to solving an incomplete information, common interest game played by multiple selves.  In contrast, we adopt an choice-theoretic approach to axiomatize choice bracketing directly and model it as a simplifying heuristic to deal with multi-source risks. Also, our model can be applied in simple settings without shocks to prices or preferences like experiments on choices under objective risk (\citealp{aer2009narrow}, \citealp{EF2020revealNB}). Hence the two approaches are complementary to each other. More recently, \cite{C2021separable} introduces the notion of computational complexity from computer science to EU with high-dimensional decisions. He shows that computational tractability requires the EU index to satisfy a slightly weaker version of additive separability and the tractable algorithm involves narrow bracketing. By comparison, we model narrow bracketing as deviations from the the EU paradigm to simplify the evaluation of risks. Actually, EU with an additive separable index lies in the intersection of our EU and NB representations.

\section{Primitives} \label{section_primitives}

Consider a two-dimensional outcome space $X=X_1\times X_2$, where $X_i$ is the set of outcomes in source $i\in \{1,2\}$.  Throughout the paper, we assume that $X_i$ is a nontrivial closed interval on the real line  that includes $0$. Formally,  for each $i=1,2$, $X_i=[\lci,\ci]\cap \bR$, where $\ci>\lci$, $0\in [\lci,\ci]$ and $\ci,\lci\in \bR\cup \{-\infty,+\infty\}$. Note that the outcome space can be either bounded or unbounded. We call $(x_1,x_2)\in X$  an {\it outcome profile} and $x_i$  the \textit{outcome} in source $i$ for $i\in \{1,2\}$. A positive outcome can be interpreted as a gain, while a negative one is a  loss.


A {\it (joint) lottery} is a probability measure on $X$ with a finite support. Denote $\cP$ as the set of all lotteries endowed with the topology of weak convergence and the standard mixture operation.  For each $i\in\{1,2\}$, we define $-i\in \{1,2\}$ where $i\neq -i$. For each lottery $P\in \mathcal{P}$, denote the {\it marginal lottery} of $P$ in source $i$  as $P_i\in \mathcal{L}^0(X_i)$ such that for each $x_i\in X_i$, $P_i(x_i) = \sum_{x_{-i}\in X_{-i}}P(x_1,x_2)$.  $P_1$ represents the marginal risk of $P$ in source $1$. Sometimes, we might also call  $p\in \mathcal{L}^0(X_1)\cup \mathcal{L}^0(X_2)$ a {\it single-source} lottery and  $P\in \cP$ a {\it multi-source} lottery. For each lottery $P\in \cP$ and marginal lottery $p\in \lxone \cup \lxtwo$, we denote $\supp(P):=\{(x_1,x_2)\in X_1\times X_2: P(x_1,x_2)>0\}$ and $\supp(p):=\{x\in X_1\cup X_2: p(x)>0\}$.  When there is no confusion, we write the degenerate marginal lottery $\delta_{x}$ as $x$ for $x\in X_1\cup X_2$.

We are especially interested in a subspace of lotteries called  {\it product lotteries} $\hat{\mathcal{P}}= \mathcal{L}^0(X_1)\times \mathcal{L}^0(X_2)\subsetneq \cP$. A  product lottery $(p,q)\in \hat{\mathcal{P}}$ is a lottery where the  marginal lotteries   $p$ and $q$ are independent from each other and hence the risks in two sources are not interdependent. To see when $\hat{\mathcal{P}}$ is relevant, note that in most experiments on  narrow bracketing (e.g., \cite{aer2009narrow} and \cite{EF2020revealNB}), the instructions would stress that risks in different monetary gambles are resolved independently. Similarly, studies involving background noises, such as \cite{ecta2000calibration}, \cite{ecta2008calibration} and \cite{mpst2020background}, typically assume that the background risk is independent from the risky decision at hand. Hence,  focusing on the domain of product lotteries $\hat{\mathcal{P}}$ is sufficient in those applications.

The primitive of our analysis is a binary relation $\succsim$ on $\cP$. We define the {\it narrow preference in source $1$} $\succsim_1$ as the restriction of $\succsim$ on $\lxone\times \{0\}$, that is, $p \succsim_1 q$ if and only if $(p,0)\succsim (q,0)$ for each $p,q\in \lxone$. Since the marginal lottery in source $2$ is fixed at $\delta_0$, the comparison of lotteries $(p,0)$ and $(q,0)$ can be interpreted as the comparison of marginal lottery $p$ and $q$ in source $1$. In this case, the DM faces the traditional one-dimensional decision problem in source 1 as if source 2 does not exist. Symmetrically, we denote $\succsim_2$ as the narrow preference in source $2$ by restricting $\succsim$ on $ \{0\}\times\lxtwo$. These notions will prove useful when we define choice bracketing in Section \ref{section_representation}.

It is worthwhile to mention that our framework can accommodate many different  economic applications, depending on our interpretations of the two sources of outcomes. For example, in lab experiments on choices under risk, they can represent money or tokens in two different gambles; in individual portfolio choice problems, they can represent account balances on the stock market and the bitcoin market respectively; in intertemporal consumption-savings problems, they can represent consumptions in two different periods. We will elaborate more on those applications in Section \ref{section_app}.

\section{Representations} \label{section_representation}
In this section, we introduce different decision rules adopted by a DM faced with two-source risk. We start with the expected utility model as the benchmark. As is discussed in the introduction, people in practice usually deviate from the benchmark systematically by adopting some simplifying heuristics. In the following we will focus on two such heuristics: choice bracketing and correlation neglect.

\subsection{Benchmark: Expected Utility}\label{section_EU}

For each function $f:X\rightarrow \bR$ or $f:X_i\rightarrow \bR$ for some $i=1,2$, we say $f$ is {\it regular} if it is continuous, strictly monotone and bounded. If the domain of $f$ is compact (i.e., when $X=[0,\cone]\times [0,\ctwo]$), then boundedness is implied by continuity and hence redundant. The definition of an expected utility representation is standard. 

\begin{definition}[EU]\label{def_EU} Let $\succsim$ be a binary relation on $\mathcal{P}$ and let $w: \mathbb{R}_+^2\rightarrow \mathbb{R}$ be a regular function. The utility index $w$ is an  {\bf expected utility (EU)} representation of $\succsim$ if $\succsim$ is represented by $V^{EU}:{\mathcal{P}}\rightarrow \mathbb{R}$, which is defined by
$$V^{EU}(P)  = \sum_{x,y} w(x,y)P(x,y).$$
\end{definition}

\subsection{Heuristic One: Choice Bracketing}\label{section_NB}
We first consider choice bracketing, where the DM might evaluate risks in different sources in isolation and hence make choices in some source without regard to lotteries in the other source.  We start with the case where decisions in both sources are made separately.  For each $i\in \{1,2\}$ and each regular function $f:X_i\rightarrow \mathbb{R}$, we denote the {\it certainty equivalent} of $p\in \mathcal{L}^0(X_i)$ under $f$ as $CE_f(p)= f^{-1}(\sum_x f(x)p(x))$. As $X_i$ is a closed interval and $f$ is strictly monotone and continuous,  the certainty equivalent is well-defined and $CE_f(p)\in X_i$ for each $p\in \lxi$. The definition of narrow bracketing representation is as follows.

 \begin{definition}[NB]\label{def_NB} Let $\succsim$ be a binary relation on $\mathcal{P}$ and let $w: X_1\times X_2\rightarrow \mathbb{R}, v_i: X_i\rightarrow \mathbb{R}, i=1,2,$ be regular functions. The tuple $(w,v_1,v_2)$ is a  {\bf (fully) narrow bracketing (NB)} representation of $\succsim$ if $\succsim$ is represented by $V^{NB}:{\mathcal{P}}\rightarrow \mathbb{R}$, which is defined by
$$V^{NB}(P)  =w(CE_{v_1}(P_1),CE_{v_2}(P_2)).$$
\end{definition}
 
Intuitively, $v_i$ is the EU index of narrow preference $\succsim_i$ for $i=1,2$ and $w$ represents the DM's preference in the absence of risk.  If $\succsim$ admits a NB representation, then the DM evaluates each lottery by first reducing the marginal lotteries in both sources to their certainty equivalents under source-sensitive utility indices. Hence it captures the idea that choices made in source $1$ are independent of alternatives in source $2$ and vice versa.

Then we consider the case where the DM adopts partial narrow bracketing and choices made in one source are independent of alternatives in the other source while the reverse fails. For each lottery $P\in \mathcal{P}$ and $x$ in the support of $P_1$, i.e., $P_1(x)>0$, we denote the conditional lottery $P_{2|x}$ as the conditional distribution of outcomes in source 2 given $x$ in source  1, which represents the conditional risk in source $2$. Formally, for each $y\in X_2$, $P_{2|x}(y)=P(x,y)/P_1(x)$. Then we say a preference admits a backward induction bracketing representation if the DM first reduces the conditional risks in source $2$ to their certainty equivalents and then evaluates the risk in source $1$.   

 \begin{definition}[BIB]\label{def_BIB} Let $\succsim$ be a binary relation on $\mathcal{P}$ and let $w: X_1\times X_2\rightarrow \mathbb{R}, v_2: X_2\rightarrow \mathbb{R}$ be regular functions. The tuple $(w,v_2)$ is a  {\bf  backward induction bracketing (BIB)} representation of $\succsim$ if $\succsim$ is represented by $V^{BIB}:{\mathcal{P}}\rightarrow \mathbb{R}$, which is defined by
$$V^{BIB}(P)  = \sum_{x} w(x,CE_{v_{2}}(P_{2|x}))P_1(x).$$
\end{definition}

In a BIB representation $(w,v_2)$, the DM adopts the following backward induction evaluation process: i) conditional on each possible outcome $x$ in source 1, the DM first evaluates the conditional risk in source 2 by replacing the conditional lottery $P_{2|x}$ with its certainty equivalent under EU index $v_2$; ii) Then the DM evaluates the risk in source 1 using EU index $w$. It is important to notice that the evaluation of conditional risks is independent of the outcome in source $1$, which turns out to be the key behavioral deviation of BIB from the EU benchmark. Actually, conditional on outcome $x$ in source $1$, if we replace $v_2$ in the BIB representation with $w(\cdot, x)$, then we will exactly get the EU representation with index $w$.  Hence, BIB captures the idea of narrowly bracketing risks in source $2$. In Section  \ref{section_app}, we will discuss the implications of BIB in time preferences and compare it with the well-known Kreps-Porteus preferences, which also admits a  backward induction interpretation.

 Symmetrically, when the DM only narrowly brackets risks in source 1, then we derive the  forward induction bracketing representation. For each $P\in \mathcal{P}$ and $y$ in the support of $P_2$,  we denote $P_{1|y}$ as the conditional distribution of outcomes in source 1 given outcome $y$ in source 2.

\begin{definition}[FIB]\label{def_FIB} Let $\succsim$ be a binary relation on $\mathcal{P}$ and let $w: X_1\times X_2\rightarrow \mathbb{R}, v_1: X_1\rightarrow \mathbb{R}$ be regular functions. The tuple $(w,v_1)$ is a  {\bf  forward induction bracketing (FIB)} representation of $\succsim$ if $\succsim$ is represented by $V^{FIB}:{\mathcal{P}}\rightarrow \mathbb{R}$, which is defined by
$$V^{FIB}(P)  = \sum_{x} w(CE_{v_{1}}(P_{1|y}),y)P_2(x).$$
\end{definition}

In applications where there is a natural order on the two sources, FIB might be more appropriate than BIB. For instance, risk in source 1 can be interpreted as the background risk or endowment risk, while risk in source 2 can be interpreted as the risk in the current decision problem such as the portfolio choice.  In an experiment on choices under risk, background risk includes show-up fees, payoffs from other rounds in the experiment and wealth outside the laboratory.

\subsection{Heuristic Two: Correlation Neglect}\label{section_CN}

In this section, we consider the second simplifying heuristic: correlation neglect, where the DM finds it difficult to deal with the correlation structure of risks in different sources and hence treats the lottery as if its marginal lotteries are independent from each other. We will introduce the counterparts of previous representations by imposing correlation neglect.

We start with a DM who is only subject to correlation neglect compared to the EU benchmark. 

\begin{definition}[EU-CN]\label{def_EU-CN} Let $\succsim$ be a binary relation on $\mathcal{P}$ and let $w:  X_1\times X_2\rightarrow \mathbb{R}$ be a regular function. The utility index $w$ is an  {\bf expected utility with correlation neglect (EU-CN)} representation of $\succsim$ if $\succsim$ is represented by $V^{EU-CN}:{\mathcal{P}}\rightarrow \mathbb{R}$, which is defined by
$$V^{EU-CN}(P)  = \sum_{x,y}w(x,y)P_1(x)P_2(y).$$
\end{definition}

The behavior of a DM with an EU-CN representation agrees with the EU benchmark on the set of product lotteries $\hat{\cP}$, but she ignores the interdependence of risks from different sources even if they are not independent.  

Now we study the interplay of choice bracketing and correlation neglect. First, it is easy to see that NB satisfies correlation neglect as the DM takes certainty equivalents of the marginal lotteries directly.  Second, suppose that the DM narrowly brackets marginal risks in source 2 after ignoring the correlation structure, then we get the following representation.

\begin{definition}[BIB-CN]\label{def_BIB-CN} Let $\succsim$ be a binary relation on $\mathcal{P}$ and let $w: X_1\times X_2\rightarrow \mathbb{R}, v_2: X_2\rightarrow \mathbb{R}$ be regular functions. The tuple $(w,v_2)$ is a  {\bf  backward induction bracketing with correlation neglect (BIB-CN)} representation of $\succsim$ if $\succsim$ is represented by $V^{BIB-CN}:{\mathcal{P}}\rightarrow \mathbb{R}$, which is defined by
$$V^{BIB-CN}(P)  = \sum_{x} w(x,CE_{v_{2}}(P_{2}))P_1(x).$$

\end{definition}

Our characterization results in Section \ref{section_axioms} allow for a general representation that incorporates both NB and BIB-CN as special cases. The generalization is based on the idea that whether the DM narrow brackets the marginal risk in source 1 might depend on the marginal risk in source 2. For example, suppose the two sources represent today and tomorrow respectively and the DM's preference is represented by the following function for some fixed outcome $a\in X_2$ tomorrow: 
$$U(p,q)=\begin{cases}
w(CE_{v_1}(p),CE_{v_2}(q)), \hbox{~if~} CE_{v_2}(q)\leq a,\\
\sum_x w(x,CE_{v_2}(q))p(x), \hbox{~if~} CE_{v_2}(q)> a.\end{cases}$$
That is, the utility representation adopts a threshold structure and the DM might be either NB or BIB-CN depending on the certainty equivalent of tomorrow's marginal lottery. Intuitively, if tomorrow's stakes are low, then the DM might make today's choices independent of tomorrow's outcomes to simplify the decision process. If instead tomorrow's stakes are high enough, she would be more careful about evaluating today's risk by taking into account the income effect of tomorrow's lottery. To some extent, this example can be interpreted as a version of endogenous choice bracketing. In order to keep continuity on the boundary (i.e., when $CE_{v_2}(q)= a$), $w(\cdot,a)$ must be a positive affine transformation of $v_1$.

The next definition extends the above idea by generalizing the threshold structure that determines when the DM switches between NB and BIB-CN. 

\begin{definition}[GBIB-CN]\label{def_GBIB-CN} Let $\succsim$ be a binary relation on $\mathcal{P}$, let $w:  X_1\times X_2\rightarrow \mathbb{R}, v_i: X_i\rightarrow \mathbb{R}, i=1,2$, be regular functions and let $H_2$ be an open subset of $X_2$ with $0\not\in H_2$. The tuple $(w,v_1,v_2,H_2)$ is a  {\bf generalized backward induction bracketing with correlation neglect (GBIB-CN)} representation of $\succsim$ if $\succsim$ is represented by $V^{GBIB-CN}:{\mathcal{P}}\rightarrow \mathbb{R}$, which is defined by
$$V^{GBIB-CN}(P)= \begin{cases}
w(CE_{v_1}(P_1),CE_{v_2}(P_2)), \hbox{~if~} CE_{v_2}(P_2)\in X_2 \backslash H_2,\\
\sum_x w(x,CE_{v_2}(P_2))P_1(x), \hbox{~if~} CE_{v_2}(P_2)\in H_2,\end{cases}$$
where for any $y\in \partial H_2$, i.e., the boundary of set $H_2$, $w(\cdot,y)$ is a positive affine transformation of $v_1$.  

\end{definition}


Notice that an open subset of the  real line can be represented by a countable union of disjoint open intervals.\footnote{The proof is given by \lemmaref{lemma_open set} in the appendix.} Hence, the GBIB-CN representation captures the idea that locally the DM exhibits narrow bracketing either  in both sources, or just in source 2. Specifically, when $H_2$ is empty, GBIB-CN reduces to NB; when the closure of $H_2$ is $X_2$, GBIB-CN reduces to BIB-CN.

Symmetrically, we can modify the definitions of BIB-CN and GBIB-CN to accommodate the case where the DM narrowly brackets risks in source $1$.

\begin{definition}[FIB-CN]\label{def_FIB-CN} Let $\succsim$ be a binary relation on $\mathcal{P}$ and let $w: X_1\times X_2\rightarrow \mathbb{R}, v_1: X_1\rightarrow \mathbb{R}$ be regular functions. The tuple $(w,v_1)$ is a  {\bf  forward induction bracketing with correlation neglect (FIB-CN)} representation of $\succsim$ if $\succsim$ is represented by $V^{FIB-CN}:{\mathcal{P}}\rightarrow \mathbb{R}$, which is defined by
$$V^{BIB-CN}(P)  =\sum_y w(CE_{v_1}(P_1),y)P_2(y).$$

\end{definition}

\begin{definition}[GFIB-CN]\label{def_GFIB-CN} Let $\succsim$ be a binary relation on $\mathcal{P}$, let $w:  X_1\times X_2\rightarrow \mathbb{R}, v_i: X_i\rightarrow \mathbb{R}, i=1,2$, be regular functions and let $H_1$ be an open subset of $X_1$ with $0\not\in H_1$. The tuple $(w,v_1,v_2,H_1)$ is a  {\bf generalized forward induction bracketing with correlation neglect (GFIB-CN)} representation of $\succsim$ if $\succsim$ is represented by $V^{GFIB-CN}:{\mathcal{P}}\rightarrow \mathbb{R}$, which is defined by
$$V^{GFIB-CN}(P)= \begin{cases}
w(CE_{v_1}(P_1),CE_{v_2}(P_2)), \hbox{~if~} CE_{v_1}(P_1\in X_1 \backslash H_1,\\
\sum_y w(CE_{v_1}(P_1),y)P_2(y), \hbox{~if~} CE_{v_1}(P_1)\in H_1,\end{cases}$$
where  for any $x\in \partial H_1$, $w(x,\cdot)$ is a positive affine transformation of $v_2$.  

\end{definition}

We end this section with some remarks on the above representations: EU, EU-CN, BIB, GBIB-CN, FIB and GFIB-CN. It is worthwhile to mention that each of those functional forms has intuitive and clear implications on the extent to which the DM adopts choice bracketing and correlation neglect. Moreover, each deviation from the EU benchmark deals with the multi-source risk in a relatively simpler way in terms of computation. Our results in the next section characterize those seemingly extreme representations by relaxing the standard vNM independence axiom in a reasonable manner. This approach is different from a typical decision theory paper which would involve a universal representation. We interpret the distinction as an important and necessary feature of our framework instead of a drawback, since our goal is to model choice bracketing and correlation neglect as simplifying and intuitive heuristics, while many ``intermediate'' functional forms in our setup would be either complicated or hard to interpret.  

One natural way to unify two representations is to consider their weighted averages. For instance, one popular representation in the literature of choice bracketing (e.g., \citealp*{aer2006framing}, \citealp{aer2009narrow}, \citealp{EF2020revealNB}) is ``partial narrow bracketing'', which features an $\alpha$-mixture of EU and NB\footnote{Actually the NB representation used in the literature differs from our Definition \ref{def_NB}. We will discuss their distinction in Section \ref{section_app} and argue why our version might be more appropriate.}. Besides axiomatic reasons, we justify our exclusion of such ``intermediate'' representations in two ways. First, the computation of the weighted average utility of EU and NB is arguably more complex and involving than the computation of either representation. This contradicts with our interpretation that choice bracketing should simplify the evaluation process compared to the EU benchmark. Second, using three well-designed experiments, \cite{EF2020revealNB} show that very few (around 5\%) of their  subjects are best classified to partial narrow bracketing. This suggests that incorporating such intermediate cases might not necessarily have larger explanatory power in practice. Similar arguments can be employed to justify why we exclude intermediate representations of correlation neglect.

\section{Axioms}\label{section_axioms}
In this section, we present our axioms and characterization theorems. \thmref{thm_cn} focuses on the representations that exhibit correlation neglect, that is, EU-CN, BIB-CN and FIB-CN.  Then we extend the result to  \thmref{thm_BIB} and \cororef{coro_FIB} to incorporate models without correlation neglect.

We start with the axioms shared by the two characterization results. The first axiom assumes rationality of the DM. 

  \medskip

  \noindent{\bf Axiom  Weak Order}: $\succsim$ is complete and transitive.

  \medskip

The next axiom is about monotonicity of the preference $\succsim$ with respect to some notion of dominance. In the case with single-source risk, there is an agreed definition of first order stochastic dominance. However, its extension to multiple sources is not self-obvious. Luckily, we only need a weak notion of dominance, which only involves the comparison of a lottery with a degenerate lottery. For any lottery  $P\in \cP$ and degenerate lottery $(x_1,x_2)\in X_1\times X_2$, we say $P$ {\it dominates} $({x_1},{x_2})$ if $P\neq ({x_1},{x_2})$ and $y_1\geq x_1, y_2\geq x_2$ for all $y_1\in \supp(P_1), y_2\in\supp(P_2)$. Symmetrically,  we say $({x_1},{x_2})$ {\it dominates} $P$ if $P\neq ({x_1},{x_2})$ and $y_1\leq x_1, y_2\leq x_2$ for all $y_1\in \supp(P_1), y_2\in\supp(P_2)$. Then Axiom 2 states that the preference $\succsim$ is monotonic with respect to dominance. 

\medskip

\noindent {\bf Axiom  Monotonicity}: For each $P\in \cP$ and $(x_1,x_2)\in X_1\times X_2$, $P\succ ({x_1},{x_2})$ if $P$ dominates $({x_1},{x_2})$ and $({x_1},{x_2})\succ P$ if  $({x_1},{x_2})$ dominates $P$. 

\medskip

Now we will introduce the continuity axiom. One reasonable candidate is the standard topological continuity axiom, which guarantees that $\succsim$ has a continuous representation.

\medskip

\noindent{\bf Axiom  Continuity}: For each $Q\in {\mathcal{P}}$, the sets $\{P\in \mathcal{P}: P\succ Q\}$ and $\{P\in \mathcal{P}: Q\succ P\}$ are open subsets of ${\mathcal{P}}$.

    \medskip
    
However, an important observation is that BIB violates Axiom Continuity generically. To see why, recall that a DM with BIB evaluates each lottery $P$ by first replacing the conditional  lotteries in source 2 with its certainty equivalent and then taking expected utility for the constructed new lottery. This would result in discontinuity when a small change in the lottery leads to a drastic change in the conditional lotteries. For instance, suppose that  $\succsim$ admits a BIB representation $(w,v_2)$. For each positive integer $n$, define $P^n=1/2(\delta_1,\delta_2) + 1/2(\delta_{1-1/n},\delta_3)$. Easy to see that $P^n$ weakly converges to $P=1/2(\delta_1,\delta_2) + 1/2(\delta_{1},\delta_3)$. Then Axiom Continuity requires 
$$\frac{1}{2}w(1,2) + \frac{1}{2}w(1,3)= w(1,CE_{v_2}(\frac{1}{2}\delta_2 +\frac{1}{2} \delta_3))$$
which implies that $v_2$ should be related to $w(1,\cdot)$. Actually, we can show that under Axiom Continuity, a preference that admits a BIB representation also admits an EU representation. 

In the above example of $P^n$ and $P$, the drastic change in conditional lotteries results from the fact that $P^n$ are not product lotteries and  not all outcomes in the support of $P^n_1$ change as $n$ increases. This captures the key insights of how BIB violates Axiom Continuity. Similar arguments hold for FIB. In order to maintain continuity as strong as possible while allowing for BIB and FIB, we weaken Axiom Continuity into three parts. 

The first part guarantees that topological continuity holds on the set of product lotteries. 

\medskip

\noindent{\bf Axiom Topological Continuity over Product Lotteries}:  For each $Q\in {\mathcal{P}}$, the sets $\{P\in \hat{\mathcal{P}}: P\succ Q\}$ and $\{P\in \hat{\mathcal{P}}: Q\succ P\}$ are open subsets of $\hat{\mathcal{P}}$.

\medskip

The second part states that continuity holds if we only change the probability weights without changing the outcomes in the support. This is exactly the notion of mixture continuity.

\medskip

\noindent{\bf Axiom Mixture Continuity}:  For each $P,R,Q\in \mathcal{P}$, the sets $\{\alpha\in[0,1]: \alpha P + (1-\alpha)Q\succ R\}$ and $\{\alpha\in[0,1]: R\succ \alpha P + (1-\alpha)Q\}$ 	are open subsets of $[0,1]$ in the relative topology.
 
    \medskip

By comparison, the last part deals with continuity concerning changes of outcomes in the support instead of the probability weights. To avoid drastic variation in the conditional lotteries, we need to make sure that all outcomes in the same  source change by the same amount unless they have reached the bounds of the outcome space. This can be achievable by a modified notion of convolution with tight upper bounds. For each $P\in\cP$ and $a_1,a_2>0$, we define $P * (a_1,a_2)\in \cP$ such that the probability of $(x,y)\in X_1\times X_2$ in $P$ is transferred to $(\min\{x+a_1, \cone\},\min\{ y+a_2,\ctwo\})$.\footnote{The formal definition of $P * (a_1,a_2)$ is as follows. Recall that $X_i^o$ is the  interior of $X_i$, $i=1,2$. For each $(x,y)\in X_1\times X_2$, if $x+a_1\in X_1^o, y+a_2\in X_2^o$, $P *(\delta_{a_1},\delta_{a_2}) (x+a_1, y+a_2)=P(x,y)$; if $y+a_2\in X_2^o$, $P *(\delta_{a_1},\delta_{a_2}) (\cone, y+a_2) = \sum_{x+a_1 >\cone} P(x, y)$; if  $x+a_1\in X_2^o$, $P *(\delta_{a_1},\delta_{a_2}) (x+a_1, \ctwo) = \sum_{y+a_2 >\ctwo} P(x, y)$. In addition, $P *(\delta_{a_1},\delta_{a_2}) (\cone, \ctwo) = \sum_{x+a_1>\cone, y+a_2 >\ctwo} P(x, y)$.  } Intuitively,  lottery $P * (a_1,a_2)$ is lottery $P$ plus a sure gain of $a_i$ in source $i$ for $i=1,2$, up to the upper bounds imposed by the outcome space. Similarly, we can define $p *\delta_a$ for $p\in \lxone \cup\lxtwo$.  The third part of the continuity axiom guarantees that $\succsim$ is continuous as sure gains converge to $0$. 

\medskip

\noindent{\bf Axiom Continuity over Sure Gains}:  For each $P,Q\in \cP$ and any two sequences $\epsilon_n,\epsilon_n'$ such that for each $n$, $\epsilon_n,\epsilon_n'>0$, and $\epsilon_n,\epsilon_n'\rightarrow 0$ as $n\rightarrow \infty$,
  	$$P *(\delta_{\epsilon_n}, \delta_{\epsilon_n'})\succsim Q,~\forall n \Longrightarrow P\succsim Q \hbox{~and~} Q\succsim P *(\delta_{\epsilon_n}, \delta_{\epsilon_n'}),~\forall n \Longrightarrow Q\succsim P.$$ 
 
    \medskip

Our Axiom Weak Continuity summarizes the above three relaxations of Axiom Continuity. 

\medskip

\noindent{\bf Axiom Weak Continuity}:  $\succsim$ satisfies  Axiom Topological Continuity over Product Lotteries, Axiom Mixture Continuity and Axiom Continuity over Sure Gains.
 
    \medskip

Now consider the standard vNM independence axiom, which characterizes EU. 

  \medskip
 \noindent {\bf Axiom Independence}: For each $P,Q,R\in \hat{\mathcal{P}}$ and $\alpha\in (0,1)$,
  $$P\succ Q\Longrightarrow \alpha P + (1-\alpha) R\succ \alpha Q + (1-\alpha) R.$$
  
  \medskip

Under Axiom Weak Continuity, easy to see that Axiom  Independence is equivalent to the following stronger axiom. 

  \medskip
 \noindent {\bf Axiom  Bi-independence}: For each $P,Q,R, S\in \hat{\mathcal{P}}$ and $\alpha\in (0,1)$,
  $$P\succ Q, R\sim S\Longrightarrow \alpha P + (1-\alpha) R\succ \alpha Q + (1-\alpha) S.$$
  
  \medskip

In order to introduce our relaxation of Axiom Bi-independence, we first assume correlation neglect and focus on the set of product lotteries. Then we will relax correlation neglect to an independence axiom over lotteries which differ in the correlation structure.

\subsection{Correlation Neglect}\label{subsection_cn}

 {\bf Axiom Correlation Neglect}: For each $P\in\mathcal{P}$, $P\sim (P_1, P_2)$.
 
  \medskip
Axiom Correlation Neglect states that the DM is indifferent between each lottery and the product lottery with the same marginals. Then it suffices to study the preference $\succsim$ restricted on the set of  product lotteries $\hat{\mathcal{P}}$. This suits the applications where risks from different sources are independent such as experiments on choice bracketing \citep*{aer2006framing}.  Another interesting example is the Nash equilibrium in a two-player game. \cite{fishburn1982book} characterizes multilinear utility, which is exactly EU-CN restricted to $\hat{\mathcal{P}}$, as a foundation for expected utility in the 2-player game involving mixed strategies. The next axiom is key to \cite{fishburn1982book}'s results. 

 \medskip
  {\bf Axiom Multilinear Independence}: For each $P,Q,R,S\in \hat{\mathcal{P}}$, $\alpha\in (0,1)$ and $i,j\in\{1,2\}$, if $P_i=R_i, Q_j=S_j$, then
  $$P\succ Q, R\sim S\Longrightarrow \alpha P + (1-\alpha) R\succ \alpha Q + (1-\alpha) S.$$
  
  \medskip

In contrast to Axiom Bi-independence, Axiom Multilinear Independence imposes two restrictions on the independence property. First, we only consider product lotteries. Second, whenever we want to mix two product lotteries, their marginal lotteries should be the same in at least one source. Technically, this is required to guarantee that the mixed lottery also has independent marginals\footnote{Notice that the set of product lotteries $\hat{\cP}$ is not a mixture space under the mixture operation defined on $\cP$. For instance, $(\delta_0,\delta_0)$ and $(\delta_1,\delta_1)$ are product lotteries, but their mixture $1/2(\delta_0,\delta_0) + 1/2(\delta_1,\delta_1)$ is not.}. The following lemma directly follows from \cite{fishburn1982book} and characterizes EU-CN.

\begin{lemma}\label{lemma_eu-cn}\emph{\citep{fishburn1982book}} Let $\succsim$ be a binary relation on $\mathcal{P}$. The following statements are equivalent:

i). The relation $\succsim$ satisfies Weak Order, Monotonicity, Weak Continuity, Multilinear Independence and Correlation Neglect;

ii). There exists an EU-CN representation of $\succsim$.
	
\end{lemma}

However, when we incorporate choice bracketing, Axiom Multilinear Independence will also be violated. 

\begin{exam}\label{exam_MI}
	Suppose that $\succsim$ admits a NB representation $(w,v_1,v_2)$ with $w(x,y)=x+y$ and $v_1(x)=v_2(x)=\sqrt{x}$ for all $x,y\geq 0$.  Let $p_1=\delta_{25}$, $q_1=\delta_{16}$, $r=\delta_{0}$, $s=\delta_{9}$, $q_2=\delta_{25}$ and $p_2=\delta_{(4+\epsilon)^2}$ for some $\epsilon>0$. Then 
	\begin{equation*}
	V^{NB}(p_1,p_2) = 	25 + (4+\epsilon)^2>  16 +25 = V^{NB}(q_1,q_2) \Longrightarrow (p_1, p_2)\succ (q_1,q_2),
	\end{equation*}
\begin{equation*}
	V^{NB}(p_1,r) = 	25 + 0= 16 +9 = V^{NB}(q_1,s) \Longrightarrow (p_1, r)\sim (q_1,s).
	\end{equation*}
	However, for $\alpha=1/2$, the utilities of the mixed lotteries are 
	\begin{equation*}
	V^{NB}(p_1,\alpha p_2 + (1-\alpha)r) = 	25 + \frac{(4+\epsilon)^2}{4}, ~~~ V^{NB}(q_1,\alpha q_2 + (1-\alpha)s)= 16 + 16.
	\end{equation*}
	If $0<\epsilon< 2\sqrt{7}-4$, then
	\begin{equation*}
	V^{NB}(p_1,\alpha p_2 + (1-\alpha)r) < V^{NB}(q_1,\alpha q_2 + (1-\alpha)s)   \Longrightarrow (p_1,\alpha p_2 + (1-\alpha)r)\prec (q_1,\alpha q_2 + (1-\alpha)s).
	\end{equation*}
	Hence Axiom Multilinear Independence fails. 
\end{exam}

Now we introduce our main independence axiom in the set of product lotteries. For each $i\in\{1,2\}$, we denote $-i\in \{1,2\}$ with $i\neq -i$. 

\medskip

\noindent {\bf Axiom  Weak Independence}:
\begin{enumerate}
  	\item [(i)] Axiom Conditional Independence: For $p,q,r,s\in \lxone, p',q',r',s'\in\lxtwo$ and $\alpha\in (0,1)$,
  $$(s,p')\succ (s,q')\Longrightarrow  (s,\alpha p' + (1-\alpha)r')\succ (s,\alpha q' + (1-\alpha)r').$$
  $$(p,s')\succ (q,s')\Longrightarrow  (\alpha p + (1-\alpha)r,s')\succ (\alpha q + (1-\alpha)r,s').$$
  	\item [(ii)] Axiom Weak Multilinear Independence: For each $P,Q,R,S\in\hat{\mathcal{P}}$, $\alpha\in (0,1)$ and $i,j\in\{1,2\}$, if $P_i=R_i,  Q_j=S_j$, $P_{-i}\sim_{-i} R_{-i}$ and $Q_{-j}\sim_{-j} S_{-j}$, then
  $$P\succ Q, R\sim S\Longrightarrow \alpha P + (1-\alpha) R\succ \alpha Q + (1-\alpha) S.$$
  \end{enumerate}

\medskip

The first part of the axiom states that if we fix the marginal lottery in one source, then the vNM independence axiom holds for marginal lotteries in the other source. For each $p\in \lxone$, we denote $\succsim_{2|p}$ as the restriction of $\succsim$ on $\{p\}\times \lxtwo$. $\succsim_{2|p}$ can be interpreted as the conditional preference in source $2$ given lottery $p$ in source $1$. When $p=\delta_0$, $\succsim_{2|p}$ agrees with $\succsim_2$, the narrow preference in source $2$. Similarly, we can define $\succsim_{1|q}$ as the conditional preference in source $1$ given $q\in \lxtwo$. Along with Axiom Weak Order and Axiom Weak Continuity, Conditional Independence guarantees that each conditional preference admits an EU representation. Hence, choice bracketing differs from the EU benchmark in terms of how the evaluations of the two marginal lotteries are aggregated. 

The second part is a local version of Axiom Multilinear Independence. It requires that the independence property holds only if the two product lotteries that are  mixed are ``similar'' enough, in the sense that they should agree on the marginal lottery in one source, and their marginal lotteries  in the other source should be indifferent according to the narrow preference. To see why this reflects choice bracketing, suppose that the DM narrowly brackets risks in source $-i$,  then she will evaluate the marginal lottery in source $-i$ using the narrow preference $\succsim_{-i}$, regardless of the marginal lottery in source $i$.  $P_i=R_i$ and $P_{-i}\sim_{-i} R_{-i}$ suggests that $P$ should be indifferent to $R$, which implies $\alpha P + (1-\alpha) R\sim P$ by Conditional Independence. Also, we know $Q\succ\alpha Q + (1-\alpha) S$. Hence Weak Multilinear Independence holds trivially and the axiom is redundant. Similarly arguments hold if DM narrowly brackets risks in source $-j$. As a result, Weak Multilinear Independence is not redundant only if either the DM broadly brackets risks or she only narrowly brackets risks in source $i=j$. In the latter case, mixture of lotteries only occurs in source $-i$. To conclude, the second part of Axiom Weak Independence states that the independence property holds locally for a source if the DM does not narrowly brackets  risks in that source. This explains why we interpret choice bracketing as violations of the independence property.

Now we are ready to state our first representation theorem under correlation neglect.

\begin{theorem}\label{thm_cn}
 Let $\succsim$ be a binary relation on $\mathcal{P}$. The following statements are equivalent:

i). The relation $\succsim$ satisfies Weak Order, Monotonicity, Weak Continuity, Weak Independence and  Correlation Neglect;

ii). The relation $\succsim$ admits one of the following representations: EU-CN, GBIB-CN and GFIB-CN.

    Moreover, in all representations $H_1,H_2$ are unique, $v_1,v_2$ are unique up to a positive affine transformation and in EU-CN, $w$ is unique up to a positive affine transformation.
\end{theorem}

It is worthwhile to mention that \thmref{thm_cn} characterizes three seemingly distant representations with correlation neglect and choice bracketing,  while the axioms do not seem to predict such a feature ex ante. Moreover, EU-CN, GBIB-CN and GFIB-CN all satisfy Axiom Continuity and we keep the Axiom Weak Continuity in \thmref{thm_cn}  just for consistency with \thmref{thm_BIB} below.

\subsection{Correlation Sensitivity}\label{subsection_cc}

In this section we discard Axiom Correlation Neglect to incorporate representations that are sensitive to the correlation structure of risks in different sources. Since Axiom Weak Independence only involves product lotteries, we need another independence axiom for lotteries whose marginals are not independent. 

\medskip

\noindent {\bf Axiom  Correlation Consistency}: Suppose $P\succ Q$ with  $P_i=Q_i$, $i=1,2$, $R\sim S$ and $$ supp(P_1)\cap supp(R_1)= supp(Q_1)\cap supp(S_1)=\emptyset,$$ then for all $\alpha\in (0,1)$ 
$$\alpha P + (1-\alpha)R \succ \alpha Q + (1-\alpha)S.$$

\medskip

Axiom Correlation Consistency relaxes Axiom Bi-independence as it focuses on the independence property of the correlation structure. Since $P$ and $Q$ share the same marginal lotteries in both sources, $P\succ Q$ means that the DM prefers the correlation structure of $P$ to that of $Q$. As she is indifferent between $R$ and $S$, if the mixture of $P,Q$ and $R,S$ does not ``infect'' the original correlation structures of $P$ and $Q$, then after mixture, the impacts of $R$ and $S$ will cancel out and the preference ranking between $P$ and $Q$ should remain unchanged. How do judge whether or not the correlation structure is ``infected'' after mixture? Notice that each lottery can be decomposed into a marginal lottery in source $1$ and a profile of conditional lotteries in source $2$. Hence one candidate measure of the correlation structure is the profile of conditional lotteries. That is why we need the additional qualification that $ supp(P_1)\cap supp(R_1)= supp(Q_1)\cap supp(S_1)=\emptyset$. It says that in source 1, the marginals  of two lotteries that are mixed have disjoint supports, which guarantees that the profile of conditional lotteries in source $2$ are not affected by the mixture. As a result, the preference over the correlation structures of $P$ and $Q$ persist after the mixture and hence the independence property holds. 

Moreover, notice that under Axiom Correlation Neglect, lotteries with the same marginals should always be indifferent and hence Axiom Correlation Consistency trivially holds. This implies that Axiom Correlation Consistency relaxes Axiom Bi-independence and Axiom Correlation Neglect. The next representation theorem generalizes \thmref{thm_cn} by simply replacing Axiom Correlation Neglect with Axiom Correlation Consistency.

\begin{theorem}\label{thm_BIB} Let $\succsim$ be a binary relation on $\mathcal{P}$. The following statements are equivalent:

i). The relation $\succsim$ satisfies Weak Order, Monotonicity, Weak Continuity, Weak Independence and    Correlation Consistency;

ii). The relation $\succsim$ admits one of the following representations: EU, BIB, EU-CN, GBIB-CN and GFIB-CN.

    Moreover, in all representations $H_1,H_2$ are unique, $v_1,v_2$ are unique up to a positive affine transformation and in EU, EU-CN, BIB, $w$ is unique up to a positive affine transformation.
\end{theorem}

By symmetry, an alternative measure of the correlation structure is to decompose each lottery into a marginal lottery in source $2$ and a profile of conditional lotteries in source $1$. Then the qualification naturally changes to disjoint supports of marginal lotteries in source $2$. This observation leads to the following axiom and corollary. 

\medskip

\noindent {\bf Axiom  Forward Correlation Consistency}: Suppose $P\succ Q$ with  $P_i=Q_i$, $i=1,2$, $R\sim S$ and $$ supp(P_2)\cap supp(R_2)= supp(Q_2)\cap supp(S_2)=\emptyset,$$ then for all $\alpha\in (0,1)$ 
$$\alpha P + (1-\alpha)R \succ \alpha Q + (1-\alpha)S.$$

\medskip

\begin{corollary}\label{coro_FIB}
Let $\succsim$ be a binary relation on $\mathcal{P}$. The following statements are equivalent:

i). The relation $\succsim$ satisfies Weak Order, Monotonicity, Weak Continuity, Weak Independence and  Forward  Correlation Consistency;

ii). The relation $\succsim$ admits one of the following representations: EU, FIB, EU-CN, GBIB-CN and GFIB-CN.

    Moreover, in all representations $H_1,H_2$ are unique, $v_1,v_2$ are unique up to a positive affine transformation and in EU, EU-CN, FIB, $w$ is unique up to a positive affine transformation.
\end{corollary}

Our representation results provide an axiomatic foundation for choice bracketing and correlation neglect. Across all representations we consider, we impose the implicit consistency condition that the preferences over riskless outcome profiles are the same. This provides a unified framework to compare different models, apply the same model across different economic settings and discover unexpected connections among distinct economic problems. We will provide some examples in Section \ref{section_app}.

\section{Proof Sketch}\label{section_sketch}
In this section, we briefly discuss the proof sketch of the two representation theorems in Section \ref{section_axioms}. We will focus on sufficiency of the axioms. It is worthwhile to note that \thmref{thm_cn} serves as an intermediate result in  our proof of  \thmref{thm_BIB}. As a result, although \thmref{thm_cn} seems like a corollary of \thmref{thm_BIB} , it needs to be proved first. 

For  \thmref{thm_cn}, by Axiom Correlation Neglect, it suffices to consider the preference over product lotteries $\hat{\cP}$. For any $q\in \lxone$ and $q'\in \lxtwo$, we denote the restriction of $\succsim$ on  $\mathcal{L}^0(X_1)\times \{q'\}$ as $\succsim_{1|q'}$ and the restriction of $\succsim$ on  $ \{q\} \times\mathcal{L}^0(X_2)$ as $\succsim_{2|q}$. We first show that $\succsim_{i|q}$ admits an EU representation for each $i\in \{1,2\}$ and $q\in \lxmi$. 

We define that the {\it independence property} holds for tuple $(P,Q,R,S)\in \hat{\cP}^4$ with $P_i=R_i$, $Q_j=S_j$ for some $i,j\in\{1,2\}$ and $P\succsim R, Q\succsim S$ if one of the following conditions hold:
\begin{itemize}
    \item $P\succ Q, R\sim S$ and for all $\alpha\in (0,1)$, $\alpha P + (1-\alpha) R\succ \alpha Q + (1-\alpha) S$;
    \item $P\sim Q, R\succ S$ and for all $\alpha\in (0,1)$, $\alpha P + (1-\alpha) R\succ \alpha Q + (1-\alpha) S$;
    \item $P\sim Q, R\sim S$ and for all $\alpha\in (0,1)$, $\alpha P + (1-\alpha) R\sim \alpha Q + (1-\alpha) S$;
    \item $P\succ Q, R\succ S$ and for all $\alpha\in (0,1)$, $\alpha P + (1-\alpha) R\succ \alpha Q + (1-\alpha) S$.
    \end{itemize}
We argue that it suffices to consider the case with $P\sim R$ and $Q\sim S$. Along with Axiom Weak Continuity,  Axiom Weak Multilinear Independence states that the independence property holds for any such tuple $(P,Q,R,S)$ with $P_{-i}\sim_{-i}R_{-i}$ and $Q_{-j}\sim_{-j}S_{-j}$. The rest of the proof proceeds as we discuss to what extent this ``local'' property can be generalized in different cases. 

If the DM narrowly brackets risks in both sources, then easy it is to show that the preference $\succsim$ admits an NB representation;

If the DM only narrowly brackets risks in one source, say, in source $2$, then it is sufficient to  focus on the subset of lotteries $\lxone \times X_2$ where the marginal lottery in source $2$ is degenerate. First, we notice that by assuming broad bracketing in source $1$, the independence property holds on a nontrivial set of product lotteries. Then we show that if the independence property holds on two sets of product lotteries respectively, then it also holds on their union. Finally, we apply the standard open cover arguments to extend the independence property and show that it would lead to a GBIB-CN representation. 

If the DM does not narrowly bracket risks in either source, then by a similar but more complex proof, we can show that $\succsim$ must admit an EU-CN representation. Actually, the interesting part is to exclude intermediate cases between EU-CN and GBIB-CN/GFIB-CN. 

The proof of \thmref{thm_BIB} can be decomposed into four steps. First, we restrict the preference $\succsim$ to product lotteries $\hat{\cP}$ and derive the corresponding partial representations on $\hat{\cP}$ by \thmref{thm_cn}. If further Axiom Correlation Neglect holds, then we are done. From now on, suppose that this axiom fails.  Second, we show that  Axiom Correlation Sensitivity can be strengthened to a natural relaxation of Axiom Independence, where the marginals of lotteries that are mixed have disjoint supports in source $1$. Third, by embedding the set of lotteries as a subspace of temporal lotteries in \cite{KP1978}, we can extend $\succsim$ to the set of temporal lotteries while satisfying the axioms in  \cite{KP1978}. Hence, $\succsim$ admits a KP-style representation on $\cP$. Finally, by making use of consistency of the two representations in the previous steps on product lotteries $\hat{\cP}$, we conclude that only representations stated in \thmref{thm_BIB} are feasible. 

\section{Applications and Discussions}\label{section_app}

\subsection{Simultaneous Monetary Prizes}\label{section_money}

In most economic applications like portfolio choices and labor supply decisions, the outcome in both sources is money or the numeraire. Also, the payoffs in experiments typically takes the form of tokens, which can be exchanged to money at a fixed rate. We assume that $X_1= X_2 =\bR$ in this section. 

When there is no risk and the DM receives money form both sources simultaneously, we argue that she will evaluate each outcome profile by adding up the monetary prizes. Consider a trivial example where a worker can choose between two payment schemes after finishing two identical tasks. In scheme 1, she will receive \$200 from the first task and \$220 from the second one. In scheme 2, she will get \$210 from the first task and \$200 from the second one. All payments are made at the same time by cash and tasks have already been done. Then arguably the worker should choose scheme 1, from which she can get \$10 more. The idea is summarized in the following axiom. 

\medskip

\noindent  {\bf Axiom Broad Bracketing without Risk}: For each $x,y\in\bR$, $(x,y)\sim (x+y,0)\sim (0,x+y)$.

\medskip

For any of the previous representations, this additional axiom just requires that the utility over degenerate lotteries $w:\mathbb{R}^2\rightarrow\mathbb{R}$ can be replaced by $u:\mathbb{R}\rightarrow\mathbb{R}$ such that $w(x,y)=u(x+y)$ for any $x,y\geq 0$. Suppose we further assume symmetry of the narrow preferences, that is, 

\medskip

\noindent  {\bf Axiom Symmetry}: For each $p,q\in \lr$, $(0,p)\succsim (0,q)$ if and only if $(p,0)\succsim (q,0)$.

\medskip
Then NB can be rewritten as 
	$$V^{NB}(P)= u( CE_{v}(P_1) +  CE_{v}(P_2)), \forall~P\in\cP.$$
where $u$ and $v$ are regular functions. Notice that $u$ is strictly monotone, the preference with a NB representation is also represented by 
	$$\hat{V}^{NB}(P)=  CE_{v}(P_1) +  CE_{v}(P_2), \forall~P\in\cP.$$
In other words, a narrow bracketer  evaluates a monetary lottery by the summing up the certainty equivalents of its marginal lotteries.

We contrast our representation with the commonly used functional form in the literature of narrow bracketing, which focuses on product lotteries. We adapt the utility function in \cite*{aer2006framing} and \cite{aer2009narrow} to our framework as follows: 
$$U(p,q)= \lambda \sum_{x,y}u(x+y)p(x)q(y) + (1-\lambda)[\sum_x u(x)p(x) +\sum_y u(y)q(y)],\forall~(p,q)\in\hat{\cP}, $$
where $\lambda \in [0,1]$ and $1-\lambda$ determines the degree of narrow bracketing. At the end of Section \ref{section_NB}, we already argued why such $\lambda$-mixture  models are excluded in our framework. Here we focus on the case with $\lambda=0$, that is, the DM admits (fully) narrow bracketing. Instead of summing up the certainty equivalents of marginal lotteries like in our NB, the above functional form sums up the expected utility of  marginal lotteries.\footnote{If we do not assume Axiom Broad Bracketing without Risk, then $U(p,q)=\sum_x u(x)p(x) +\sum_y u(y)q(y)$ is actually a special case of the intersection of NB and EU. This is the expected discounted utility model in  in time preference with discount factor $1$ (see Section \ref{section_unified}). However, as is argued above, Axiom Broad Bracketing without Risk is more reasonable in the setting with simultaneous monetary prizes. This again reflects our point that if we want to compare models of choice bracketing under risk, we need to maintain the same preference over degenerate lotteries.} This implies that the DM is subject to narrow bracketing even when there is no risk and hence she might prefer $(x,y)$ over $(x',y')$ even if $x+y< x'+y'$. 

One should notice that  \cite*{aer2006framing} and \cite{aer2009narrow} focus on the case where the choice problems in different sources are ``independent''. That is, besides assuming risks in two sources are resolved independently, they also assume that  the availability of gambles in one decision problem does not depend on available gambles in the other one. For instance, if $(x,y)$ and $(x',y')$ are available options, then $(x,y')$ and $(x',y)$ should also be available. In this restricted choice setup, their narrow bracketing representation will predict the same behavior as our NB model. However, if we consider the more general choice domain, the two models differ and theirs will predict narrow bracketing even without risk.

\subsubsection{Experimental Evidence of Narrow Bracketing}

Now we show how our model can accommodate the experimental evidence of narrow bracketing. Consider the following classic experiment introduced by \cite{S1981framing} and developed by \cite{aer2009narrow}. 

 \begin{exam}\label{ex0}
Suppose you face the following pair of concurrent decisions. All lotteries are independent. First examine both decisions, then indicate your choices. Both choices will be payoff relevant, i.e., the gains and losses will be added to your overall payment.

	{\it Decision (1): Choose between:}
	
	~~~~~{\it A. A sure gain of \$2.40.}
	
	~~~~~{\it B. A 25 percent chance to gain \$10.00, and a 75 percent chance to gain \$0.00.}

	{\it Decision (2): Choose between:}
	
	~~~~~{\it C. A sure loss of \$7.50.}
	
	~~~~~{\it D. A 75 percent chance to lose \$10.00, and a 25 percent chance to lose \$0.00.}
 \end{exam}

Since gains and losses from the two decision problems are aggregated, the DM should focus only on the distribution of overall monetary prizes. For example, the combination of $B$ and $C$ produces a lottery with a 1/4 chance of gaining \$2.50 and a 3/4 chance of losing \$7.50. By comparison, the lottery induced by the combination of $A$ and $D$ is a 1/4 chance of gaining \$2.40 and a 3/4 chance of losing \$7.60. The $BC$ combination is equal to the $AD$ combination plus a sure gain of  \$0.10 and hence the $AD$ combination is first-order stochastically dominated.  However, across different treatments in \cite{S1981framing} and \cite{aer2009narrow}, a reasonably large fraction (above 28\%) of subjects chose $A$ in decision (1) and $D$ in decision (2). Notice that $BC$ dominates $AD$ by adding a sure gain and they are reasonably similar in terms of complexity, hence previous models that are monotone or incorporate complexity aversion cannot explain the common choices of dominated options  without choice bracketing.

 Suppose the DM is narrowly bracketing with the following representation over product lotteries $\hat{\cP}$:
$$\hat{V}^{NB}(p,q)= CE_{v}(p) +  CE_{v}(q).$$
	where the utility index $u$ satisfies 
	\begin{equation}\label{eq_loss_averse}
	    v(x)= \begin{cases}
\sqrt{x}, \hbox{~if~} x\geq 0,\\
-2\sqrt{-x}, \hbox{~if~}  x<0.\end{cases}
	\end{equation}

This is a standard reference-dependent model with the reference point fixed at $0$. Easy to show that this model can accommodate the choice of A and D over B and C as
$$CE_{v}(p_A)=2.4> 0.625= CE_{v}(p_B), CE_{v}(p_D)=-5.625> -7.5= CE_{v}(p_C).$$

\subsubsection{Background Risk}\label{section_critique}

In this section, we interpret the risk in source $1$ as the background or endowment risk, and the risk in source $2$ as the risky decision at hand. \cite{ecta2000calibration} formally identifies a tension between expected utility and risk aversion regarding small gambles when choices only depend on the (distribution of) final wealth. His calibration theorem shows that a  low level of risk aversion with respect to small gambles leads to an absurdly high level of risk aversion with respect to large gambles. \cite{ecta2008calibration} then extend Rabin's calibration results to non-expected utility models satisfying certain differentiability conditions. \cite*{mpst2020background} further suggest that ``theories that do not account for narrow framing--whereby
independent sources of risk are evaluated separately by the decision maker--cannot explain
commonly observed choices among risky alternatives.''

One prominent thought experiment of Rabin's critique is as follows: if an EU maximizer turns down 50-50 gambles of losing \$1000 or gaining \$1050 for all initial wealth levels, then she would always turn down 50-50 gambles of losing \$20,000 or gaining any sum.

To avoid the unrealistic behavior by an EU maximizer, we consider a DM who has difficulty integrating risks in difference sources and will use narrow bracketing as a simplifying heuristic. Intuitively, outcomes in the gamble at hand are more ``accessible'' than background wealth levels \citep{k2011thinking} and the DM might not take into account the background risk when deciding whether to accept the gamble at hand. For example, suppose that the DM's NB utility function is given by Equation (\ref{eq_loss_averse}). One can easily show that the DM will always reject 50-50 of losing  \$1000 or gaining \$1050, and accept 50-50 gamble losing \$20,000 and gaining \$80,050. \footnote{An alternative way is to assume the DM does not fully ignore the background risk. Instead, she might only consider the certainty equivalent of the background risk and admits a FIB-CN representation. However, in order to accommodate the Rabin's critique, we need to extend the current model to allow for non-expected utility representations with first-order risk aversion over marginal lotteries like in \cite{ecta1991DA}. \label{footnote_FORA}}

\subsection{Time and Risk Preferences}\label{section_time}

From now on, we interpret the outcomes in two sources as consumptions or monetary prizes in two different periods. Concretely, source 1 is labeled as period 1 or present, and source 2 is labeled as period 2 or future. Sometimes we name an outcome profile as a consumption profile or consumption path. 

The rest of this section consists of two parts. The first part connects our representations with different seemingly distinct time preference models in the literature and provides a unified framework for them. It is worthwhile to emphasize that our characterization theorem originates from simplifying heuristics for  multi-source risks, without any ex-ante normative assumptions for intertemporal choices. In the second part, we study implications of some commonly studied axioms in time preferences and propose a new model that can accommodate many desirable properties, including separation of time and risk preferences, indifference to temporal resolution of uncertainty and stationarity.

\subsubsection{A Unified Framework}\label{section_unified}

In this section, we show that EU, BIB and NB have nice counterparts in time preferences that have been well studied in the literature. 

First, most commonly used time preferences are special cases of EU.  The most prominent example is the Expected Discounted Utility (EDU) model:
		$$V^{EDU}(P)=\mathbb{E}_{P_1}[u(x)]+\beta \mathbb{E}_{P_2}[u(y)].$$
where $u$ is the EU index in each period and $\beta\in [0,1]$ is the discount factor. The DM evaluates each lottery by the summation of expected utility of each marginal lottery weighted by the discount factor. One can easily show that this functional form is also a special case of NB. 

One natural extension of EDU is the Kihlstrom-Mirman (KM) model (\citealp*{KM1974}, \citealp*{DGO2020SI}) given by 
		$$V^{KM}(P)=\mathbb{E}_P\big[\phi\big( \frac{1}{\sum_t D(t)} \sum_{t=1}^2 D(t)u(x_t) \big)\big].$$
where the DM first evaluates each consumption path $(x_1,x_2)$ using discounted utility and then takes expected value of the utility profiles by applying additional curvature $\phi$. 

Second, Selden (1978) and Selden and Stux (1978) introduce a alternative time preference model to EU called Dynamic Ordinal Certainty Equivalent (DOCE), where the DM first takes the certainty equivalents of marginal lotteries in each period and then evaluate the profile of certainty equivalents. This exactly agrees the time preference interpretation of our NB model. 
 $$V^{DOCE}(P)=V^{NB}(P)  = w(CE_{v_1}(P_1),CE_{v_2}(P_2)).$$
One nice feature of DOCE/NB is that it can accommodate the tension between stochastic impatience and risk aversion over time lotteries introduced by \cite*{ecta2020time}, which is violated by most existing models of time preferences including  Epstein-Zin preferences \citep{ecta1989EZ} and risk-sensitive preferences \citep{ieee1995HS} as is shown in \cite*{DGO2020SI}.

Finally, we connect BIB with the the preferences in \cite{KP1978} and show that they have similar functional forms involving backward induction and offer similar predictions, despite their distinct behavioral motivations and characterizations.

 \cite{KP1978} extend lotteries over consumption paths to {\it temporal lotteries} in order to model temporal resolution of uncertainty. In the two-period setup, the set of temporal lotteries is $\mathcal{D}^*:=\mathcal{L}^0(X_1\times \mathcal{L}^0(X_2))$. To see why the set of temporal lotteries $\mathcal{D}^*$ is strictly larger than the set of lotteries $\mathcal{P}$, take the lottery $P=1/2(\delta_1,\delta_1) + 1/2(\delta_1,\delta_2)\in \mathcal{P}$ for an instance. There are two temporal lotteries that correspond to $P$: $d_1=1/2\delta_{(1,\delta_1)} + 1/2\delta_{(1,\delta_2)}$ and $d_2=\delta_{(1,1/2\delta_1+1/2\delta_2)}$. Notice that the marginal lottery of $P$ in the first period is deterministic and $P$ only involves risk in the second period. It remains unspecified about the timing of resolution of such risk. By comparison, temporal lotteries $d_1$ and $d_2$ contain exactly the same uncertainty in the second period as $P$, which resolves in the first period for temporal lottery $d_1$ and in the second period for $d_2$. In other words, there are two dated types of mixtures for deterministic consumption paths $(\delta_1,\delta_1)$ and $(\delta_1,\delta_2)$ in temporal lotteries, but only one type of mixture in lotteries. Generically, the DM can have strict ranking between temporal lotteries $d_1$ and $d_2$, which reflects her preference for early or late resolution of uncertainty. We define that a lottery $P\in \mathcal{P}$ is {\it induced} by the temporal lottery $d\in \mathcal{D}^*$ if for any $(x,y)\in X$, 
 $$P(x,y)=\sum_{(x,q)}d(x,q)q(y)$$
Following the above example of $P$ and $d_1$, $d_2$, we can show that every lottery is induced by some temporal lottery and there exist lotteries induced by more than one temporal lotteries. Thus the domain of temporal lotteries $\mathcal{D}^*$ is strictly richer than the set of lotteries $\mathcal{P}$.
 
 By applying the vNM axioms to temporal lotteries with mixture in period 1 and to temporal lotteries whose uncertainty only resolves in period 2 with mixture in period 2, \cite{KP1978} axiomatize a general class of Kreps-Porteus preferences (henceforth KP). To get proper comparison with BIB, we focus on the history-independent KP, whose representation $V^{KP}: \mathcal{D}^*\rightarrow \mathbb{R}$ is characterized by a tuple of regular functions $(w,v_2)$ such that $w: X\rightarrow \mathbb{R}, v_2: X_2\rightarrow \mathbb{R}$, and 
	$$V^{KP}(d)  = \sum_{(x,p)} w(x,CE_{v_{2}}(p))d(x,p).$$
By comparison, the BIB representation with the same tuple $(w,v_2)$ is given by 
$$V^{BIB}(P)  = \sum_{x} w(x,CE_{v_{2}}(P_{2|x}))P_1(x).$$

The above two representations share the same backward inductive procedure to evaluate multi-period risk. The DM first reduces the risk (resolving) in period 2 into its certainty equivalent under some history-independent expected utility index $v_2$. This transforms the original (temporal) lottery into one with only uncertainty (resolving) in period 1. Then the DM evaluates the new (temporal) lottery based on its expected utility under some index $w$. This recursive structure allows for the adoption of dynamic programming methods in optimization problems, and partially explains the popularity of the \cite{KP1978} framework in the past decades. Most recursive models, including the famous Epstein-Zin preferences \citep{ecta1989EZ} (henceforth EZ) and risk-sensitive preferences \citep{ieee1995HS} (henceforth HS)\footnote{\cite{ieee1995HS} originally formulate the risk-sensitive preference as an optimal control problem with risk-adjusted costs. \cite*{ecta2017recursive} show how it can be interpreted as a monotone recursive preference over temporal lotteries.}, are generalizations of \cite{KP1978} to temporal lotteries in an infinite horizon setting. 

It is worthwhile to mention three differences between BIB and history-independent KP preferences. First, the history-independent KP representations are defined on a strictly richer domain than BIB, which allows for resolution of uncertainty in different periods. $V^{KP}$ reduces to $V^{BIB}$ on the subdomain of temporal lotteries where uncertainty about outcomes in period $t$ resolves in period $t$ for each $t$. Second, as a direct implication of different domains, BIB exhibits indifference to temporal resolution of risk, while the history-independent KP preference satisfies this property if and only if it reduces to an expected utility representation. This distinction is essential in our discussion about asset market puzzles in Section \ref{section_puzzle}. Finally, the history-independent KP preferences satisfy the vNM independence axiom over temporal lotteries, while BIB violates it on the domain of lotteries. Actually, BIB satisfies Axiom Independence if and only if it also admits an EU representation\footnote{Another difference is that the history-independent KP preferences can satisfy topological continuity on its domain, which is not true for the BIB preferences. However, we can show that BIB satisfies Axiom Independence if and only if it satisfies topological continuity. In other words, backward induction bracketing can be interpreted as a joint relaxation of the independence and continuity properties.}. 

This suggests that shared recursive procedure in the two models is based on different rationales. In the history-independent KP preferences, it comes from non-indifference to temporal resolution of uncertainty and the model remains consistent with the expected utility paradigm. In the BIB representations, it results from choice bracketing, which is a simplifying heuristic to evaluate multi-source risk and can be behaviorally characterized via a deviation from the expected utility paradigm. They serve as two distinct and complementary justifications for the backward inductive procedure and neither of them is universally superior or inferior to the other.   One can actually enrich the domain of our framework to develop a model with choice bracketing and temporal resolution of uncertainty simultaneously. However, for some applications in Section \ref{section_puzzle}, we will argue that our framework based on choice bracketing might be more suitable.

Recall that our \thmref{thm_BIB} characterizes EU, BIB and NB among other representations by relaxing the vNM independence axiom. We provide a unified framework for those seemingly distinct or even competing models in the literature, based solely on simplifying heuristics for multi-source risks and no normative time preference properties. 
This reveals a deep connection between choice bracketing, which has usually been considered as an exotic behavioral bias or errors, and the commonly accepted models of time preferences.

\subsubsection{A New Model: KM-BIB}\label{section_KMBIB}

In this section, we will discuss some desirable normative properties of time preferences in the literature and then propose a new class of models based on BIB that can simultaneously satisfy all these properties except for one. For simplicity, assume that $X_1=X_2=\bRp$.

We start with the separation of time and risk preferences. It is well-known that in any EDU model, the inverse of the elasticity of intertemporal substitution (EIS) coincides with the coefficient of relative risk aversion (RRA). That is, the time preference and the risk preference and intertwined together. However, enormous empirical evidence in macroeconomics, finance and behavioral economics has shown the necessity to separate the two coefficients\footnote{Indirect evidence includes the failure to explain the equity premium puzzle with EDU \citep{jme1985puzzle}. See more discussion in Section \ref{section_puzzle}. For direct evidence, \cite*{qje1997evidence} find that RRA and EIS are uncorrelated through a cross section of American households and \cite{aer2012evidence} show similar results in a lab experiment.}. Actually, one motivation of EZ and DOCE is to achieve such separation of time and risk preferences. Suppose that $v_1$ and $v_2$ measure the risk preferences with in each period and $w$ measures the time preference over deterministic consumption paths. By separation of time and risk preferences, we mean that for each $w$, we can represent arbitrary risk preferences in each period by choosing appropriate $v_1$ and $v_2$, and vice versa.

The second property is indifference to temporal resolution of uncertainty, which is implicitly assumed by our choice domain of lotteries. Although EZ can have separate parameters for time and risk preferences, the separation depends on specific preferences for temporal resolution of uncertainty. For instance, if the DM is indifferent to when the uncertainty is resolved, then the two parameters must be the same and EZ agrees with EDU. Hence, despite the fact that many people might prefer  early or late resolution of uncertainty, it is still worthwhile to separate such preferences with other desirable properties in time preferences by assuming indifference to temporal resolution of uncertainty.

The third property is stationarity. The following two axioms illustrate the idea of \cite{ecta1960stationary} that ``the passage of time does not have an effect on preferences''. 

\medskip

  {\bf Axiom History Independence}: For any $x,y\geq 0$ and $p,q\in \lrp$, $(x,p)\succsim (x,q)$ if and only if $(y,p)\succsim (y,q)$.
  
  \medskip

  {\bf Axiom Stationarity}: For any  $p,q\in \lrp$, $p\succsim_2 q$ if and only if $p\succsim_1 q$.
  
  \medskip

Axiom History Independence states that, from the ex-ante perspective, the choice in period 2 in independent of the history of consumption in period 1. Axiom Stationarity is an adaption of Axiom 5 in \cite*{ecta2017recursive}, which is defined on temporal lotteries with infinite horizons, to our framework. It reflects the time-invariance of the DM's risk preference. Following \cite{ecta1960stationary} and \cite*{ecta2017recursive}, when there is no confusion, we use the term ``stationarity'' to represent the conjunction of both axioms.

The fourth property is recursivity, which requires that the mere passage of time does not affect the DM's preferences.\footnote{We use recursivity here since we adopt an ex-ante approach. Recursivity is essentially identical to the notion of time consistency or dynamic consistency in \cite{ecta1985tc} when we study the DM's preferences in every period and assume time invariance. }  It connects ex-ante and ex-post choices and permits the use of backward induction and dynamic programming methods.  We adapt the notion of  recursivity defined on the domain of temporal lotteries in \cite{CE1991recursive} and \cite{ecta2017recursive} to the domain of lotteries. 

 {\bf Axiom -- Recursivity}: For all $\pi_i\in (0,1)$, $x_i,x'_i\geq 0$ $q_i,q'_i \in \lrp$, $i=1,...,n$, such that $\sum_{i=1}^n\pi_i=1$, for each $i\neq j$, $x_i=x_j$ implies $q_i=q_j$ and $x'_i=x'_j$ implies $q'_i=q'_j$, if $(x_i,q_i)\succsim (x'_i,q'_i)$ for each $i$, then $\sum_i{\pi_i}(x_i,q_i)\succsim \sum_i{\pi_i}(x'_i,q'_i)$. Moreover, the latter preference is strict if $(x_i,q_i)\succ (x'_i,q'_i)$ for some $i$.

Axiom Recursivity states that if the DM prefers one lottery to another after the realization of consumption in period 1, then her preference over the two lotteries should be the same ex ante. It is easy to show that BIB satisfies recursivity.

The fifth property is called discounted utility without risk, which states that the preference over deterministic consumption paths agrees with EDU and can be represented by the summation of discounted utility in each period. 

  {\bf Assumption 1 -- Discounted Utility without Risk}: There exists a regular function $u:\mathbb{R}_+\rightarrow\mathbb{R}$ and $\beta\in [0,1]$ such that for all $x_1,x_2,y_1,y_2\geq 0$, $({x_1},{x_2})\succsim ({y_1},{y_2})$ if and only if $u(x_1)+\beta u(x_2)\geq u(y_1)+\beta u(y_2)$.
  
  \medskip

\cite*{DGO2020SI} show that KM is exactly the class of EU that admits a discounted utility representation when there is no risk. Similarly, a BIB representation $V^{BIB}$ satisfies Discounted Utility without Risk if and only if 
$$V^{BIB}(P)=\sum_{x}\phi\big(u(x)+\beta u(CE_{v_{2}}(P_{2|x}))\big)P_1(x)$$
where $\phi,u,v_2:\mathbb{R}_+\rightarrow\mathbb{R}$ are regular and $\beta\in [0,1]$ is the discount factor. If we further assume stationarity, then we derive the following KM-BIB representation: 
\begin{equation}\label{eq_KM-BIB}
    		V^{KM-BIB}(P)=\sum_{x}\phi\big(u(x)+\beta u(CE_{\phi\circ u}(P_{2|x}))\big)P_1(x).
\end{equation}
where $\phi,v_2:\mathbb{R}_+\rightarrow\mathbb{R}$ are regular and $\beta\in [0,1]$.

 By definition of the domain and the functional form, KM-BIB exhibits indifference to temporal resolution of uncertainty and history independence. Also, KM-BIB achieves a separation between time and risk preferences. The time preference  is determined by $u$. The risk preference in each period is represented by EU index $\phi\circ u$. This means that $\phi$ is the additional curvature used only in the case of risk preference and determines the separation between risk aversion and intertermporal substitution. We summarize those insights in the following claim.
 
 \begin{claim}\label{claim_KM-BIB}
 If the preference $\succsim$ admits a KM-BIB representation in (\ref{eq_KM-BIB}), then it satisfies discounted utility without risk, separation of time and risk preferences, indifference to temporal resolution of uncertainty, recursivity and stationarity. 
 \end{claim}

Now we introduce two notions of ``risk aversion'' across different periods for KM-BIB as an application. The first is correlation aversion introduced by \cite{EB2007correlation}. 

  {\bf Axiom -- Correlation Aversion}: For any $x_2>x_1$ and $y_2>y_1$, 
		$${\frac{1}{2}}(\delta_{x_1},\delta_{y_2})+{\frac{1}{2}}(\delta_{x_2},\delta_{y_1})  \succsim {\frac{1}{2}}(\delta_{x_1},\delta_{y_1})+{\frac{1}{2}}(\delta_{x_2},\delta_{y_2})$$

Notice the two lotteries agree on the marginal lotteries in both periods and only differ in the correlation structure. By monotonicity, the DM's most preferred outcome path is $(\delta_{x_2}, \delta_{y_2})$ and her least preferred outcome path is $(\delta_{x_1}, \delta_{y_1})$. Axiom Correlation Aversion requires that she prefers the mixture between intermediate outcome paths to the mixture between extreme outcome paths. We can show that a stationary KM-BIB satisfies Axiom Correlation Aversion if and only if the additional curvature $\phi$ is concave. \cite*{DGO2020SI} introduce a similar notion called residual risk aversion, which is also captured by the concavity of $\phi$. This suggests that  correlation aversion coincides with residual risk aversion in our framework.

Another relevant notion is called long-run risk aversion. Consider the following two consumption plans. In the first one, a coin is flipped independently in each period, and the payoff is \$1 if it lands on heads and \$0 if it lands on tails. This scenario is referred to as short-run risk. In the second plan, a coin is flipped only once at the beginning of period 1, and the payoff is either \$1 or \$0 in both periods. This scenario is referred to as long-run risk. EDU exhibits indifference between long-run risk and short-run risk, while the sensitivity to  long-run risk has been used to explain many financial puzzles \citep{jf2004LLR}. We adapt (and simplify) the notion of long-run risk aversion in \cite{ecta2013temporal} as follows.

{\bf Axiom -- Long-run Risk Aversion}: For any $x_2>x_1$, 
		$$ (\frac{1}{2}\delta_{x_1}+\frac{1}{2}\delta_{x_2}, \frac{1}{2}\delta_{x_1}+\frac{1}{2}\delta_{x_2} )  \succsim \frac{1}{2}(\delta_{x_1},\delta_{x_1})+\frac{1}{2}(\delta_{x_2},\delta_{x_2})$$

Easy to see that Axiom Long-run Risk Aversion is implied by Axiom Correlation Aversion for EU models, which include KM. However, for KM-BIB, the two notions differ and long-run risk attitude might depend on higher-order curvature of $\phi$. 

We end this section with a property that is generically violated by KM-BIB and relates it to the experimental evidence on choice bracketing in \cite{aer2009narrow}. Consider the following modification of the ordinal dominance property in \cite{jet1990non-eu} and the monotonicity condition in \cite*{ecta2017recursive}.

\medskip

  {\bf Axiom -- Ordinal Dominance}: For each $n>0$,  $(x_1^i,x_2^i),(y_1^i,y_2^i)\in\mathbb{R}^2_+$ for $i=1,...,n$ and $(\pi_1,...,\pi_n)\in[0,1]^n$ with $\sum_{i}\pi_i=1$, if either $x_1^i=x_1^j, y_1^i=y_1^j$ for all $i,j=1,...,n$, or $x_1^i\neq x_1^j, y_1^i\neq y_1^j$ for all $i\neq j$,  then
  $$(\delta_{x_1^i},\delta_{x_2^i})\succsim (\delta_{y_1^i},\delta_{y_2^i}), ~\forall~i=1,...,n\Longrightarrow \sum_{i=1}^n\pi_i(\delta_{x_1^i},\delta_{x_2^i})\succsim \sum_{i=1}^n \pi_i(\delta_{y_1^i},\delta_{y_2^i}).$$
  
    \medskip
    
Intuitively, ordinal dominance requires that the DM would never choose an action if another available action is preferable in every state of the world. Notice that the original ordinal dominance  property is defined over temporal lotteries and it holds for mixture of temporal lotteries in any period $t$ so long as they have the same deterministic history of consumptions before period $t$. In order to maintain similar interpretations in the space of lotteries, we require that ordinal dominance holds either when there is only uncertainty in the first period (i.e., the case where $x_1^i\neq x_1^j, y_1^i\neq y_1^j$ for all $i\neq j$) or in the second period (i.e., the case where $x_1^i=x_1^j, y_1^i=y_1^j$ for all $i,j$).

\cite*{ecta2017recursive} show that ordinal dominance is a tight restriction for recursive preferences. Specifically, the recursive KP preference satisfies ordinal dominance if and only if it is either  a risk-sensitive (HS) preference, where the risk attitude exhibits constant absolute risk aversion,  or belongs to the class of \cite{U1968time}, which is a special case of expected utility with infinite horizon. Similar results also hold in our framework. Axiom Ordinal Dominance is only generically satisfied by EU and it can be satisfied by other representations only if the DM has constant absolute risk aversion.

Suppose that instead of intertemporal choices, we interpret marginal lotteries in two sources as simultaneous monetary gambles and assume that the DM broadly brackets degenerate marginal lotteries. Then Axiom Ordinal Dominance exactly reduces to ``first-order stochastic dominance'' considered in \cite{aer2009narrow}. This implies a deep connection between the empirical evidence of narrow bracketing in experiments (\citealp{aer2009narrow}, \citealp{EF2020revealNB}) and the theoretical difficulty to satisfy ordinal dominance in recursive preferences \citep*{ecta2017recursive}.

\subsection{Asset Market Puzzles}\label{section_puzzle}

Since \cite{jme1985puzzle} introduced the equity premium puzzle, many puzzling facts of asset markets have been observed and challenged the validity of various models, including the standard expected discounted utility (EDU) model\footnote{Actually, the inflexibility of EDU to explain the  equity premium puzzle is one major motivation of the literature on recursive preferences. See \cite{ecta1989EZ} for a detailed discussion. }. It has been well understood that those puzzles are quantitative and explanations with extreme parameter values are usually regarded as inadequate.

One popular approach to address asset pricing puzzles is to use recursive preferences that permit the separation of time and risk preferences \citep{jpe1991EZ}. For instance, the long-run risks model of \cite{jf2004LLR} has provided a unified rationalization of several puzzling facts in asset markets  by combining the EZ preference with an endowment process featuring a persistent predictable component for consumption growth and its volatility. However, \cite*{aer2014time_premium} point out that the quantitative assessment of the preference for early resolution of uncertainty has been ignored in the macro-finance literature. The authors show that the parameter values used in \cite{jf2004LLR} imply that the DM is willing to give  up 25 or 30 percent of her lifetime consumption in order to have all risks resolved in period 1. Such timing premium is arguably too high as the risk is about consumption instead of income or asset returns, and there is no apparent instrumental value of information by early resolution of uncertainty. 

Note that in most applications of EZ in macroeconomics and finance,  temporal resolution of uncertainty is not explicitly involved. Instead, there is an implicit assumption that uncertainty about consumption and other state variables in period $t$  resolves in period $t$. Hence, if the main goal is to achieve the separation of time and risk preferences, then  our framework with narrow bracketing might be suitable since it only involves lotteries over deterministic consumption paths. As a result, indifference to temporal resolution of uncertainty  automatically holds and the timing premium is always zero.

We are not the first to include narrow bracketing to explain financial puzzles. \cite{qje1995myopic} provide an explanation of the equity premium puzzle by combining loss aversion and narrow bracketing. They argue that investors dislike stocks because they look at their portfolios frequently and evaluate the nominal changes in their accounts with loss aversion, even though they might save for a distant future. However, their approach differs from ours significantly. \cite{qje1995myopic} hinges on loss aversion and the behavioral assumption that agents care about nominal changes in the accounts. Also, it is not clear how their approach can be applied to  other puzzles studied in \cite{jf2004LLR}. By comparison, our approach is to adopt narrow bracketing to provide a foundation for an EZ-style preference which simultaneously satisfies separation of time and risk preferences and indifference to  temporal resolution of uncertainty. Note we are not claiming that zero timing premium is  normatively or positive appealing. Instead, our main message is that the attitude towards temporal resolution of uncertainty can be isolated from the separation of time and risk preferences. Actually, one can extend our model to the space of temporal lotteries and allow for the preference for either early or late resolution of uncertainty.

In the rest of the section, we will propose an alternative to the CRRA-CES EZ model used in \cite{jf2004LLR}. First,  consider the following  special case of  KM-BIB in two periods by assuming $u(x)=x^{\rho}$ and $\phi(x)=x^{\alpha/\rho}$ with $\rho<1, 0\neq \alpha<1$ and $0<\beta<1$: 
	$$U^{KM-BIB}(P)= \sum_{c_1} \big[(1-\beta)c_1^{\rho} + \beta[\mathbb{E}_{P_{2|c_1}}(c_2^{\alpha})]^{\rho/\alpha}\big]^{\alpha/\rho}P_1(c_1).$$
Like CRRA-CES EZ, the time preference parameter EIS is $1/(1-\rho)$ and the risk preference parameter RRA is $1-\alpha$, which reveals a separation of the two preferences.

Then we briefly discuss how to extend the above model to one with multiple periods. For simplicity, assume that the consumption space in each period $t=1,...,T$ is a compact interval $C$, where $T$ can be $+\infty$. The set of deterministic consumption paths is $C^{T}$ with a generic element $\mathbf{c}=(c_t)_{t=1}^T$. For each consumption path $\mathbf{c}\in C^T$, we denote the subsequence of consumptions in the first $t$ periods as $\mathbf{c}^t=(c_{\tau})_{\tau=1}^t$. 

The preference is defined on the lottery space $\mathcal{P}=\mathcal{L}(C^T)$. Here we allow for lotteries with infinite supports to accommodate applications in finance. For each lottery $P$, denote $P_{[t]}$ as the marginal lottery in the first $t$ periods, $1\leq t< T$. For each subsequence of consumptions $\mathbf{c}^t$ in the support of $P_{[t]}$, we define $\phi(P|\mathbf{c}^t)$ as the conditional lottery starting from period $t+1$, given that consumptions in the first $t$ periods are $\mathbf{c}^t$. When $T<+\infty$, $\phi(P|\mathbf{c}^t)\in \mathcal{L}(C^{T-t})$ and when $T= +\infty$, $\phi(P|\mathbf{c}^t)\in  \mathcal{L}(C^{\infty})$. 
Note that for each finite $T$, $\mathcal{L}(C^{T-t})$ is homeomorphic to a subset of $ \mathcal{L}(C^{\infty})$ where the consumptions are always $0$ from period $t+1$ on. So we will focus on the case with an infinite horizon.

 The following notions are adapted from recursive preferences on temporal lotteries (\citealp{CE1991recursive}, \citealp{ecta2017recursive}) to our framework. For each $V:\mathcal{P}:=\mathcal{L}(C^{\infty})\rightarrow \mathbb{R}$ and $p\in \mathcal{P}$, denote 
 $$m_V(P)(B)\equiv P_1\big\{c\in C:V(c,\phi(P|c)\in B) \big\}, \forall B\in \mathcal{B}(V(\mathcal{P}))$$
where $V(\mathcal{P})\subset \mathbb{R}$ is the image of $V$ on $\mathcal{P}$ and $\mathcal{B}(V(\mathcal{P}))$ is the set of all Borel subsets of $V(\mathcal{P})$. Then $m_V(P)$ is a probability measure over utilities conditional on the current consumption. Now we define the \textit{recursive preference over lotteries} as $V:\mathcal{P}\rightarrow \mathbb{R}$ with
\begin{align*}
    V(P) &= I(m_V(P)),\\
    V(c,q)&= W(c,V(q)),
\end{align*}
where $m_V(P)$ is defined as above, $I:\mathcal{L}(\mathbb{R})\rightarrow \mathbb{R}$ is a \textit{certainty equivalent}, that is, $I$ is continuous, increasing with respect to first order stochastic dominance and $I(x)=x$ for each $x\in \mathbb{R}$, $W:C\times \mathbb{R}\rightarrow \mathbb{R}$ is continuous and strictly increasing in the second argument. It is worthwhile to mention that, unlike \cite{CE1991recursive} and  \cite{ecta2017recursive}, $V$ is generically discontinuous since the mapping $m_V$ is  discontinuous. This is similar to the discontinuity of BIB in Section \ref{section_axioms}.

In order to get the CRRA-CES KM-BIB model, we can set $I=\phi^{-1}\circ \mathbb{E} \circ\phi$ with $\phi(x)=x^{\alpha\slash\rho}$ and $W(c,v)=(1-\beta)c^{\rho} + \beta v$, where $\rho<1, 0\neq \alpha<1$ and $0<\beta<1$. The recursive preference is equivalent to the following recursion of value functions (up to a monotone transformation): 
\begin{equation}\label{eq_recursion1}
    U_t^{\rho} = (1-\beta)c_t^{\rho} + \beta \big[\mathbb{E}_{\phi(P|\mathbf{c}^t)_1}\big(U_{t+1}^{\alpha} \big) \big]^{\frac{\rho}{\alpha}}
\end{equation}
where $U_t$ is the value in period $t$ and the expectation is computed with respect to $\phi(P|\mathbf{c}^t)_1$, which is the probability distribution of the consumptions in period $t+1$ conditional on consumptions in the first $t$ periods $\mathbf{c}^t$. 

\Equationref{eq_recursion1} can also be rewritten as a special case of the more general recursion as in \cite{aer2014time_premium}, which is defined over temporal lotteries:
\begin{equation}\label{eq_recursion2}
    U_t^{\rho} = (1-\beta)c_t^{\rho} + \beta \big[\mathbb{E}_{t}\big(U_{t+1}^{\alpha} \big) \big]^{\frac{\rho}{\alpha}}
\end{equation}

\Equationref{eq_recursion2} can lead to different models given different assumptions on $\mathbb{E}_{t}$. Take any temporal lottery $d$ and its induced lottery $P$. If the conditional expectation in period $t$ $\mathbb{E}_{t}$ is computed utilizing all available information in period $t$ about future consumptions, including both the history of consumptions $\mathbf{c}^t$ and information due to early resolution of uncertainty in $d$, then \Equationref{eq_recursion2} is exactly the CRRA-CES EZ model adopted in \cite{jf2004LLR} and \cite{aer2014time_premium}.

By comparison, if in period $t$, the consumer evaluates future utility solely based on information consisting of consumptions up to period $t$, then she will exhibit indifference to temporal resolution of uncertainty and \Equationref{eq_recursion2} reduces to \Equationref{eq_recursion1}, i.e., the CRRA-CES KM-BIB model. This also suggests that CRRA-CES KM-BIB agrees with CRRA-CES EZ on temporal lotteries where there is no early resolution of uncertainty. 

Moreover, if in period $t$, the consumer evaluates future utility as if she knows the realizations of all future consumptions, then she also satisfies indifference to temporal resolution of uncertainty. In this case, \Equationref{eq_recursion2} reduces to the static interpretation of CRRA-CES EZ as a special case of EU as studied in \cite{DGO2020SI}. 

It is important to distinguish between the information at {\it evaluation} mentioned above and the information at {\it decision}. Take a standard consumption saving problem for an example. The information at {\it decision} is the actual available information for the consumer in each period $t$ when she contemplates the consumption and portfolio choice in period $t$, including all previous consumptions, portfolio weights, state variables, current income shocks and so on. Any decision rule of consumptions and portfolio weights will induce a temporal lottery over consumptions. 

Then the consumer evaluates this temporal lottery using her utility function. Different utility functions will imply different ``as-if'' information at {\it evaluation} used to compute the expectation $\mathbb{E}_{t}$ in \Equationref{eq_recursion2}. Concretely, CRRA-CES EZ implies that the information at {\it evaluation} agrees with the information at {\it decision}; a consumer with CRRA-CES KM-BIB computes $\mathbb{E}_{t}$ as if she only knows consumptions up to period $t$; the static interpretation of CRRA-CES EZ in \cite{DGO2020SI} suggests that the  information at {\it evaluation} includes realizations of all future consumptions.

We end this section with a brief discussion on how to apply CRRA-CES KM-BIB in finance and macroeconomics. First, when there is no early resolution of uncertainty in all feasible consumption plans in the problem,   CRRA-CES KM-BIB shares the same predictions as CRRA-CES EZ. However, this condition fails in most applications including \cite{jf2004LLR} and we need to distinguish the information at evaluation from the information at decision carefully. In the case where $1-\alpha=$ RRA > 1\slash EIS $=1-\rho$, i.e. $\rho>\alpha$,\footnote{In the CRRA-CES EZ model, $\rho>\alpha$ if equivalent to preference for early resolution of uncertainty, although it has no such implication in our model. This condition has been either verified or assumed in empirical works on asset pricing. See \cite{jf2004LLR} for a detailed discussion. Hence we focus on the case with $\rho>\alpha$.} we conjecture that we can use consumption plans with essentially no early resolution of uncertainty to approximate the optimal value of the consumer.\footnote{Here we briefly discuss the intuition behind this conjecture. We say a temporal lottery is {\it feasible} if it can be induced by a feasible decision rule. Consider a feasible temporal lottery that has early resolution of uncertainty and different valuations under EZ and KM-BIB. Then there must exist two decision nodes in the same period $t$ where i) the histories of consumptions coincide, ii) the future utility prospects differ, and iii) the consumer chooses the same current consumption. This is exactly where KM-BIB is discontinuous in the weak convergence topology on $\mathcal{P}$. Since $\rho>\alpha$, the value of the temporal lottery is strictly higher under EZ than that under KM-BIB. Hence, if we slightly modify the consumption at one of the  two decision nodes, the above early resolution of uncertainty can be eliminated and the utility of the new temporal lottery under KM-BIB would be at least weakly higher than before. Repeat the argument and we can essentially eliminate all early resolution of uncertainty in the optimal consumption plan. I'm currently working on the formal analysis in a follow-up work.} Then we can derive the (approximate) Euler equations as in \cite{jpe1991EZ} and conduct similar analysis in \cite{jf2004LLR}  and \cite{aer2009rare}. In this way, we believe that the explanatory power of CRRA-CES KM-BIB is comparable with that of CRRA-CES EZ in those applications.

\section{Conclusion}
This paper generalizes the expected utility model for preference over lotteries on multi-source outcome profiles to incorporate two simplifying heuristics commonly used in the aggregation of risks:  choice bracketing and correlation neglect. We provide characterization results for the generalized models by relaxing the vNM independence axiom. We then apply our framework and  representations to different setups by varying the interpretations of different sources of outcomes. For example, with the interpretation of simultaneous monetary gambles, our model can explain experimental findings on narrow bracketing in \cite{aer2009narrow}. With the interpretation of background risk, our model provides one way to accommodate risk aversion over small favorable gambles in \cite{ecta2000calibration}. With the interpretation of intertemporal choices, we provide a unified framework to study several seemingly distinct models of time preferences in the literature and introduce a new class of models that can satisfy many desirable normative properties on time preferences.

One main point of the paper is that narrow bracketing and correlation neglect can be modelled as natural distortions of the independence axiom and should not be viewed as more ``irrational'' or more behavioral  than other commonly accepted non-EU theories in the literature. Hence, we think it might be worthwhile to incorporate these two heuristics in various economic applications. 

In a follow-up work, we formally extend the current  the two-source framework to multiple sources and infinite horizons and axiomatize a recursive version of our KM-BIB model. We also show how to apply our model in macroeconomics and finance. In another ongoing work, we consider general models of correlation misperception by modeling the correlation structure among risks in difference sources using copula theory. Moreover, as mentioned in footnote \ref{footnote_FORA}, one can extend our framework to consider non-EU models in each single source and incorporate factors like first order risk aversion and Allais Paradox.

\bibliographystyle{ecta}
\bibliography{bracketing}

\newpage

\section*{Appendix: Omitted Proofs}

For simplicity, we use abbreviations for each axiom. We have Axiom Weak Order (WO), Monotonicity (M), Weak Continuity (WC), Weak Independence (WI), Correlation Neglect (CN) and Correlation Sensitivity (CS). Also, We will denote the first part of Axiom WI as Axiom CI and the second part as Axiom WMI. For any $q\in \lxone$ and $q'\in \lxtwo$, we denote the restriction of $\succsim$ on  $\mathcal{L}^0(X_1)\times \{q'\}$ as $\succsim_{1|q'}$ and the restriction of $\succsim$ on  $ \{q\} \times\mathcal{L}^0(X_2)$ as $\succsim_{2|q}$. $\succsim_{1|q'}$ is called the {\it conditional preference} in source $1$ given lottery $q'$ in source 2 and $\succsim_{2|q}$ is called the {\it conditional preference} in source $2$ given lottery $q$ in source 1. 

If $\cone =+\infty$, then we denote that $(\cone,q)\succ (p,q)$ for all $(p,q)\in \hat{\cP}$. If $\lcone =-\infty$, then we denote that $(\lcone,q)\prec (p,q)$ for all $(p,q)\in \hat{\cP}$. Similar notions can be defined for $\ctwo = +\infty$ and $\lctwo =-\infty$.

\begin{proof}[Proof of \thmref{thm_cn}] ~\\
 {$ii)\Rightarrow i)$.} We first prove the necessity of these axioms.  Axioms WO and CN trivially hold.  With  Axiom CN, Axiom WC is equivalent to continuity of $\succsim$ on the subdomain of product lotteries $\hat{\mathcal{P}}$, which is implied by the continuity and boundedness of $w$, $v_1$ and $v_2$. 

For $p,q\in \lr$, we denote $p\succsim_{FOSD} q$ if for any $x\in \bR$, $\sum_{y\leq x}p(y)\leq \sum_{y\leq x}q(y)$ and $p\succ_{FOSD} q$ if $p\succsim_{FOSD} q$ and $p\neq q$. For Axiom M, if $P$ dominates $(\delta_{x_1},\delta_{x_2})$, then $P_i\succsim_{FOSD}\delta_{x_i}$ for $i=1,2$ and at least one ranking is strict. Then monotonicity of $w$, $v_1$ and $v_2$ guarantees that $P\succ (\delta_{x_1},\delta_{x_2})$. By a similar argument, $(\delta_{x_1},\delta_{x_2})\succ P$ if  $(\delta_{x_1},\delta_{x_2})$ dominates $P$.  Therefore Axiom M is satisfied.

Now we check  Axiom WI.  First, if $\succsim$ admits an EU-CN representation, then by  \lemmaref{lemma_eu-cn}, $\succsim$ satisfies Axiom Multilinear Independence, and hence Axiom WI. Second, suppose that $\succsim$ admits a GBIB-CN representation $(w,v_1,v_2,H_2)$, that is, 
 $$V^{GBIB-CN}(P)= \begin{cases}
w(CE_{v_1}(P_1),CE_{v_2}(P_2)), \hbox{~if~} CE_{v_2}(P_2)\in X_2\backslash H_2\\
\sum w(x,CE_{v_2}(P_2))P_1(x), \hbox{~if~} CE_{v_2}(P_2)\in H_2\end{cases}$$
For each $p\in \lxone$, $\succsim_{2|p}$ is represented by an EU with index $v_2$. For each $p'\in \lxtwo$, when $CE_{v_2}(p')\in X_2\backslash H_2$, then $\succsim_{1|p'}$ is represented by an EU with index $v_1$. Moreover, as $0\not\in H_2$, $\succsim_1$ is admits an EU representation with index $v_1$. When $CE_{v_2}(p')\in H_2$, then $\succsim_{1|p'}$ is represented by an EU with index $w(\cdot,CE_{v_2}(p'))$. Hence, Axiom CI is satisfied.

Then we check Axiom WMI.  Fix $P,Q,R,S\in\hat{\mathcal{P}}$, $\alpha\in (0,1)$ and $i,j\in\{1,2\}$ with $P_i=R_i,  Q_j=S_j, P_{-i}\sim_{-i} R_{-i}$, $Q_{-j}\sim_{-j} S_{-j}$, $P\succ Q$ and $R\sim S$. First we claim that Axiom WMI holds if $P\sim R$ or $Q\sim S$.  Suppose that $P\sim R$, then by Axiom CI, for all $\alpha\in (0,1)$, $\alpha P +(1-\alpha)R\sim P \sim R$ and either $Q\succsim \alpha Q + (1-\alpha)S$ or $S\succsim \alpha Q + (1-\alpha)S$. If $Q\succsim S$, then $\alpha P +(1-\alpha)R\sim P \succ Q\succsim  \alpha Q + (1-\alpha)S$. If instead $S\succ Q$, then  $\alpha P +(1-\alpha)R\sim R \sim S\succ  \alpha Q + (1-\alpha)S$. This proves Axiom WMI. Similar arguments hold for $Q\sim S$.

Now we consider the following three cases. 
\begin{itemize}
    \item Case 1: Suppose that $i=1$. Then $P_1=R_1$ and $P_2\sim_2 R_2$, which implies $CE_{v_2}(P_2)= CE_{v_2}(R_2)$. We know $P\sim R$ and hence Axiom WMI holds.
    
    \item Case 2: Suppose that $j=1$. Then $Q_1=S_1$ and $Q_2\sim_2 S_2$, which implies $CE_{v_2}(Q_2)= CE_{v_2}(S_2)$. We know $Q\sim S$ and hence Axiom WMI holds.
    
     \item Case 3: Suppose that $i=j=2$. \\
     If $CE_{v_2}(P_2)\in X_2\backslash H_2$ or $CE_{v_2}(Q_2)\in X_2\backslash H_2$, then either $P\sim R$ or $Q\sim S$ and we are done. If $CE_{v_2}(P_2), CE_{v_2}(Q_2)\in H_2$, then the GBIB-CN representation is linear in marginal lotteries in source 1. Then for any $\alpha \in (0,1)$,
     \begin{align*}
      V^{GBIB-CN}(\alpha P+ (1-\alpha)R) &= \alpha  V^{GBIB-CN}(P)+ (1-\alpha)V^{GBIB-CN}(R)\\
      &> V^{GBIB-CN}(Q)+ (1-\alpha)V^{GBIB-CN}(S)\\
      &= V^{GBIB-CN}(\alpha Q+ (1-\alpha)S).
     \end{align*}
     This verifies Axiom WMI.
     
\end{itemize}
Thus Axiom WI holds for GBIB-CN. A symmetric proof applies if $\succsim$ admits a FBIB-CN representation. This completes the proof for necessity of axioms.

\bigskip

\noindent {$i)\Rightarrow ii)$.} Suppose that all axioms hold.  For each $x_i\in X_i$ and $i=1,2$,  denote $\Pi^i(x_i)$ as the set of marginal lotteries in source $i$ with certainty equivalent $x_i$. Formally,  $\Pi^i(x)= \{p\in \mathcal{L}^0(X_i): p\sim_i \delta_{x_i}\}$.  For any two product lotteries $P,Q \in \hat{\cP}$ with $P\succsim Q$,  let $[Q,P]$ denote the set of all product lotteries whose utilities lie between $P$ and $Q$, that is, $[Q,P]= \{S\in \hat{\cP}: P\succsim S\succsim Q\}$.  

Also, for each $q_i\in \mathcal{L}^0(X_i)$ and $x_i\in X_i$, $i=1,2$, define 
\begin{align*}
\Gamma(x_1,x_2)= \bigcup_{\substack{P,Q\in \Pi^1(x_1)\times \Pi^2(x_2),\\ P\succsim Q}}[Q,P]
\end{align*}
\begin{align*}
\Gamma_{1,q_1}(x_2)= \bigcup_{\substack{P,Q\in \{q_1\}\times \Pi^2(x_2),\\ P\succsim Q}}[Q,P]~,~~ 
\Gamma_{2,q_2}(x_1)= \bigcup_{\substack{P,Q\in  \Pi^1(x_1)\times \{q_2\},\\ P\succsim Q}}[Q,P].
\end{align*}
Intuitively, $\Gamma(x,y)$ includes all product lotteries whose utilities are bounded by lotteries in $\Pi^1(x)\times \Pi^2(y)$. $\Gamma_{1,q_1}(y)$ and $\Gamma_{2,q_2}(x)$ admit similar interpretations. We further define 
$$\Gamma_{1,q_1}=\bigcup_{x_2\in X_2} \Gamma_{1,q_1}(x_2)~,~~\Gamma_{2,q_2}=\bigcup_{x_1\in X_1} \Gamma_{2,q_2}(x_1).$$

For any set of lotteries $\mathcal{A}\subseteq \cP$, denote $\max_{\succsim} \mathcal{A}=\{P\in \mathcal{A}: P\succsim Q,\forall Q\in \mathcal{A}\}$ whenever it is well-defined. That is, $\max_{\succsim} \mathcal{A}$ is the set of most preferred lotteries in $\mathcal{A}$ under $\succsim$. 

Finally, for any set $A$, we denote $A^o$ as its interior  and $\partial A$ as its boundary with respect to the appropriate topology.  

{\bf Step 1: Direct implications of axioms.} 

First, by Axiom CN, $P\sim (P_1, P_2)$ for all $P\in\mathcal{P}$ and it suffices to consider the restriction of $\succsim$ on product lotteries $\hat{\mathcal{P}}$. Then Axiom WC implies that $\succsim$ satisfies topological continuity. 

The following lemma shows the EU representation of the conditional preference $\succsim_{i|q}$ for each $i=1,2$ and $q\in \lxmi$.

\begin{lemma}\label{lemma_narrow pre}
For each $i=1,2$ and $q\in \lxmi$, the conditional preference $\succsim_{i|q}$ admits an EU representation with a utility index $v_{i|q}$, which is continuous, bounded and unique up to a positive affine transformation. Moreover, if $q\in X_{-i}$, then $v_{i|q}$ can be chosen to be strictly monotone (and hence regular). 
\end{lemma}

\begin{proof}[Proof of \lemmaref{lemma_narrow pre} ]
Fix $i=1,2$ and $q\in \lxmi$.  By Axiom WC, the conditional preference $\succsim_{i|q}$ is continuous. By Axiom CI, $\succsim_{i|q}$ admits an EU representation with a  continuous utility index $v_{i|q}$ defined on $X_i$, which is unique up to a positive affine transformation. Normalize that $v_{i|q}(0)=0$.   Suppose by contradiction that $v_{i|q}$ is unbounded, then for any positive integer $n$, there exists $x_n\in \bR$ such that $|v_{i|q}(x_n)|>n$.  There exists a subsequence $\{x_{n_k}\}_{k\geq 1}$ such that $v_{i|q}(x_{n_k})>n_k$ for each $k$ or $v_{i|q}(x_{n_k})<-n_k$ for each $k$. Suppose, without loss of generality,  that the former case holds.\footnote{This proof technique will be used for multiple times below. For simplicity, we will call it ``the subsequence arguments'' and denote the subsequence as the original sequence, which is without loss of generality.} Consider the marginal lottery $p_{n_k} = \frac{1}{n_k}\delta_{x_{n_k}} + \frac{n_k-1}{n_k}\delta_0$ for each $k$. By  continuity of $v_{i|q}$, we can find $\epsilon>0$ with $v_{i|q}(\epsilon)\in (0,1)$.
 For each $k$, the utility of $p_{n_k}$ is $U_{i|q}(p_{n_k})= \frac{1}{n_k}v_{i|q}(x_{n_k})>1>v_{i|q}(\epsilon) = U_{i|q}(\delta_{\epsilon})$, which means $p_{n_k}\succsim_{i|q} \delta_{\epsilon}$.  Meanwhile, $p_{n_k}\xrightarrow{w}\delta_0 \prec_{i|q} \delta_{\epsilon}$ as $v_{i|q}(\epsilon)>v_{i|q}(0)=0$. This contradicts with the continuity of $\succsim_{i|q}$. As a result, $v_{i|q}$ is bounded.  Moreover, if $q\in X_{-i}$, that is, $q=\delta_y$ for some $y\in X_{-i}$, then by Axiom M, we know $v_{i|q}$ must be strictly monotone. 
 \end{proof}

When $q=\delta_0$, then the conditional preference in source $i$ agrees with the narrow preference in source $i$ and its EU index is denoted as $v_i$ for simplicity.  It is worthwhile to note that for each $i=1,2$, if $x\in X_i\backslash X_i^o$, then $\Pi_i(x)=\{\delta_x\}$. 

A direct corollary of \lemmaref{lemma_narrow pre} guarantees the existence of ``certainty equivalents''.

\begin{corollary}\label{coro_CE}
For each $P\in \hat\cP$, there exists $x_1,y_1\in X_1, x_2,y_2\in X_2$ such that $P\sim (P_1,x_2)\sim (x_1,P_2)\sim (y_1,y_2).$
\end{corollary}

\begin{proof}[Proof of \cororef{coro_CE} ]
Suppose that $P_1\not\in X_1, P_2\not\in X_2$. The case where $P_1\in X_1$ or $P_2\in X_2$ is easier to prove. By \lemmaref{lemma_narrow pre}, we know there exists $a,a'\in X_2$ such that $v_{2|P_1}(a)>\sum_xv_{2|P_1}(x)P_2(x)> v_{2|P_1}(a')$. Since $v_{2|P_1}$ is continuous and $X_2$ is a closed interval, there exists $x_2\in X_2$ where $v_{2|P_1}(x_2)= \sum_xv_{2|P_1}(x)P_2(x)$, which implies $P\sim (P_1,x_2)$.  Similarly, we can find $x_1\in X_1$ with $P\sim (x_1,P_2)$. Now let $y_2=x_2$. Repeat the above arguments for product lottery $(P_1,x_2)$ and we know there exists $y_1\in X_1$ such that $(y_1,x_2)\sim (P_1,x_2)\sim P$.
\end{proof}

The next lemma summarizes two implications of Axiom WC and Axiom WI. 

\begin{lemma}\label{lemma_axiom implication}
(i). For each $P,Q,R\in \hat{\cP}$ with $P\succsim R\succsim Q$, $P\succ Q$ and $P_i=Q_i$ for some $i\in \{1,2\}$, then there exists a unique $\lambda\in [0,1]$ such that $R\sim \lambda P + (1-\lambda)Q$.\\
(ii). For each $P,Q,R,S\in\hat{\mathcal{P}}$, $\alpha\in (0,1)$ and $i,j\in\{1,2\}$, if $P_i=R_i,  Q_j=S_j, P_{-i}\sim_{-i} R_{-i}$ and $Q_{-j}\sim_{-j} S_{-j}$, then
\begin{align*}
    &P\sim Q, R\sim S\Longrightarrow \alpha P + (1-\alpha) R\sim \alpha Q + (1-\alpha) S\\
    &P\succ Q, R\succ S\Longrightarrow \alpha P + (1-\alpha) R\succ \alpha Q + (1-\alpha) S
\end{align*}
\end{lemma}

\begin{proof}[Proof of \lemmaref{lemma_axiom implication}]
(i). Denote $I=\{\eta\in [0,1]: R\succ \eta P + (1-\eta)Q\}$ and $\lambda=\sup I$. $\lambda$ is well-defined as $I$ is bounded. We claim that $\lambda P+(1-\lambda)Q\sim R$. If $\lambda P+(1-\lambda)Q\succ R$, then $\lambda>0$ and by mixture continuity of $\succsim$, there exists $\epsilon>0$ with $(\lambda-\epsilon) P+(1-\lambda + \epsilon)Q\succ R$. This implies $\lambda-\epsilon\not\in I$. Since  $P\succ Q$ and $P_i=Q_i$ for some $i$, by Axiom CI, for any $\alpha,\beta\in [0,1]$, $\alpha P+(1-\alpha) Q \succ \beta P+(1-\beta) Q$ if any only if $\alpha >\beta$, which implies $[\lambda-\epsilon, \lambda]\cap I =\emptyset$ and leads to a contradiction with $\lambda=\sup I$.  If instead $R\succ \lambda P+(1-\lambda)Q$, then  there exists $\epsilon>0$ with $R\succ (\lambda+\epsilon) P+(1-\lambda - \epsilon)Q$ and hence $\lambda+\epsilon\in I$, which again contradicts with the definition of $\lambda$.\\
(ii). Consider the case where $P\sim Q, R\sim S$. If $P\sim R$, then the result trivially holds as $\alpha P + (1-\alpha) R\sim R\sim Q \sim \alpha Q + (1-\alpha) S$. Without loss of generality, suppose that $P\succ R$. Then $Q\succ S$ and $P\succ \alpha P + (1-\alpha) R\succ R$, $Q\succ \alpha Q + (1-\alpha) S\succ S\sim R$. Suppose by contradiction that $\alpha P + (1-\alpha) R\succ \alpha Q + (1-\alpha) S$. By part (i), there exists a unique $\lambda\in (0,1)$ with $\alpha Q+ (1-\alpha)S\sim \lambda (\alpha P + (1-\alpha) R) + (1-\lambda) R = \alpha\lambda P + (1-\alpha\lambda)R$. Notice that $Q\sim P\succ \alpha P + (1-\alpha) R$, $S\sim R$, $Q_j=S_j$, $Q_{-j}\sim_{-j} S_{-j}$, $(\alpha P + (1-\alpha) R)_i = P_i=R_i$ and $(\alpha P + (1-\alpha) R)_{-i} \sim_{-i} R_{-i}$. The last one holds as $\succsim_{-i}$ admits an EU representation. Hence, Axiom WMI implies that
$$\alpha Q + (1-\alpha) S\succ \lambda (\alpha P + (1-\alpha) R) + (1-\lambda) R = \alpha\lambda P + (1-\alpha\lambda)Q$$
which leads to a contradiction. The case for $\alpha P + (1-\alpha) R\prec \alpha Q + (1-\alpha) S$ is symmetric.

Now assume $P\succ Q, R\succ S$.  If $P\sim R$, then the result  holds as $\alpha P + (1-\alpha) R\sim P\succ \max_{\succsim}\{Q,S\}\succsim \alpha Q + (1-\alpha) S$ for all $\alpha\in (0,1)$.  Without loss of generality, suppose $P\succ R$. 

If $R\succsim Q$, then  $\alpha P + (1-\alpha) R\succ  R\succsim \max_{\succsim}\{Q,S\}\succsim \alpha Q + (1-\alpha) S$ for all $\alpha\in (0,1)$. 

If $Q\succ R$, then $P\succ Q\succ R\succ S$. By part (i) of this lemma, we can find $\lambda\in (0,1)$ such that $R\sim \lambda Q+ (1-\lambda) S :=S'$. Then $S'_j=S_j=Q_j$ and $S'_{-j}= \lambda Q_{-j}+ (1-\lambda) S_{-j}\sim_{-j}  Q_{-j}$ as $Q_{-j}\sim_{-j} S_{-j}$. Then the primitives of Axiom WMI hold for the tuple $(P,Q,R,S')$ and for any $\alpha\in (0,1)$, 
$$\alpha P + (1-\alpha)   R\succ \alpha Q + (1-\alpha)S' \succ \alpha Q + (1-\alpha)S.$$
The second strict ranking comes from Axiom CI and $S'\sim R\succ S$. This completes the proof. \end{proof}

A tuple $(P,Q,R,S)\in \hat{\cP}^4$ is called {\it proper} if $P_i=R_i$, $Q_j=S_j$ for some $i,j\in\{1,2\}$ and $P\succsim R, Q\succsim S$. A proper tuple $(P,Q,R,S)$  satisfies the {\it independence property} if one of the following conditions holds:
\begin{itemize}
    \item $P\succ Q, R\sim S$ and for all $\alpha\in (0,1)$, $\alpha P + (1-\alpha) R\succ \alpha Q + (1-\alpha) S$;
    \item $P\sim Q, R\succ S$ and for all $\alpha\in (0,1)$, $\alpha P + (1-\alpha) R\succ \alpha Q + (1-\alpha) S$;
    \item $P\sim Q, R\sim S$ and for all $\alpha\in (0,1)$, $\alpha P + (1-\alpha) R\sim \alpha Q + (1-\alpha) S$;
    \item $P\succ Q, R\succ S$ and for all $\alpha\in (0,1)$, $\alpha P + (1-\alpha) R\succ \alpha Q + (1-\alpha) S$.
    \end{itemize}

We end this section by showing that for  each $x\in X_1, y\in X_2$, any product lottery in $\Gamma(x,y)$ is indifferent to some lottery in $\Pi^1(x)\times \Pi^2(y)$. Similar results also hold for $\Gamma_{1,q_1}(y)$ and $\Gamma_{2,q_2}(x)$ for each $q_1\in \lxone, q_2\in \lxtwo$.

\begin{lemma}\label{lemma_achievable}
Fix  $q_1\in \lxone, q_2\in \lxtwo$ and $x\in X_1, y\in X_2$.  \\
(i). For each $P\in \Gamma(x,y)$, there exists $P'\in \Pi^1(x)\times \Pi^2(y)$ with $P'\sim P$; \\
(ii). For each $P\in \Gamma_{1,q_1}(y)$, there exists $P'\in \{q_1\}\times \Pi^2(y)$ with $P'\sim P$; \\
(iii). For each $P\in \Gamma_{2,q_2}(x)$, there exists $P'\in  \Pi^1(x)\times \{q_2\}$ with $P'\sim P$.
\end{lemma}

\begin{proof}[Proof of \lemmaref{lemma_achievable}]
(i). By definition, there exists $Q,Q'\in \Pi^1(x)\times \Pi^2(y)$ with $Q\succsim P\succsim Q'$. Denote $Q''=(Q_1,Q'_2)\in \Pi^1(x)\times \Pi^2(y)$. We have either  $Q''\succsim P\succsim Q'$ or $Q\succsim P\succsim Q''$. As $Q''_1=Q_1, Q''_2=Q'_2$, by part (i) of \lemmaref{lemma_axiom implication}, we know there exists $\lambda\in (0,1)$ with $P\sim \lambda Q'' + (1-\lambda)Q'$ or $P \sim \lambda Q'' + (1-\lambda)Q$. \lemmaref{lemma_narrow pre} guarantees that $\lambda Q'' + (1-\lambda)Q', \lambda Q'' + (1-\lambda)Q \in \Pi^1(x)\times \Pi^2(y)$. The proofs for (ii) and (iii) are similar.\end{proof}

\bigskip

{\bf Step 2. Suppose that the DM narrowly brackets marginal lotteries in both sources.} That is,  $(p,q)\sim (\delta_x,\delta_y)$ for all $(x,y)\in X_1\times X_2$ and $(p,q)\in \Pi^1(x)\times \Pi^2(y)$. The following lemma shows that $\succsim$ must admit a NB representation. 

\begin{lemma}\label{lemma_nb}
Suppose that $(p,q)\sim (\delta_x,\delta_y)$ for all $(x,y)\in X_1\times X_2$ and $(p,q)\in \Pi^1(x)\times \Pi^2(y)$, then $\succsim$ admits a NB representation.
\end{lemma}
\begin{proof}[Proof of \lemmaref{lemma_nb}]
By \lemmaref{lemma_narrow pre}, for $i=1,2$, denote $v_i$ as the EU index of $\succsim_i$. Since $X_i$ is a closed interval and $v_i$ is regular, the certainty equivalent function $CE_{v_i}$ is well-defined. Then for any $(p,q)\in \hat{\mathcal{P}}$, we know $(p,q)\sim (\delta_{CE_{v_1}(p)}, \delta_{CE_{v_2}(q)})$.

 Denote a binary relation $\hat{\succsim}$ over $X_1\times X_2$ such that for all $(x,y), (x',y') \in X_1\times X_2$, $(\delta_x,\delta_y)\succsim (\delta_{x'},\delta_{y'})$ if and only if $(x,y)\hat\succsim (x',y')$. Axiom WC implies that $\hat\succsim$ is continuous on $X_1\times X_2$, which is a separable metric space. By Debreu's Theorem, $\hat\succsim$ admits a continuous representation $w$. Axiom M guarantees that $w$ is strictly monotone. Without loss of generality, we can assume that $w(0,0)=0$ and $w$ is bounded, because the bounded monotone transformation $w'(x,y)= 1-exp(-w(x,y))$ for $w(x,y)\geq 0$ and $w'(x,y)= exp(w(x,y))-1$ for $ w(x,y)< 0$ still represents $\hat{\succsim}$. Therefore we can find regular functions $w, v_1$ and $v_2$  such that for all $P,Q\in\cP$,
\begin{align*}
    P\succsim Q&\Longleftrightarrow (\delta_{CE_{v_1}(P_1)}, \delta_{CE_{v_2}(P_2)})\succsim (\delta_{CE_{v_1}(Q_1)}, \delta_{CE_{v_2}(Q_2)})\\
    &\Longleftrightarrow ({CE_{v_1}(P_1)}, {CE_{v_2}(P_2)})\hat\succsim ({CE_{v_1}(Q_1)}, {CE_{v_2}(Q_2})\\
    &\Longleftrightarrow w(CE_{v_1}(P_1),CE_{v_2}(P_2))\geq w(CE_{v_1}(Q_1),CE_{v_2}(Q_2))
\end{align*}
That is, $\succsim$ admits a NB representation $(w,v_1,v_2)$.
\end{proof}

From now on, we maintain the assumption that there exist $(x,y)\in X_1\times X_2$ and $(p,q)\in \Pi^1(x) \times \Pi^2(y)$ such that $(p,q)\not\sim (\delta_x,\delta_y)$.

\bigskip

{\bf Step 3: Suppose that the DM narrowly brackets marginal lotteries in source 2.} This is equivalent to assuming $(p,q)\sim (p,q')$ for all $(x,y)\in X_1\times X_2$, $p\in \Pi^1(x)$ and $q,q'\in \Pi^2(y)$. Denote the condition as {\bf Assumption 1}.  Then for any $(p,q), (p',q')\in \hat{\mathcal{P}}$ with $q\sim_2 q'$, $(p,q)\succsim (p',q')$ if and only if $(p,{CE_{v_2}(q)}) \succsim (p',{CE_{v_2}(q)})$. Hence we can focus on the restriction of $\succsim$ on $\lxone \times X_2$.

By the assumption at the end of Step 2, we can find $(x_0,y_0)\in X_1\times X_2$, $p_0,p_0'\in \Pi^1(x_0)$ and $q_0\in \Pi^2(y_0)$ such that $(p_0,q_0)\succ (p_0',q_0)$.  This implies $y_0\neq 0$. By Axiom WC, it is without loss of generality to assume $y_0\in X_2^o$. By \lemmaref{lemma_narrow pre}, we know that $\succsim_{1|q_0}$ admits an EU representation with a continuous and bounded utility index $v_{1|\delta_{y_0}}$. Recall that $X_i^o$ is the interior of $X_i$ with respect to $\bR$, $i=1,2$. For each $x\in X_i^o$, there exists $y,y'\in X_i^o$ with $y>x>y'$. Suppose that there exists $x_1\in X_1^o$ such that $(p_1,q_0)\sim (p'_1,q_0)$ for all $p_1,p'_1\in \Pi^1(x_1)$. Clearly, $x_1\neq x_0$. Denote $\up{x},\dw{x}\in X_i^o$ with $\up{x}>x_1>\dw{x}$. As $v_1$ and $v_{1|\delta_{y_0}}$ are unique up to positive affine transformations, we can set $v_{1|\delta_{y_0}}(\dw{x})= v_1(\dw{x})$ and $v_{1|\delta_{y_0}}(x_1)= v_1(x_1)$.  For any $x\in X_1$ with $x>x_1$, we can find $\alpha\in (0,1)$ with $\alpha \delta_{x} +(1-\alpha)\delta_{\dw{x}}\in \Pi^1(x_1)$. Then $ \alpha v_1(x) +(1-\alpha)v_1(\dw{x})= v_1(x_1)$ and $(\alpha \delta_{x} +(1-\alpha)\delta_{\dw{x}}, q_0)\sim (\delta_{x_1},q_0)$, which implies $$ \alpha v_{1|\delta_{y_0}}(x) + (1-\alpha) v_{1|\delta_{y_0}}(\dw{x})= v_{1|\delta_{y_0}}(x_1)=v_1(x_1)=\alpha v_1(x) +(1-\alpha)v_1(\dw{x}).$$
Since $v_{1|\delta_{y_0}}(\dw{x})= v_1(\dw{x})$ and $\alpha\in (0,1)$, $v_{1|\delta_{y_0}}(x)=v_1(x)$. Specifically, we have $v_{1|\delta_{y_0}}(\up{x})=v_1(\up{x})$.
Now we consider $x\in X_1$ with $x<x_1$. There exists $\beta\in (0,1)$ with $\beta \delta_{\up{x}} +(1-\beta)\delta_x\in \Pi^1(x_1)$. Then $ \beta v_1(\up{x}) +(1-\beta)v_1({x})= v_1(x_1)$ and $(\alpha \delta_{\up{x}} +(1-\alpha)\delta_{{x}}, q_0)\sim (\delta_{x_1},q_0)$, which also implies $v_{1|\delta_{y_0}}(x)=v_1(x)$.  Thus $v_{1|\delta_{y_0}}\equiv v_1$, contradicting with $(p_0,q_0)\succ (p'_0,q_0)$ as $p_0,p'_0\in \Pi^1(x_0)$.   Thus, there exists $y_0\in X_2$ such that for any $x\in X^o_1$, we can find $p_x,p'_x\in \Pi^1(x)$ and $q_0\in \Pi^2(y_0)$ with $(p_x,q_0)\succ (p'_x,q_0)$.

Denote $\Sigma^2:= \{y\in X_2^o: \exists~ x\in X_1 \hbox{~and~} p,p'\in \Pi^1(x), q\in \Pi^2(y) \hbox{~s.t.~} (p,q) \succ (p',q) \}$. $\Sigma^2$ is nonempty as $y_0\in \Sigma^2$ and is open in $X_2$ by Axiom WC. Also $0\not\in \Sigma^2$ and hence $\Sigma^2 \subseteq \bR\backslash \{0\}$. 
Denote the closure of $\Sigma^2$ in $X_2$ as $cl({\Sigma}^2)$.

 The following lemma provides a sufficient condition for a proper tuple to satisfy the independence property.

 \begin{lemma}\label{lemma_source2_ind}
Suppose that Assumption 1 holds. Then a proper tuple $(P,Q,R,S)\in (\lxone \times X_2)^4$ satisfies the independence property if $P_2= R_2=\delta_{y_1}, Q_2=S_2=\delta_{y_2}$ with $y_1,y_2\in cl({\Sigma}^2)$.
\end{lemma}

The proof of \lemmaref{lemma_source2_ind} requires several intermediate results. The first one assures that we can focus on the case where $P\sim Q$, $R\sim S$.  

 \begin{lemma}\label{lemma_source2_strict_ind}
Suppose that Assumption 1 holds. $(P,Q,R,S)\in (\lxone \times X_2)^4$ is a proper tuple with $P_2= R_2=\delta_{y_1}, Q_2=S_2=\delta_{y_2}$ with $y_1,y_2\in cl({\Sigma}^2)$. If the independence property holds for any such $(P,Q,R,S)$ with $P\sim Q, R\sim S$, then the independence property holds for any such $(P,Q,R,S)$ with $P\succsim Q, R\succsim S$.
\end{lemma}

\begin{proof}[Proof of \lemmaref{lemma_source2_strict_ind}]
Following similar arguments in the proof of \lemmaref{lemma_axiom implication}, it suffices to consider the case where $P\succsim Q \succ R \succsim S$. By \lemmaref{lemma_narrow pre}, there exist $\alpha\in (0,1]$ and $\beta\in [0,1)$ such that $P' = \alpha P + (1-\alpha) R \sim Q$ and  $S' = \beta Q + (1-\beta) S \sim R$. Then the independence property holds for $(P',Q,R,S')$, that is, for any $\lambda\in (0,1)$, $\lambda P' + (1-\lambda)R \sim \lambda Q + (1-\lambda)S'$. By \lemmaref{lemma_narrow pre} and $P\succsim P'$, $S'\succsim S$, we have 
$$\lambda P + (1-\lambda)R \succsim \lambda P' + (1-\lambda)R \sim \lambda Q + (1-\lambda)S'\succsim \lambda Q + (1-\lambda)S.$$
At least one of the above weak preference rankings would be strict if $P\succ Q$ or $R\succ S$. \end{proof}

\lemmaref{lemma_source2_local} shows the result in \lemmaref{lemma_source2_ind} holds locally, that is, when the utilities of $P$ and $R$ are ``close enough''.

 \begin{lemma}\label{lemma_source2_local}
Suppose that Assumption 1 holds. Then a proper tuple $(P,Q,R,S)\in (\lxone \times X_2)^4$ satisfies the independence property if $P\sim Q, R\sim S$, $P_2= R_2=\delta_{y_1}, Q_2=S_2=\delta_{y_2}$ with $y_1,y_2\in {\Sigma}^2$ and there exist $x_1,x_2\in X_1$ such that $P,Q,R,S\in \Gamma(x_1,y_1)\cap \Gamma(x_2,y_2)$.
\end{lemma}

\begin{proof}[Proof of \lemmaref{lemma_source2_local}]
Suppose  $(P,Q,R,S)\in (\lxone \times X_2)^4$ is a proper tuple that satisfies the conditions stated in the lemma. By \lemmaref{lemma_achievable}, there exist $P'_1, R'_1\in \Pi^1(x_1)$ and $Q'_1, S'_1\in \Pi^1(x_2)$ such that $P':=(P'_1,P_2)\sim P \sim Q\sim Q':=(Q'_1,Q_2)$ and $R':=(R'_1,R_2)\sim R \sim S\sim S':=(S'_1,S_2)$. By part 2 of \lemmaref{lemma_axiom implication},  for any $\alpha\in (0,1)$, $\alpha P' + (1-\alpha)R'\sim \alpha Q' + (1-\alpha)S'$. Finally, \lemmaref{lemma_narrow pre} implies that $\alpha P' + (1-\alpha)R'\sim \alpha P + (1-\alpha)R$ and $\alpha Q' + (1-\alpha)S'\sim \alpha Q + (1-\alpha)S$. By transitivity of $\succsim$, $\alpha P + (1-\alpha)R\sim \alpha Q + (1-\alpha)S$. 
\end{proof}

The next lemma shows that if the independence property holds on two sets of product lotteries respectively, then it also holds on their union. 

 \begin{lemma}\label{lemma_source2_union}
Suppose that Assumption 1 holds. $(P,Q,R,S)\in (\lxone \times X_2)^4$ is a proper tuple where $P\sim Q, R\sim S$, $P_2= R_2=\delta_{y_1}, Q_2=S_2=\delta_{y_2}$ with $y_1,y_2\in {\Sigma}^2$. Fix any $T^i\in \hat\cP$ for $i=1,...,4$ with $T^4\succ T^2\succ T^3 \succ T^1$. If the independence property holds for any such $(P,Q,R,S)$ with $\{P,Q,R,S\}\subseteq [T^1, T^2]$ or $\{P,Q,R,S\}\subseteq [T^3,T^4]$, then it also holds for any such $(P,Q,R,S)$  with $\{P,Q,R,S\}\subseteq [T^1, T^4]$.
\end{lemma}

\begin{proof}[Proof of \lemmaref{lemma_source2_union}]
Without loss of generality, we can assume $P\succ R$ and $T^1\succsim (0,{y_i}), i=1,2$, otherwise either the lemma is  trivial or we can modify  $T^1$ without changing the lemma.  Moreover, it suffices to focus on the case where $P\sim Q \succ T^2$ and $R\sim S\prec T^3$. Fix any $W_1,W_2$ with $T^2\succ W_2\succ W_1 \succ T^3$. Then we have 
$$T^4\succsim P\sim Q \succ T^2 \succ W_2 \succ W_1 \succ T^3\succ R\sim S \succsim T^1.$$

By \lemmaref{lemma_narrow pre}, we can find $\hat{P}, \hat{Q}, \hat{R}, \hat{S}$ such that $\hat{P}\sim\hat{Q}\sim W_2$, $\hat{R}\sim\hat{S}\sim W_1$ and $\hat{P}_2=\hat{R}_2 =P_2=\delta_{y_1}$, $\hat{Q}_2=\hat{S}_2 =Q_2=\delta_{y_2}$. Notice that $P,Q,\hat{R},\hat{S}\in [T^3,T^4]$, where the independence property holds. Then there exists $\lambda\in (0,1)$ such that $\lambda P + (1-\lambda)\hat{R}\sim \hat{P}\sim \lambda Q + (1-\lambda)\hat{S}$. Similarly, we can find $\lambda'\in (0,1)$ with  $\lambda' \hat{P} + (1-\lambda'){R}\sim \hat{R}\sim \lambda' \hat{Q} + (1-\lambda'){S}$.

Actually in the construction of $\hat{R}$ and $\hat{S}$, there exist $\eta_1,\eta_2\in (0,1)$ with 
$$\eta_1 P+ (1-\eta_1)R\sim \hat{R}\sim \hat{S}\sim \eta_2 Q + (1-\eta_2)S.$$
We claim that $\eta_1=\eta_2$. To see this, as $P_2=R_2=\hat{P}_2=\hat{R}_2$, we know 
\begin{align*}
    \lambda'\hat{P} + (1-\lambda')R\sim  \lambda\lambda'{P} +  (1-\lambda)\lambda'\hat{R} + (1-\lambda')R\sim \hat{R}
\end{align*}
which implies 
$$\hat{R}\sim \frac{\lambda\lambda'}{\lambda\lambda'+ (1-\lambda')}P +\frac{1-\lambda'}{\lambda\lambda'+ (1-\lambda')}R $$
and hence $\eta_1 =\frac{\lambda\lambda'}{\lambda\lambda'+ (1-\lambda')}$. Similarly we can show that $\eta_2 = \frac{\lambda\lambda'}{\lambda\lambda'+ (1-\lambda')}=\eta_1:= \eta^{w_1}$. 

A symmetric argument shows that there exists $\eta^{w_2}$ with $\eta^{w_1}<\eta^{w_2}<1$ and 
$$\eta^{w_2} P+ (1-\eta^{w_2})R\sim \hat{P}\sim \hat{Q}\sim \eta^{w_2} Q + (1-\eta^{w_2})S.$$

Now we consider $\eta$ with $\eta^{w_1}< \eta< \eta^{w_2}$. Notice that
\begin{align*}
    \eta P + (1-\eta) R&= \frac{\eta-\eta^{w_1}}{\eta^{w_2}-\eta^{w_1}}[\eta^{w_2} P+ (1-\eta^{w_2})R] + \frac{\eta^{w_2}-\eta}{\eta^{w_2}-\eta^{w_1}}[\eta^{w_1} P+ (1-\eta^{w_1})R]\\
    &\sim \frac{\eta-\eta^{w_1}}{\eta^{w_2}-\eta^{w_1}}\hat{P} + \frac{\eta^{w_2}-\eta}{\eta^{w_2}-\eta^{w_1}}\hat{R}.
\end{align*}
Similarly, 
$$\eta Q + (1-\eta) S\sim \frac{\eta-\eta^{w_1}}{\eta^{w_2}-\eta^{w_1}}\hat{Q} + \frac{\eta^{w_2}-\eta}{\eta^{w_2}-\eta^{w_1}}\hat{S}.$$
As $\hat{P}\sim\hat{Q},\hat{R}\sim\hat{S}\in [T^3, T^4]$ and $\hat{P}_2=\hat{R}_2$, $\hat{Q}_2=\hat{S}_2$, the independence property holds for $(\hat{P},\hat{Q}, \hat{R}, \hat{S})$ and hence 
$$ \eta P + (1-\eta) R\sim \frac{\eta-\eta^{w_1}}{\eta^{w_2}-\eta^{w_1}}\hat{P} + \frac{\eta^{w_2}-\eta}{\eta^{w_2}-\eta^{w_1}}\hat{R}\sim \frac{\eta-\eta^{w_1}}{\eta^{w_2}-\eta^{w_1}}\hat{Q} + \frac{\eta^{w_2}-\eta}{\eta^{w_2}-\eta^{w_1}}\hat{S}\sim\eta Q + (1-\eta) S.$$

Then we check the independence property for $\eta^{w_2}<\eta<1$.
\begin{align*}
   \hat{P}\sim \eta^{w_2} P + (1-\eta^{w_2}) R&= \frac{\eta^{w_2}-\eta^{w_1}}{\eta-\eta^{w_1}}[\eta P+ (1-\eta)R] + \frac{\eta-\eta^{w_2}}{\eta-\eta^{w_1}}[\eta^{w_1} P+ (1-\eta^{w_1})R]\\
    &\sim \frac{\eta^{w_2}-\eta^{w_1}}{\eta-\eta^{w_1}}[\eta P+ (1-\eta)R] + \frac{\eta-\eta^{w_2}}{\eta-\eta^{w_1}}\hat{R}.
\end{align*}
Similarly, 
\begin{align*}
   \hat{Q}\sim \eta^{w_2} Q + (1-\eta^{w_2}) S\sim \frac{\eta^{w_2}-\eta^{w_1}}{\eta-\eta^{w_1}}[\eta Q+ (1-\eta)S] + \frac{\eta-\eta^{w_2}}{\eta-\eta^{w_1}}\hat{S}.
\end{align*}
Note that $\eta P + (1-\eta)R, \eta Q + (1-\eta)S, \hat{R},\hat{S} \in  [T^3, T^4]$, and $(\eta P + (1-\eta)R)_2=\hat{R}_2, (\eta Q + (1-\eta)S)_2=\hat{S}_2$. By the condition stated in the lemma and the proof of \lemmaref{lemma_source2_strict_ind}, the independence property holds for $(\eta P + (1-\eta)R, \eta Q + (1-\eta)S, \hat{R},\hat{S})$. Whenever  $\eta P + (1-\eta)R\not\sim \eta Q + (1-\eta)S$, we know $\hat{P}\not\sim \hat{Q}$, a contradiction. Thus, $\eta P + (1-\eta)R\sim \eta Q + (1-\eta)S$.

The proof for the case with $\eta^{w_1}>\eta>0$ is symmetric. Hence for all $\eta\in (0,1)$, $\eta P + (1-\eta) R \sim \eta Q+ (1-\eta) S$. \end{proof}

Now we extend the local result in \lemmaref{lemma_source2_local} to a bounded set. Recall that $(P,Q,R,S)\in (\lxone \times X_2)^4$ is a proper tuple where $P\sim Q, R\sim S$, $P_2= R_2=\delta_{y_1}, Q_2=S_2=\delta_{y_2}$ with $y_1,y_2\in {\Sigma}^2$. The independence property holds for $(P,Q,R,S)$ trivially if $y_1=y_2$. Without loss of generality, we assume $y_1>y_2$ and $P\succ R$.

Take any $\hat{T}\succ T^1 \succ T^2 \succ (\delta_a,\delta_{y_1})\succ (\delta_a,\delta_{y_2})$ with $a\in X_1$, $\hat{T}_2=T^1_2=T^2_2=\delta_{y_2}$, $T^1_1=\delta_{z_1}$, $T^2_1=\delta_{z_2}$ and $\hat{T}, T^1,T^2\in \hat{\cP}$. As $y_1,y_2\in \Sigma^2$, by \lemmaref{lemma_narrow pre}, we can find $\hat{x}_1$, $\hat{x}_2\in X_1^o$ and $p_1,q_1\in \Pi^1(\hat{x}_1)$,  $p_2,q_2\in \Pi^1(\hat{x}_2)$ such that 
$$(q_1,\delta_{y_1})\sim T^1\prec (p_1,\delta_{y_1}) \hbox{~and~} (q_2,\delta_{y_2})\sim T^1\prec (p_2,\delta_{y_2}).$$

By \lemmaref{lemma_achievable}, we know that 
$$[T^1, (p_1,\delta_{y_1} )]\cap [T^1, (p_2,\delta_{y_2} )]\subseteq \Gamma(\hat{x}_1, y_1)\cap \Gamma(\hat{x}_2, y_2).$$

For any $z\in [z_2,z_1]$, we can choose $\lambda_z$ and $\eta_z\in [0,1]$ such that 
$$\lambda_z(q_1,\delta_{y_1}) + (1-\lambda_z)(\delta_a,\delta_{y_1})\sim (\delta_z,\delta_{y_2})\sim \eta_z(q_2,\delta_{y_2}) + (1-\eta_z)(\delta_a,\delta_{y_2}).$$

By \lemmaref{lemma_narrow pre}, $\lambda_z(p_1,\delta_{y_1}) + (1-\lambda_z)(\delta_a,\delta_{y_1})\succ \lambda_z(q_1,\delta_{y_1}) + (1-\lambda_z)(\delta_a,\delta_{y_1})\sim (\delta_z,\delta_{y_2})$ and $\eta_z(p_2,\delta_{y_2}) + (1-\eta_z)(\delta_a,\delta_{y_2})\succ\eta_z(q_2,\delta_{y_2}) + (1-\eta_z)(\delta_a,\delta_{y_2})\sim (\delta_z,\delta_{y_2})$.

Denote $\hat{x}_1^z$, $\hat{x}_2^z$ for each $z\in [z_1,z_2]$ with 
$$\lambda_z q_1 + (1-\lambda_z)\delta_a\in \Pi^1(\hat{x}_1^z), \hbox{~and~} \eta_z q_2 + (1-\eta_z)\delta_a\in \Pi^1(\hat{x}_2^z).$$
This leads to 
$$[(\delta_z,\delta_{y_2}), (\lambda_z p_1 + (1-\lambda_z)\delta_a,\delta_{y_1})]\cap [(\delta_z,\delta_{y_2}), (\eta_z p_2 + (1-\eta_z)\delta_a,\delta_{y_2})]\subseteq \Gamma(\hat{x}_1^z, y_1)\cap \Gamma(\hat{x}_2^z, y_2).$$
Take the union across all $z$  between $z_1$ and $z_2$, and by Axiom WC, we have 
\begin{equation}
[T^2, T^1]\subseteq \bigcup_{z_2\leq z\leq z_1}\big(\Gamma(\hat{x}_1^z, y_1)\cap \Gamma(\hat{x}_2^z, y_2)\big).    
\end{equation}
    In order to get an open cover of $[T^2,T^1]$, notice that for $\epsilon>0$ small enough with $\hat{T}\succ (\delta_{z_1+\epsilon},\delta_{y_2}) \succ (\delta_{z_2-\epsilon},\delta_{y_2})\succ (\delta_a,\delta_{y_1})$, we have 
$$[T^2, T^1]\subseteq \bigcup_{z_2-\epsilon\leq z\leq z_1+\epsilon}\big(\Gamma(\hat{x}_1^z, y_1)\cap \Gamma(\hat{x}_2^z, y_2)\big).    
$$

For each $z_2-\epsilon\leq z\leq z_1+\epsilon$, $\Gamma(\hat{x}_1^z, y_1)\cap \Gamma(\hat{x}_2^z, y_2)$ has a non-empty interior. Hence we can find an open cover of $[T^2,T^1]=[(\delta_{z_2},\delta_{y_2}),(\delta_{z_1},\delta_{y_2})]$ as $\{C^z\}_{z_2-\epsilon\leq z\leq z_1+\epsilon}$ with $C^z\subset \Gamma(\hat{x}_1^z, y_1)\cap \Gamma(\hat{x}_2^z, y_2)$. Notice that $X_1\times \{\delta_{y_2}\}$ is isomorphic to $X_1\subseteq \bR$ and in the corresponding topology $[T^2,T^1]=[(\delta_{z_2},\delta_{y_2}),(\delta_{z_1},\delta_{y_2})]$ is isomorphic to $[z_2,z_1]$, which is closed and bounded. By Heine–Borel theorem, we can find a finite subcover of $\{C^z\}_{z_2-\epsilon\leq z\leq z_1+\epsilon}$ for $[T^2,T^1]$. Denote the subcover as $\{C^{z_k}\}_{k=1}^K$.

Take any proper tuple $(P,Q,R,S)$ with $P\sim Q, R\sim S$, $P_2= R_2=\delta_{y_1}, Q_2=S_2=\delta_{y_2}$ with $y_1>y_2\in {\Sigma}^2$. By \lemmaref{lemma_source2_local}, the independence property holds for $(P,Q,R,S)$ if $P,Q,R,S\in C^{z_k}$ for any $k=1,...,K$. Then \lemmaref{lemma_source2_union}
implies that the independence property holds for $(P,Q,R,S)$ if $P,Q,R,S\in [T^2,T^1]\subseteq \bigcup_{k=1}^KC^{z_k}$. By arbitrariness of $T^1$, $T^2$ and $a\in X_2$, fix any $\hat{z},\hat{z}'\in X_1$ with $(\delta_{\hat{z}},\delta_{y_2})\succ (\delta_{\hat{z}'}, \delta_{y_1})$, then the given tuple $(P,Q,R,S)$ always satisfies the independence property so long as $(\delta_{\hat{z}},\delta_{y_2})\succ P,Q,R,S\succ(\delta_{\hat{z}'}, \delta_{y_1})$.

There are two gaps between the current argument and a complete proof of  \lemmaref{lemma_source2_ind}. First, we have ruled out the possibility that some lottery in the tuple might be indifferent to $(\delta_{\lcone},\delta_{y_1})$ or  $(\delta_{\cone},\delta_{y_2})$. Second, we have assumed that $y_1,y_2\in \Sigma^2$, instead of its closure. We will bridge the gap by utilizing Axiom WC.

\begin{proof}[Proof of \lemmaref{lemma_source2_ind}]
Following the above arguments, it suffices to consider a tuple $(P,Q,R,S)$ with $P\sim Q\sim (\delta_{z},\delta_{y_2})\succ R\sim S\sim (\delta_{z'},\delta_{y_2})$, $P_2= R_2=\delta_{y_1}, Q_2=S_2=\delta_{y_2}$ where $y_1> y_2\in {\Sigma}^2$ and $z>z'$. We have already shown the case with  $(\delta_{\cone},\delta_{y_2})\succ P\succ R\succ (\delta_{\lcone},\delta_{y_1})$.

 Now suppose $R\sim (\delta_{\lcone},\delta_{y_1})$, where $\lcone>-\infty$. By Axiom M, it must be the case that $R=(\delta_{\lcone},\delta_{y_1})$.  Take a sequence of $\{\lambda_n\}_{n\geq 1} \subset (0,1)$ with $\lambda_n\rightarrow 0$. For each $n$, denote $S^n= (\lambda_n Q_1 + (1-\lambda_n)S_1, S_2)$ and by \lemmaref{lemma_narrow pre}, we can find $\beta_n$ with $R^n = (\beta_n P_1 + (1-\beta_n)R_1, R_2)\sim S^n$. Clearly, $R^n\succ R= (\delta_{\lcone},\delta_{y_1})$ for each $n$ and hence the independence property holds for $(P, Q, R^n, S^n)$, that is, for each $\alpha\in (0,1)$, 
$$\alpha P + (1-\alpha)R^n \sim \alpha Q + (1-\alpha)S^n.$$

 Easy to see that as $n$ goes to infinity, $\beta_n$ converges to 0 and hence $S^n\xrightarrow{w}S, R^n\xrightarrow{w}R$.  By continuity of $\succsim$ on $\hat{\cP}$,  we have 
 for each $\alpha\in (0,1)$, 
$$\alpha P + (1-\alpha)R \sim \alpha Q + (1-\alpha)S.$$

A similar proof works for the case with $P\sim (\delta_{\cone},\delta_{y_2})$ and/or $R\sim (\delta_{\lcone},\delta_{y_1})$. Hence, the independence property holds for all $(P,Q,R,S)$ with $P\sim Q, R\sim S$, $P_2= R_2=\delta_{y_1}, Q_2=S_2=\delta_{y_2}$ with $y_1, y_2\in {\Sigma}^2$.

Now we consider  $y_1\in cl({\Sigma}^2)$ and $y_2\in {\Sigma}^2$.  By definition, we can find a sequence $\{y^n_1\}_{n\geq 1} \subseteq \Sigma^2$ such that $y^n_1\rightarrow y_1$ as $n\rightarrow \infty$. Using the standard subsequence arguments, we further assume that $y^n_1\geq y_1$ for all $n$. The case where $y^n_1\leq y_1$ for all $n$ is symmetric. Denote $P^n = (P_1, \delta_{y^n_1}) \succsim P \sim Q$ and $R^n = (R_1, \delta_{y^n_1}) \succsim R \sim S$. For each $n$, we increase $y_2$ gradually to $y_2'$ until either  $(Q_1,\delta_{y_2'}) \sim P^n$ or $(S_1,\delta_{y_2'}) \sim R^n$. Denote such $y_2'$ as $y^n_2$. Without loss of generality, suppose $Q^n:=(Q_1,\delta_{y_2^n}) \sim P^n$.  Then we can find $S_1^n$ such that  $S^n = (S^n_1,\delta_{y_2^n})\sim R^n$. This is guaranteed by Axiom M and  $P\sim Q\succ R\sim S$. Hence the independence property applies for $(P^n, Q^n,R^n, S^n)$ and for each $\lambda\in (0,1)$, 
$\lambda P^n  + (1-\lambda)R^n \sim  \lambda Q^n  + (1-\lambda)S^n$.  Easy to see that $P^n\xrightarrow{w} P$, $Q^n\xrightarrow{w} Q$, $R^n\xrightarrow{w} R$, $S^n\xrightarrow{w} S$. By continuity of $\succsim$ on $\hat{\cP}$, the independence property holds for $(P,Q,R,S)$. Similarly, the result holds for $y_2\in cl({\Sigma}^2)$ and $y_1\in {\Sigma}^2$.

Finally, assume $y_1>y_2$ with $y_1,y_2\in cl({\Sigma}^2)\backslash {\Sigma}^2$.  Suppose that there exists $y_3\in {\Sigma}^2$ with $y_2<y_3<y_1$.   Since $P,R\in \Gamma_{2,\delta_{y_1}}\cap \Gamma_{2,\delta_{y_2}}$, by Axiom M,  $P,R\in \Gamma_{2,\delta_{y_3}}$. Then there exist $p_1', r_1'$ with $P'=(p_1', \delta_{y_3})\sim Q\sim P$ and $R'=(r_1', \delta_{y_3})\sim R\sim S$. By applying the previous result for $(P,P',R,R')$ and $(Q,P',S,R')$ respectively, we know that the independence property holds for any $(P,Q,R,S)$. Otherwise, we can find a sequence $\{y^n_1\}_{n\geq 1} \subseteq \Sigma^2$ such that $y^n_1\rightarrow y_1$ as $n\rightarrow \infty$ and $y^n_1>y_1$ for all $n$. Then the argument in the previous paragraph follows.

By \lemmaref{lemma_source2_strict_ind}, the independence property holds for any proper tuple $(P,Q,R,S)$ with $P_2= R_2=\delta_{y_1}, Q_2=S_2=\delta_{y_2}$ with $y_1, y_2\in cl({\Sigma}^2)$. This completes the proof. \end{proof}

We are now ready to show that $\succsim$ must admit a GBIB-CN representation. 

 \begin{lemma}\label{lemma_GBIB-CN}
Suppose that Assumption 1 holds. Then $\succsim$ admits a GBIB-CN representation. 
\end{lemma}

\begin{proof}[Proof of \lemmaref{lemma_GBIB-CN}]
 The proof idea is analogue to the proof of \lemmaref{lemma_eu-cn} in \cite{fishburn1982book}. Recall that we can focus on $\succsim$ restricted to $\lxone \times X_2$. For any $(p_1,\delta_x),(p_2,\delta_x)\in \lxone \times cl(\Sigma^2)$ with $(p_1,\delta_x)\succ (p_2,\delta_x)$, we claim that there exists some function $f$ representing $\succsim$ on $(\lxone \times cl(\Sigma^2))\cap [(p_2,\delta_x),(p_1,\delta_x)]$ such that $f$ is continuous and linear in the first source, that is, for any $(q_1,\delta_y), (q_2,\delta_y) \in (\lxone \times cl(\Sigma^2))\cap [(p_2,\delta_x),(p_1,\delta_x)]$ and $\alpha\in [0,1]$, $f(\alpha q_1+ (1-\alpha)q_1,\delta_y) = \alpha f(q_1,\delta_y) + (1-\alpha) f(q_2,\delta_y)$. Also, such $f$ is unique up to a  positive affine transformation. For simplicity, we call $f$ a $MAP_1$ function. 

To prove the claim, notice that by \lemmaref{lemma_axiom implication}, for any $Q\in [(p_2,\delta_x),(p_1,\delta_x)]$, there exists a unique $\lambda_Q\in [0,1]$ such that $Q\sim \lambda_Q (p_1,\delta_x) + (1-\lambda_Q)(p_2,\delta_x)$. Define $f:[(p_2,\delta_x),(p_1,\delta_x)]\rightarrow [0,1]$ such that $f(Q)=\lambda_Q$. As $(p_1,\delta_x)\succ (p_2,\delta_x)$, we have 
$$f(Q)\geq f(Q')\Longleftrightarrow Q\sim  \lambda_Q (p_1,\delta_x) + (1-\lambda_Q)(p_2,\delta_x)\succsim \lambda_{Q'} (p_1,\delta_x) + (1-\lambda_{Q'})(p_2,\delta_x)\sim Q'$$
Hence $f$ represents $\succsim$ on $[(p_2,\delta_x),(p_1,\delta_x)]$. Continuity of $\succsim$ on $\hat{\cP}$ assures that $f$ is continuous. Then we show the linearity of $f$ in source 1 on $(\lxone \times cl(\Sigma^2))\cap [(p_2,\delta_x),(p_1,\delta_x)]$. Take any $(q_1,\delta_y), (q_2,\delta_y) \in (\lxone \times cl(\Sigma^2))\cap [(p_2,\delta_x),(p_1,\delta_x)]$. By definition of $f$, we have 
\begin{align*}
    &(q_1,\delta_y)\sim f(q_1,\delta_y) (p_1,\delta_x) + (1-f(q_1,\delta_y))(p_2,\delta_x),\\
    &(q_2,\delta_y)\sim f(q_2,\delta_y) (p_1,\delta_x) + (1-f(q_2,\delta_y))(p_2,\delta_x).
\end{align*}
Clearly, for any $\alpha\in (0,1)$, $\alpha (q_1,\delta_y) + (1-\alpha)(q_2,\delta_y) \in (\lxone \times cl(\Sigma^2))\cap [(p_2,\delta_x),(p_1,\delta_x)]$. By definition of $f$ and \lemmaref{lemma_source2_ind}, 
\begin{align*}
    \alpha (q_1,\delta_y) + (1-\alpha)(q_2,\delta_y)\sim &[\alpha f(q_1,\delta_y)  + (1-\alpha)f(q_2,\delta_y)](p_1,\delta_x) \\
      & + [1-\alpha f(q_1,\delta_y)  -(1-\alpha)f(q_2,\delta_y)](p_2,\delta_x)\\
     \sim & f(\alpha q_1 +(1-\alpha)q_2,\delta_y) (p_1,\delta_x)\\
     &+(1-f(\alpha q_1 +(1-\alpha)q_2,\delta_y))(p_2,\delta_x)
\end{align*}
By \lemmaref{lemma_narrow pre} given $\delta_x$ in source 2, we know $f(\alpha q_1 +(1-\alpha)q_2,\delta_y)= \alpha f(q_1,\delta_y)  + (1-\alpha)f(q_2,\delta_y)$. Easy to see that a positive affine transformation of $f$ also represents $\succsim$ on $(\lxone \times cl(\Sigma^2))\cap [(p_2,\delta_x),(p_1,\delta_x)]$.

Now suppose that $f,g$ represent  $\succsim$ on $(\lxone \times cl(\Sigma^2))\cap [(p_2,\delta_x),(p_1,\delta_x)]$ and they are continuous and linear in source $1$. Without loss of generality, let $f(p_2,\delta_x)= g(p_2,\delta_x), f(p_1,\delta_x)= g(p_1,\delta_x)$. Recall that for any $Q\in (\lxone \times cl(\Sigma^2))\cap [(p_2,\delta_x),(p_1,\delta_x)]$, there is a unique $\lambda_Q$ with $Q\sim \lambda_Q (p_1,\delta_x) + (1-\lambda_Q)(p_2,\delta_x)$. By linearity of $f$ and $g$ in source 1 on $(\lxone \times cl(\Sigma^2))\cap [(p_2,\delta_x),(p_1,\delta_x)]$, we have 
\begin{align*}
    f(Q)= &\lambda_Qf( p_1,\delta_x) + (1-\lambda_Q)f(p_2,\delta_x)\\
     = & \lambda_Qf( p_1,\delta_x) + (1-\lambda_Q)f(p_2,\delta_x)\\
     =& g(Q)
\end{align*}
Hence $f\equiv g$ on $(\lxone \times cl(\Sigma^2))\cap [(p_2,\delta_x),(p_1,\delta_x)]$ and $f$ is unique up to a positive affine transformation.

As $p^1,p^2$ are arbitrary and the $MAP_1$ function is unique up to a positive affine transformation, for each $x\in cl(\Sigma^2)$, we can find a $MAP_1$ function $f$ that represents $\succsim$ on $(\lxone \times cl(\Sigma^2)) \cap \Gamma_{2,\delta_{x}}$.  Also, as $f$ is unique up to a positive affine transformation on any  $(\lxone \times cl(\Sigma^2))\cap [(p_2,\delta_x),(p_1,\delta_x)]$, $f$ is unique up to a positive affine transformation on $(\lxone \times cl(\Sigma^2)) \cap \Gamma_{2,\delta_{x}}$. 

Now choose $y\in cl(\Sigma^2)$ and $y\neq x$. Denote the $MAP_1$ function on $(\lxone \times cl(\Sigma^2)) \cap \Gamma_{2,\delta_{x}}$ as $f_x$ and the $MAP_1$ function on $(\lxone \times cl(\Sigma^2)) \cap \Gamma_{2,\delta_{y}}$ as $f_y$. If there exist $T^1\succ T^2$ with $T^1,T^2\in \Gamma_{2,\delta_{y}}\cap \Gamma_{2,\delta_{x}}$, then by \lemmaref{lemma_achievable}, we can find $p_1^x,p_2^x,p_1^y,p_2^y$ such that $[T^2,T^1] = [(p^x_2,\delta_x),(p^x_1,\delta_x)] = [(p^y_2,\delta_y),(p^y_1,\delta_y)]$. Since both $f_x, f_y$ are $MAP_1$ functions on $(\lxone \times cl(\Sigma^2))\cap [T^2,T^1]$, they must be positive affine transformations of each other on  $(\lxone \times cl(\Sigma^2))\cap [T^2,T^1]$. Fix $f_x$ and let $f_y(T^i)=f_x(T^i)$, $i=1,2$, then we have $f_x=f_y$ on $(\lxone \times cl(\Sigma^2))\cap [T^2,T^1]$. Define $\hat{f}= f_x$ on $(\lxone \times cl(\Sigma^2)) \cap \Gamma_{2,\delta_{x}}$ and $\hat{f}= f_y$ on $(\lxone \times cl(\Sigma^2)) \cap \Gamma_{2,\delta_{y}}$. Then easy to show that $\hat{f}$ is a $MAP_1$ function on $(\lxone \times cl(\Sigma^2)) \cap (\Gamma_{2,\delta_{x}}\cup \Gamma_{2,\delta_{y}})$ and is unique up to a positive affine transformation. If instead no such $T^1$ and $T^2$ exist, then the construction of a $MAP_1$ function on $(\lxone \times cl(\Sigma^2)) \cap (\Gamma_{2,\delta_{x}}\cup \Gamma_{2,\delta_{y}})$ is trivial. 

By induction, the above arguments can be applied to show the existence of a $MAP_1$ function on $(\lxone \times cl(\Sigma^2)) \cap (\bigcup_{x\in A}\Gamma_{2,\delta_{x}})$ where $A$ is a finite subset of $cl(\Sigma^2)$, and the $MAP_1$ function is unique up to a positive affine transformation. As $A$ is an arbitrary finite subset of  $cl(\Sigma^2)$ and $\bigcup_{x\in X_2}\Gamma_{2,\delta_{x}}=\hat{\cP}$,  we can find a $MAP_1$ function $V$ that represents $\succsim$ on $\lxone \times cl(\Sigma^2)$, which is unique up to a positive affine transformation.

 Define $w(x,y)=V(\delta_x,\delta_y)$ for all $(x,y)\in X_1 \times cl(\Sigma^2)$. Easy to show that $w$ is regular on $X_1 \times cl(\Sigma^2)$. This implies that for any $(p,\delta_{z}),(q,\delta_{z'})\in \lxone \times cl(\Sigma^2)$, 
$$(p,\delta_{z})\succsim (q,\delta_{z'})\Longleftrightarrow \sum_{x} w(x,z)p(x)\geq \sum_{x} w(x,z')q(x). $$

Recall that in the proof of  \lemmaref{lemma_nb}, $\succsim$ restricted to $X$ admits a regular utility representation. That implies we can extend $w$ to $X$ such that $w$ is regular and represents $\succsim$ on $X$. 

Take any $(p,z)\in \lxone \times X_2$. If $z\in cl(\Sigma^2)$, then by continuity of $w$, we can find $a_{(p,z)}\in X_1$ such that $(p,z)\sim ({a_{(p,z)}},z)$, that is, $w(a_{(p,z)},z) = \sum_{x} w(x,z)p(x)$. If $z\not\in cl(\Sigma^2)$, then by definition of $\Sigma^2$, $(p,z)\sim ({CE_{v_{1}}(p)}, z)$. 
By Axiom CN and $(p,q)\sim (p, CE_{v_2}(q))$ for all $p,q\in \lxone$, we know that for any $P\in \cP$, $P\sim (CE_{v_1}(P_1),CE_{v_2}(P_2))$ if $CE_{v_2}(P_2)\not\in cl(\Sigma^2)$
 and $P\sim (a_{(P_1, \delta_{CE_{v_2}(P_2)})},CE_{v_2}(P_2))$ if $CE_{v_2}(P_2)\in cl(\Sigma^2)$. Hence $\succsim$ is represented by 
$$V(P)= \begin{cases}
w(CE_{v_1}(P_1),CE_{v_2}(P_2)), \hbox{~if~} CE_{v_2}(P_2)\not\in cl(\Sigma^2)\\
\sum w(x,CE_{v_2}(P_2))P_1(x), \hbox{~if~} CE_{v_2}(P_2)\in cl(\Sigma^2)\end{cases}$$

Finally, if $CE_{v_2}(P_2)\in \partial \Sigma^2= cl(\Sigma^2)\backslash  \Sigma^2$, then by definition of $\Sigma^2$, $P\sim ({CE_{v_1}(P_1)},{CE_{v_2}(P_2)})$. This implies $w(\cdot, CE_{v_2}(P_2))$ must be a positive affine transformation of $v_1$. Thus, we can rewrite the representation as 
$$V(P)= \begin{cases}
w(CE_{v_1}(P_1),CE_{v_2}(P_2)), \hbox{~if~} CE_{v_2}(P_2)\not\in \Sigma^2\\
\sum w(x,CE_{v_2}(P_2))P_1(x), \hbox{~if~} CE_{v_2}(P_2)\in \Sigma^2\end{cases}$$
and hence $\succsim$ is represented by a GBIB-CN representation $(w,v_1,v_2, \Sigma^2)$, where $w,v_1,v_2$ are regular and $\Sigma^2$ is open in $X_2$ with $0\not\in \Sigma^2$.\end{proof}

\bigskip

{\bf Step 4: Suppose that the DM narrowly brackets marginal lotteries in source 1.} It is equivalent to the assumption that $(p,q)\sim (p',q)$ for all $(x,y)\in X$, $p,p'\in \Pi^1(x)$ and $q\in \Pi^2(y)$. By a symmetric argument to Step 3, we can show that the DM must admit a FBIB-CN representation.

\bigskip

{\bf Step 5: Suppose that the DM does not narrowly bracket marginal lotteries in both sources.} That is, we can find $x_1, x_2\in X_1,y_1,y_2\in X_2$ and $(p_1,q_1), (p_1,q_1')\in \Pi^1(x_1)\times \Pi^2(y_1)$,  $(p_2,q_2), (p_2',q_2)\in \Pi^1(x_2)\times \Pi^2(y_2)$ such that $(p_1,q_1)\succ (p_1,q_1')$ and $(p_2,q_2)\succ (p_2',q_2)$. We denote this condition as {\bf Assumption 2}.

The next lemma shows that we can make $x_1=y_1$ and $x_2=y_2$ in Assumption 2. 

\begin{lemma}\label{lemma_same xy}
Suppose that Assumption 2 holds. Then there exist $(x^o,y^o)\in X_1^o\times X_2^o$ and $P^o, Q^o, R^o, S^o\in \Pi^1(x)\times \Pi^2(y)$ such that  $P^o_1= Q^o_1$, $R^o_2= S^o_2$, $P^o\succ Q^o$ and $R^o\succ S^o$.
\end{lemma}

\begin{proof}[Proof of \lemmaref{lemma_same xy}]
By Assumption 2, there exist $x_1,x_2 \in X_1, y_1,y_2\in X_2$ and $(p_1,q_1), (p_1,q_1')\in \Pi^1(x_1)\times \Pi^2(y_1)$,  $(p_2,q_2), (p_2',q_2)\in \Pi^1(x_2)\times \Pi^2(y_2)$ such that $(p_1,q_1)\succ (p_1,q_1')$ and $(p_2,q_2)\succ (p_2',q_2)$. Clearly, $y_1\in X_2^o,x_2\in X_1^o$. By Axiom WC, we can also assume that $y_2\in X_2^o$ and $x_1\in X^o_1$. By \lemmaref{lemma_narrow pre}, given $p_1$ in source 1, as $y_2\in X_2^o$, we can find $\hat{q}_1,\hat{q}_1'\in \Pi^2(y_2)$ with $(p_1, \hat{q}_1)\succ (p_1, \hat{q}_1')$. By \lemmaref{lemma_narrow pre} given $q_2$ in source 2, as $x_1\in X^o_1$, we can find $\hat{p}_2,\hat{p}_2'\in \Pi^1(x_1)$ with $(\hat{p}_2, q_2)\succ (\hat{p}_2', q_2)$. Let $(x^o,y^o)=(x_1,y_2)$ and $P^o=(p_1, \hat{q}_1)\succ Q^o=(p_1, \hat{q}_1'), R^o=(\hat{p}_2, q_2)\succ S^o=(\hat{p}_2', q_2)$. This completes the proof.\end{proof}

From now on, take $x^o,y^o$ and $(P^o,Q^o, R^o,S^o)$ as given in \lemmaref{lemma_same xy}. Denote the set of all such pairs of $(x^o,y^o)$ as $O$. Similar to $\Sigma^2$, we define $\Sigma^1$ as 
$$\Sigma^1:= \{x\in X_1^o: \exists~ y\in X_2 \hbox{~and~} p\in \Pi^1(x), q,q'\in \Pi^2(y) \hbox{~s.t.~} (p,q) \succ (p,q') \}.$$

One should notice that $\Sigma^i\subseteq X_i^o, i=1,2$. It is possible, for example, that for $x=\cone$, there exists  $y\in X_2$ with $q,q'\in \Pi^2(y)$ and $(\cone,q) \succ (\cone,q')$. However, by Axiom WC, this implies that $(\cone-\epsilon,\cone) \subseteq \Sigma^1$ for some $\epsilon>0$, which implies $\cone \in cl(\Sigma^1)$. This suggests that for each $i=1,2$,  $x\in X_i\backslash cl(\Sigma^i)$, $y\in X_{-i}$,  $p\in \Pi^i(x)$ and $q,q'\in \Pi^{-i}(y)$, we must have $P\sim Q$ where $P_i=Q_i=p$ and $P_{-j}=q,Q_{-j}=q'$.

By the proof of \lemmaref{lemma_same xy}, we can show $O=\Sigma^1\times \Sigma^2$. Also, by continuity of $\succsim$ on $\hat{\cP}$, $\Sigma^1$ is also an open subset of $X_1$ and $0\not\in \Sigma^1$.  For any $x_1\in \Sigma^1$ and $x_2\in \Sigma^2$, we denote 
\begin{align*}
    \hat{\Pi}^1(x_1) &= \big\{ p\in \Pi^1(x_1): \exists~x'_2\in X_2, q\in \Pi^2(x'_2) \hbox{~s.t.~} (p,q)\not\sim (p,{x'_2}) \big\}, \\
    \hat{\Pi}^2(x_2) &= \big\{ q\in \Pi^2(x_2): \exists~x'_1\in X_1, p\in \Pi^1(x'_1) \hbox{~s.t.~} (p,q)\not\sim ({x'_1},q) \big\}.
\end{align*}
Clearly, $\hat{\Pi}^i(x_i)\subseteq {\Pi}^i(x_i)$ for each $i$ by definition. Moreover, with the same argument in Step 3, for any $i=1,2$, $x_i\in \Sigma^i$, $p_i\in  \hat{\Pi}^i(x_i)$ and $x_{-i}\in X_{-i}^o$, we can find $p_{-i}\in \Pi^{-i}(x_{-i})$ such that $T\not\sim T'$ where $T_i=T'_i=p_i$ and $T_{-i}=p_{-i}$, $T'_{-i}={x_{-i}}$.

\begin{lemma}\label{lemma_closure}
For each $i=1,2$ and $x_i\in \Sigma^i$, $cl(\hat{\Pi}^i(x_i))={\Pi}^i(x_i)$.
\end{lemma}
\begin{proof}[Proof of \lemmaref{lemma_closure}]
We will prove the result for $i=1$. The proof for $i=2$ is symmetric and omitted. By definition of $x_1\in \Sigma^1$, we can find $p\in \hat\Pi^1(x_1)$, $x'_2\in X_2$ and $q,q'\in  \Pi^2(x'_2)$ such that $(p,q)\succ (p,q')$. For any $p^o\in {\Pi}^1(x_1)\backslash \hat{\Pi}^1(x_1)$, we have $(p^o,q)\sim (p^o,q')$. By \lemmaref{lemma_axiom implication}, the independence property holds for $((p,q),(p,q'),(p^o,q),(p^o,q'))$ as $p,p^o\in \Pi^1(x_1)$. Then for any $\lambda\in (0,1)$, $(\lambda p + (1-\lambda)p^o,q)\succ (\lambda p + (1-\lambda)p^o,q')$, which implies $\lambda p + (1-\lambda)p^o\in \hat{\Pi}^1(x_1)$. Let $\lambda\rightarrow 0$ and we have $p^o\in cl(\hat{\Pi}^i(x_i))$.
\end{proof}

For any $A_1\subseteq X_1$ and $A_2\subseteq X_2$, we denote $Y_1(A_1)=\cup_{x_1\in A_1}\Pi^1(x_1)$, $Y_2(A_2)=\cup_{x_2\in A_2}\Pi^2(x_2)$ and $Y(A_1, A_2)= Y_1(A_1)\times Y_2(A_2)\subseteq \hat{\cP}$. Our goal is to show that the independence property holds for certain proper tuple $(P,Q,R,S)$ like \lemmaref{lemma_source2_ind}.

\begin{lemma}\label{lemma_both_ind}
Suppose that Assumption 2 holds. Then a proper tuple $(P,Q,R,S)$ satisfies the independence property if $P_i=R_i\in Y_i(cl(\Sigma^i))$, $Q_j=S_j\in Y_j(cl(\Sigma^j))$ for $i,j\in \{1,2\}$.
\end{lemma}

Again, the proof will rely on several intermediate lemmas. Similar to \lemmaref{lemma_source2_strict_ind}, we can  focus on the case where $P\sim Q, R\sim S$ without loss of generality. 

For any $(x_1,x_2)\in \Sigma^1\times \Sigma^2$, by definition and  \lemmaref{lemma_same xy}, there exist $(p,q), (p,q'), (\hat{p},\hat{q}), (\hat{p}', \hat{q})\in \Pi^1(x_1)\times \Pi^2(x_2)$ with $(p,q)\succ (p,q')$ and $(\hat{p},\hat{q})\succ (\hat{p}', \hat{q})$. By \lemmaref{lemma_axiom implication}, we know that a proper tuple $(P,Q,R,S)$ satisfies the independence property if $P,Q,R,S\in \Pi^1(x_1)\times \Pi^2(x_2)$ and $P\sim Q, R\sim S$. The next lemma is analogous to \lemmaref{lemma_source2_local}.

 \begin{lemma}\label{lemma_both_local}
Suppose that Assumption 2 holds. Then a proper tuple $(P,Q,R,S)$ satisfies the independence property if there exist some $i,j\in \{1,2\}$, $x_i\in \Sigma^i$, $x_j\in \Sigma^j$ and $y\in X_{-i}$, $y'\in X_{-j}$ such that $P_i=R_i=p \in  {\Pi}^i(x_i)$, $Q_j=S_j=q \in  {\Pi}^j(x_j)$, $P\sim Q$, $R\sim S$ and $P,R\in \Gamma_{i,p}(y)\cap \Gamma_{j,q}(y')$.
\end{lemma}

\begin{proof}[Proof of \lemmaref{lemma_both_local}]
By \lemmaref{lemma_achievable}, we can find $P',R'$ with $P'_i=R'_i=p$ and $P'_{-i},R'_{-i}\in \Pi^{-i}(y)$ such that $P\sim P'$ and $R\sim R'$. By \lemmaref{lemma_narrow pre}, for any $\lambda\in (0,1)$, $\lambda P+ (1-\lambda)R\sim \lambda P'+ (1-\lambda)R'$. Similarly, we can find $Q',S'$ with $Q'_j=S'_j=q$,  $Q'_{-j},R'_{-j}\in \Pi^{-j}(y')$ and $\lambda Q+ (1-\lambda)S\sim \lambda Q'+ (1-\lambda)S'$ for any $\lambda\in (0,1)$. Easy to check that the conditions in  (ii) of \lemmaref{lemma_axiom implication} hold for the tuple $(P',Q',R',S')$ and hence for any $\lambda\in (0,1)$, 
$$\lambda P +(1-\lambda)R\sim \lambda P'+ (1-\lambda)R'\sim \lambda Q'+ (1-\lambda)S'\sim \lambda Q+ (1-\lambda)S.$$
\end{proof}

The next step aims at extending this local independence property. We start with a lemma similar to \lemmaref{lemma_achievable}. It specifies the sufficient conditions for $\Gamma_{i,p}(y)\cap \Gamma_{j,q}(y')$ to be nonempty and hence \lemmaref{lemma_both_local} can be applied.

\begin{lemma}\label{lemma_achievable2}
Fix $i,j\in \{1,2\}$, $p\in \lxi$, $q\in \lxj$ and $y\in X_{-i}$. For any $P\in \Gamma_{i,p}(y)\cap \Gamma_{j,q}$, then there exists $y'\in X_{-j}$ such that $P\in \Gamma_{j,q}(y')$. 
\end{lemma}
\begin{proof}[Proof of \lemmaref{lemma_achievable2}]
This is by definition of $\Gamma_{j,q}$.
\end{proof}

The following lemma is the counterpart of \lemmaref{lemma_source2_union}, which says that if the independence property holds on two sets of product lotteries respectively, then it also holds on their union.

 \begin{lemma}\label{lemma_both_union}
Suppose that Assumption 2 holds and $(P,Q,R,S)\in \hat{\cP}^4$ is a proper tuple with $P\sim Q$ and $R\sim S$. Fix $i\in \{1,2\}$, $p\in \lxi$, $y\in X_{-i}$ and $T^j\in \hat\cP$ for $j=1,...,4$ with $T^4\succ T^2\succ T^3 \succ T^1$. If the independence property holds for any such $(P,Q,R,S)$ with $\{P,Q,R,S\}\subseteq \Gamma_{i,p}(y)\cap [T^1,T^2]$ or $\{P,Q,R,S\}\subseteq \Gamma_{i,p}(y)\cap  [T^3,T^4]$, then it also holds for any such $(P,Q,R,S)$ with $\{P,Q,R,S\}\subseteq \Gamma_{i,p}(y)\cap [T^1,T^4]$. \end{lemma}

\begin{proof}[Proof of \lemmaref{lemma_both_union}]
The proof can be directly adapted from the proof of \lemmaref{lemma_source2_union} by noting that for any $W\succ W'$, $W,W'\in \Gamma_{i,p}(y)$ implies that  $[W',W]\subseteq \Gamma_{i,p}(y)$. 
\end{proof}

Now we  extend the local result in \lemmaref{lemma_both_local} to a bounded set. Fix $i,j\in \{1,2\}$, $x_i\in\Sigma^i$, $x_j\in\Sigma^j$ and $p\in \hat{\Pi}^i(x_i)$, $q\in \hat{\Pi}^j(x_j)$. Without loss of generality, we assume $i=1$. Denote $\hat{Q}\in \hat{\cP}$ with $\hat{Q}_j=q,\hat{Q}_{-j}=a'$ for some $a'\in X_{-j}$. Similarly, fix $a\in X_2$. 

Take any  $\hat{T}\succ T^1 \succ T^2$ with $\hat{T}, T^1,T^2\in \hat{\cP}$, $\hat{T}_1=T^1_1=T^2_1=p$, $T^1_{2}=\delta_{z_1}$, $T^2_{2}=\delta_{z_2}$ and $T^2\succ \hat{Q}$, $T^2\succ (p,a)$. Then we know $z_1>z_2>0$.  By definition of $p\in \hat{\Pi}^i(x_1)$, Axiom M and Axiom WC , we can find $\up{z}>z_1$, $r_1,r_2\in \Pi^2(\up{z})$ such that $(p,r_1)\succ (p,r_2)$ and $(p,\delta_{\up{z}})\in [(p,r_2),(p,r_1)]$. For any $0<z<\up{z}$ with $(p,\delta_z)\succ \hat{Q}$, we can find $\eta_z\in (0,1)$ with $(p, \eta_z\delta_{\bar{z}} + (1-\eta_z)\delta_a)\sim (p,\delta_z)$.

 Denote $x^{z}$ such that $\eta_z\delta_{\bar{z}} + (1-\eta_z)\delta_a\in \Pi^2(x^z)$. By \lemmaref{lemma_narrow pre}, $(p, \eta_z r_1 + (1-\eta_z)\delta_a)\succ (p, \eta_z r_2 + (1-\eta_z)\delta_a)$ and 
$$(p,\delta_z)\in [(p, \eta_z r_2 + (1-\eta_z)\delta_a),(p, \eta_z r_1 + (1-\eta_z)\delta_a)]\subseteq \Gamma_{1,p}(x^z).$$

By continuity of $\succsim$ on $\hat{\cP}$, as $\bar{z}>z_1$ and $T^2=(p,\delta_{z_2})\succ \hat{Q}$, $T^2=(p,\delta_{z_2})\succ (p,\delta_a)$, we can find $\epsilon>0$ such that 
$$T^2=(p,\delta_{z_2})\succ (p, \eta_{z_2-\epsilon} r_2 + (1-\eta_{z_2-\epsilon})\delta_a)~,~T^1= (p,\delta_{z_1})\prec (p, \eta_{z_1+\epsilon} r_1 + (1-\eta_{z_1+\epsilon})\delta_a).$$

Hence we know that $\{((p, \eta_z r_2 + (1-\eta_z)\delta_a),(p, \eta_z r_1 + (1-\eta_z)\delta_a))\}_{z_1-\epsilon\leq z\leq z_2+\epsilon}$ is an open cover of $[T^2, T^1]=[(p,\delta_{z_2}), (p,\delta_{z_1})]$. Again, by compactness, it admits a finite subcover indexed by $\{z^1,...,z^n\}\subset (z_1-\epsilon,z_2+\epsilon)$. 

Consider a proper tuple $(P,Q,R,S)$ with $P_1=R_1=p\in \hat{\Pi}^i(x_1)$, $Q_j=S_j=q\in \hat{\Pi}^j(x_j)$, $j\in\{1,2\}$ and $P\sim Q, R\sim S$. Fix any $y\in X_{-i}$. For each $k=1,...,n$, by \lemmaref{lemma_both_local}, the independence property holds for $(P,Q,R,S)$ if $$P,R\in \Gamma_{j,q}(y)\cap [(p, \eta_{z^k} r_2 + (1-\eta_{z^k})\delta_a),(p, \eta_{z^k} r_1 + (1-\eta_{z^k})\delta_a)]\subseteq  \Gamma_{j,q}(y)\cap \Gamma_{1,p}(x^{z^k}).$$
 \lemmaref{lemma_both_union} implies that the independence property holds for such $(P,Q,R,S)$ if $P,R\in \Gamma_{j,q}(y)\cap [T^2,T^1]$. 

Then we show that we can get rid of the constraint that $P,R\in [T^2,T^1]$, where there exist $T^1, T^2, \hat{T}$ with $\hat{T}_1=T^1_1=T^2_1=p$ and $T^2\succ \hat{Q}$, $T^2\succ (p,\delta_a)$. The proof is similar to Step 3 by utilizing the arbitrariness of $T^1, T^2, a, a'$ and the continuity of $\succsim$ on $\hat{\cP}$.  Hence we know that the independence property holds for $(P,Q,R,S)$ if $P,R\in \Gamma_{j,q}(y)$ for any $y\in X_2$. 

 Repeat the previous proof technique by varying $y$, and we can extend the above independence property to the following global property. (Recall that our focus on $P\sim Q, R\sim S$ and $i=1$ in the previous analysis is without loss of generality.)
 \begin{lemma}\label{lemma_both_ind_weak}
Suppose that Assumption 2 holds. Then a proper tuple $(P,Q,R,S)$ satisfies the independence property if $P_i=R_i=p\in \hat{\Pi}^i(x_i)$, $Q_j=S_j=q\in  \hat{\Pi}^j(x_j)$ for $x_i\in \Sigma^i$, $x_j\in \Sigma^j$ and $i,j\in \{1,2\}$.
\end{lemma}

We are now ready to prove \lemmaref{lemma_both_ind}. 

\begin{proof}[Proof of \lemmaref{lemma_both_ind}]
\lemmaref{lemma_both_ind_weak} is weaker than \lemmaref{lemma_both_ind} as we have assumed that $p\in \hat{\Pi}^i(x_i)$, $q\in  \hat{\Pi}^j(x_j)$, $x_i\in \Sigma^i$ and $x_j\in \Sigma^j$, instead of their closures ${\Pi}^i(x_i), {\Pi}^j(x_j), cl(\Sigma^i)$ and  $cl(\Sigma^j)$ respectively. However, the proof can be completed using the same continuity arguments in the proof of \lemmaref{lemma_source2_ind}. 
\end{proof}

Recall that $\Sigma^1$ and $\Sigma^2$ are open subsets of $\bR$. The following lemma provides a characterization for a nonempty open set on the  real line. The proof is standard and we include it for completeness.

\begin{lemma}\label{lemma_open set}
Every non-empty open set $I\subseteq \bR$ can be expressed as a countable union of pairwise disjoint open intervals.
\end{lemma}

\begin{proof}[Proof of \lemmaref{lemma_open set}]
As $\bR$ is a complete metric space and $I$ is an open set in $\bR$, for any $x\in I$, there exists an open interval $I_x\subseteq I$ that contains $x$. Denote $a(x)= \inf\{y:(y,x)\subseteq I\}$ and $b(x)= \sup\{z:(x,z)\subseteq I\}$. Clearly, $a(x)<x<b(x)$. 

Denote $J(x)= (a(x), b(x))$  for each $x\in I$. We claim that $J(x)\subseteq I$. To see this,  for arbitrary $\epsilon>0$, as $a(x)$ is an infimum, there exists $z<a(x)+\epsilon$ such that $(z,x)\subseteq I$. This implies $(a(x)+\epsilon, x)\subseteq I$ and hence $(a(x),x)= \bigcup_{n=1}^{+\infty}\big(a(x) + 1/n, x\big)\subseteq I$. By a similar argument, $(x,b(x))\subseteq I$. Hence we have 
$J(x)= (a(x),x)\cup \{x\} \cup (x,b(x))\subseteq I.$

Suppose now that $a(x)=-\infty$ or $b(x)=+\infty$. If both are the case, it must be $I=\bR$ and we are done. In other cases, we observe intervals of the type $K_-(a):=(-\infty, a)$ or $K_+(a)=(a,+\infty)$, both of which are open. Assume that $I$ contains such an interval $K_-(a)$ with $a\not\in I$. By definition, we can always find such a number $a$, and there can be at most two such numbers. For instance, suppose $K_-(a)\subseteq I$. Then $I= (I\cap K_-(a))\cup  (I\cap K_+(a))$. This implies $I$ is open if and only if $I\backslash K_-(a)$ is open. Then it suffices to show that $I\backslash K_-(a)$ can be decomposed as a countable union of pairwise disjoint open intervals. Therefore, without loss of generality, we assume that there is no $x\in I$ with $a(x)=-\infty$ or $b(x)=+\infty$. 

Define a binary relation $\hat{\sim}$ on $I$ by $x \hat{\sim} y$ if and only if $J(x)=J(y)$. Easy to prove that $\hat{\sim}$ is an equivalent relation and  $\hat{\sim}$ partitions $I$. We claim that the equivalent classes are open. To see this, let $x<y$ with $x,y\in I$. When $x\hat\sim y$, we have $x\in J(x)=J(y)\ni y$. Inversely, when $x\in J(y)$, then  $(x,y)\subseteq I$ and hence $a(x)=a(y), b(x)=b(y)$. This implies $J(x)=J(y)$. Thus, the equivalent class of $x$ is exactly $J(x)$. 

Finally, as $J(x)$ is open and nonempty and the set of rational numbers is dense in the real line, each set in the partition of $I$ can be labelled by a rational number and hence the partition is countable. This implies $I=\bigcup_{n\in \bN^*} J_n$ and completes the proof. \end{proof}

For $i=1,2$, since $\Sigma^i\subseteq X_i^o\backslash\{0\}$ is open and nonempty, by \lemmaref{lemma_open set}, we can write $\Sigma^i = \bigcup_{n=1}^{N^i} J^i_n$
where $N^i\in \bN\cup\{+\infty\}$ and $J^i_n=(\dw{b}_n, \up{b}_n)$ with $\dw{b}_n, \up{b}_n\in \bR\cup \{-\infty,+\infty\}$ and $\dw{b}_n \geq \up{b}_{n-1}$ for each $n\leq N^i$. 
\begin{lemma}\label{lemma_mixture_set}
For each $i=1,2$ and $n\leq N^i$, $Y_i(cl(J^i_{n}))$ is a mixture set. 
\end{lemma}
\begin{proof}[Proof of \lemmaref{lemma_mixture_set}]
For any $p,q\in Y_i(cl(J^i_{n}))$, there exist $x_p,x_q\in cl(J^i_{n})$ such that $p\in \Pi^i(x_p)$ and $q\in \Pi^i(x_q)$. Without loss of generality, let $x_p\geq x_q$. By \lemmaref{lemma_narrow pre}, for any $\lambda\in (0,1)$, $\lambda p + (1-\lambda)q\in \Pi^i(x')$ with $x'\in [x_q,x_p]$. As $cl(J^i_{n_i})$ is a closed interval and $x_p,x_q\in cl(J^i_{n_i})$, $x'\in cl(J^i_{n_i})$. This implies $\lambda p + (1-\lambda)q\in Y_i(cl(J^i_{n}))$ and hence $Y_i(cl(J^i_{n}))$ is a mixture set.
\end{proof}

Then for each $n_1\leq N^1,n_2\leq N^2$, $Y(cl(J^1_{n_1}), cl(J^2_{n_2}))=Y_1(cl(J^1_{n_1}))\times Y_2(cl(J^2_{n_2}))$ is the product of two mixture sets. Also, $\succsim$ restricted to $Y(cl(J^1_{n_1}),cl(J^2_{n_2}))$ is continuous and satisfies Axiom MI by \lemmaref{lemma_both_ind}. By Theorem 1 in Chapter 7.2 (Page 88) of \cite{fishburn1982book}, we know that there exists a continuous and multilinear\footnote{Suppose $\mathcal{M}_1, \mathcal{M}_2 \subseteq \lr$ are mixture sets. A function $V$ is multilinear on $\mathcal{M}_1\times \mathcal{M}_2$ if $V(\lambda p + (1-\lambda)r, q) = \lambda V(p,q) + (1-\lambda) V(r,q)$ and $V(p,\lambda q + (1-\lambda)s) = \lambda V(p,q) + (1-\lambda) V(p,s)$ for all $\lambda \in (0,1)$, $p,r\in \mathcal{M}_1$ and $q,s\in \mathcal{M}_2$.} representation $V_{n_1,n_2}^{EU}$  of $\succsim$ on $Y(cl(J^1_{n_1}),cl(J^2_{n_2}))$, which is unique up to a positive affine transformation.

 We claim that $cl(\Sigma^i)=X_i$ for $i=1,2$.  Suppose by contradiction that $cl(\Sigma^i)\neq X_i$ for some $i$.  Then we can find $a_1<a_2<a_3$ such that either $(a_1,a_2)\subseteq X_i\backslash cl(\Sigma^i)$, $(a_2,a_3)\subseteq cl(\Sigma^{i})$ or $(a_2,a_3)\subseteq X_i\backslash cl(\Sigma^i)$, $(a_1,a_2)\subseteq cl(\Sigma^{i})$.   By symmetry, we will focus on the case where $i=2$ and  $(a_1,a_2)\subseteq X_2\backslash cl(\Sigma^2)$, $(a_2,a_3)\subseteq cl(\Sigma^{2})$.    We can further assume $(a_2,a_3)\subseteq cl(J^2_{n_2})$ for some $n_2$. Choose some $n_1\leq N^1$ and denote $J^1_{n_1}=(b_1,b_2)$ with $b_1<b_2$. We know that there exists a multilinear representation $V^{EU}$ for $\succsim$ on $Y([b_1,b_2], [a_2,a_3])$. 

 We first focus on the preference $\succsim$ restricted to $Y([b_1,b_2], [a_1,a_2])$. By definition, for any $P\in Y(cl(\Sigma^1), X_2\backslash cl(\Sigma^2))$, $P=(P_1,P_2)\sim ({CE_{v_1}(P_1)},P_2)$. \lemmaref{lemma_both_ind} guarantees that independence property holds for a proper tuple $(P,Q,R,S)$ where $P_1=R_1\in \Pi^1(x)$, $Q_1=S_1\in \Pi^1(x')$ with $x,x'\in [b_1,b_2]\subseteq cl(\Sigma^1)$. Then by a similar argument in {\bf Step 4}, there exists a continuous representation $V^{FIB}$ of $\succsim$ on $Y([b_1,b_2], [a_1,a_2])$ where $V^{FIB}(P_1,P_2)= V^{FIB}(\delta_{CE_{v_1}(P_1)},P_2)$ for each $(P_1,P_2)\in Y([b_1,b_2], [a_1,a_2])$,  $V^{FIB}$ is linear in the second source (i.e., a $MAP_2$ function) and unique up to a positive affine transformation.

  For each $b\in (b_1,b_2)$ and $p\in \Pi^1(b)$, by \lemmaref{lemma_narrow pre}, $\succsim_{2|p}$ on $\{p\}\times \lxtwo$ admits an EU representation with a regular utility index $v_{2|p}$. When there is no confusion, we also denote $v_{2|p}$ as the EU function.  The next lemma relates $v_{2|p}$ with $V^{FIB}$.

 \begin{lemma}\label{lemma_p.a.t.}
$\forall~b\in (b_1,b_2)$ and $p\in \Pi^1(b)$, $v_{2|p}$ is a positive affine transformation of $V^{FIB}(\delta_b,\cdot)$ on $Y^2([a_1,a_2])$.
\end{lemma}
\begin{proof}[Proof of \lemmaref{lemma_p.a.t.}]
By definition, $V^{FIB}(p,q)= V^{FIB}(\delta_b,q)$ for all $p\in \Pi^1(b)$, $q\in Y^2([a_1,a_2])$. Then it suffices to show that $v_{2|p}$ is a positive affine transformation of $V^{FIB}(p,\cdot)$ on $Y^2([a_1,a_2])$. Suppose, without loss of generality, that $v_{2|p}(a_i) = V^{FIB}(p,\delta_{a_i})$ for $i=1,2$. For any $q\in Y^2([a_1,a_2])$, if $\delta_{a_2}\succsim_{2|p} q\succsim_{2|p} \delta_{a_1}$, then there exists a unique $\lambda\in [0,1]$ such that $\lambda\delta_{a_2}+(1-\lambda)\delta_{a_1}\sim_{2|p} q$ and hence
\begin{align*}
    v_{2|p}(q) &= v_{2|p}(\lambda\delta_{a_2}+(1-\lambda)\delta_{a_1}) \\
    & = \lambda v_{2|p}(\delta_{a_2}) + (1-\lambda) v_{2|p}(\delta_{a_1})\\
    & =  \lambda V^{FIB}(p,\delta_{a_2}) + (1-\lambda)  V^{FIB}(p,\delta_{a_1})\\
    & = V^{FIB}(p,\lambda\delta_{a_2}+(1-\lambda)\delta_{a_1})\\
    & =  V^{FIB}(p,q).
\end{align*}
Similar arguments hold for $q\succ_{2|p} \delta_{a_2}$ and $\delta_{a_1}\succ_{2|p} q$. This completes the proof.
\end{proof}
 
 A direct corollary is that for all $~b\in (b_1,b_2)$ and $p\in \Pi^1(b)$, $v_{2|p}$ is a positive affine transformation of $v_{2|\delta_b}$ on $Y^2([a_1,a_2])$. We further claim that it holds on  $\lxtwo$. 

\begin{lemma}\label{lemma_p.a.t.2}
$\forall~b\in (b_1,b_2)$ and $p\in \Pi^1(b)$, $v_{2|p}$ is a positive affine transformation of $v_{2|\delta_b}$.
\end{lemma}

\begin{proof}[Proof of \lemmaref{lemma_p.a.t.2}]
By the corollary of \lemmaref{lemma_p.a.t.}, given $b\in (b_1,b_2)$ and $p\in \Pi^1(b)$,  there exist $\alpha_p>0$ and $\beta_p$ such that $v_{2|p}(q)=\alpha_p v_{2|\delta_b}(q) + \beta_p$ for all $q\in Y^2([a_1,a_2])$. 

Now consider $q\in \lxtwo \backslash Y^2([a_1,a_2])$.  If $q\succ_2 \delta_{a_2}$, then we can find $\lambda>0$ such that $\lambda q + (1-\lambda) \delta_{a_1}\in Y^2([a_1,a_2])$. This implies $v_{2|p}(\lambda q + (1-\lambda) \delta_{a_1}) = \alpha_p v_{2|\delta_b}(\lambda q + (1-\lambda) \delta_{a_1}) + \beta_p$. By linearity of $v_{2|p}$ and $v_{2|\delta_b}$,  we have
$$\lambda v_{2|p}(q) +(1-\lambda)\lambda v_{2|p} (\delta_{a_1}) = \lambda [\alpha_p v_{2|\delta_b}(q) + \beta_p] + (1-\lambda)[\alpha_p v_{2|\delta_b}(\delta_{a_1}) + \beta_p].$$
As $v_{2|p}(\delta_{a_1}) = \alpha_p v_{2|\delta_b}( \delta_{a_1}) + \beta_p$ and $\lambda>0$, we know $v_{2|p}(q) = \alpha_p v_{2|\delta_b}(q) + \beta_p$. Similar results can be shown for  $q\prec_2 \delta_{a_1}$. Thus $v_{2|p}(q) = \alpha_p v_{2|\delta_b}(q) + \beta_p$ for all $q\in \lxtwo$. 
\end{proof}

 Now we turn to $Y([b_1,b_2], [a_2,a_3])$, on which $V^{EU}$ represents $\succsim$. By a similar argument as \lemmaref{lemma_p.a.t.}, for any $b\in (b_1,b_2)$ and $p\in \Pi^1(b)$, $V^{EU}(p,\cdot)$ is a positive affine transformation of $v_{2|p}$ on $Y^2([a_2,a_3])$ and $v_{2|p}$ is a positive affine transformation of $v_{2|\delta_b}$ on $Y^2([a_2,a_3])$.  Denote $V^{EU}(p,q)=\hat{\alpha}_p v_{2|\delta_b}(q)+\hat{\beta}_p$   with $\hat{\alpha}_p>0$ and $\hat{\beta}_p\in \bR$ for each $q\in Y^2([a_2,a_3])$.  Notice that $a_2\in (a_1,a_3)$. As $X_2$ is a closed interval, $a_2\in X_2^o$. Also, $(a_1,a_2)\subseteq X_2\backslash cl(\Sigma^2)$ and $(a_2,a_3)\subseteq  cl(\Sigma^2)$ imply that $a_2\not\in \Sigma^2$. Then for each $q\in \Pi^2(a_2)$, we have $(p,q)\sim (\delta_b,q)$ and thus
$$ \hat{\alpha}_p v_{2|\delta_b}(q)+\hat{\beta}_p= V^{EU}(p,q)= V^{EU}(\delta_b,q) =  \hat{\alpha}_{\delta_b} v_{2|\delta_b}(q)+\hat{\beta}_{\delta_b}.$$
As $b\in \Sigma^1$ and $a_2\in X_2^o$, there exist $q_1,q_2\in \Pi^2(a_2)$ such that $(\delta_b,q_1)\succ (\delta_b,q_2)$. Hence, 
$$\hat{\alpha}_p v_{2|\delta_b}(q_1)+\hat{\beta}_p= \hat{\alpha}_{\delta_b} v_{2|\delta_b}(q_1)+\hat{\beta}_{\delta_b}~,~\hat{\alpha}_p v_{2|\delta_b}(q_2)+\hat{\beta}_p= \hat{\alpha}_{\delta_b} v_{2|\delta_b}(q_2)+\hat{\beta}_{\delta_b}.$$
This implies $\hat{\alpha}_p = \hat{\alpha}_{\delta_b}, \hat{\beta}_p=\hat{\beta}_{\delta_b}$ for all $p\in \Pi^1(b)$. Then for all $p\in \Pi^1(b)$ and $q\in Y^2([a_2,a_3]), $
$$V^{EU}(p,q)= \hat{\alpha}_{\delta_b} v_{2|\delta_b}(q)+\hat{\beta}_{\delta_b}= V^{EU}(\delta_b,q).$$

That is, $(p,q)\sim (\delta_b,q)$ for all $p\in \Pi^1(b)$ and $q\in Y^2([a_2,a_3])$, which suggests that $(a_2+a_3)/2\not\in\Sigma^2$, a contradiction with $(a_2,a_3)\subseteq \Sigma^2$. To conclude, $cl(\Sigma^i) =X_i$ for $i=1,2$ and hence \lemmaref{lemma_both_ind} implies Axiom MI. By \lemmaref{lemma_eu-cn}, $\succsim$ admits an EU-CN representation.

To summarize, as NB is a special case of GBIB-CN (FBIB-CN), we conclude that under the axioms stated in the theorem, $\succsim$ admits one of the following representations: EU-CN, GBIB-CN and GFIB-CN. This completes the proof for sufficiency.\end{proof}

\bigskip

\begin{proof}[Proof of \thmref{thm_BIB}] ~\\
 {$ii)\Rightarrow i)$.} We first prove the necessity of these axioms. 
 First, it is easy to verify that EU satisfies all the axioms.  By  \thmref{thm_cn}, we know representations EU-CN, GBIB-CN and GFIB-CN satisfy Axioms WO, M, WC,  WI, and they trivially satisfy Axiom CC as its primitive will never be satisfied. 
 
  For BIB, as it reduces to a special case of GBIB-CN on $\hat{\cP}$, it suffices to show that BIB satisfies the first two parts of Axiom WC, Axiom CC and Axiom M.

 Suppose that $\succsim$ admits a BIB representation $(w,v_2)$, that is, for $P\in\cP$, 
 $$V^{BIB}(P)  = \sum_{x} w(x,CE_{v_{2}}(P_{2|x}))P_1(x).$$
To verify part (i) of Axiom WC, for any $P,Q\in \cP$ and $\lambda\in [0,1]$, 
\begin{align*}
V^{BIB}(\lambda P +(1-\lambda)Q) =&\lambda \sum_{\substack{x: P_1(x)>0,\\Q_1(x)=0}} w(x, CE_{v_2}(P_{2|x}))P_1(x)  \\
&+ (1-\lambda) \sum_{\substack{x: Q_1(x)>0,\\P_1(x)=0}} w(x, CE_{v_2}(Q_{2|x}))Q_1(x)\\
	&+ \sum_{\substack{x: Q_1(x)>0,\\P_1(x)>0}} w(x, CE_{v_2}(\alpha P_{2|x} + (1-\alpha) Q_{2|x}))[\lambda P_1(x)+ (1-\lambda)Q_1(x)],
\end{align*}
where 
$$\alpha =  \frac{\lambda P_1(x)}{\lambda P_1(x)+ (1-\lambda) Q_1(x)}.$$
Then $V^{BIB}(\lambda P +(1-\lambda)Q) $ is continuous in $\lambda$ and mixture continuity holds for $\succsim$ on $\cP$.

To verify the third part of Axiom WC, i.e., Axiom Continuity over Sure Gains, for each $P\in \cP$ and any two sequences $\epsilon_n,\epsilon_n'\rightarrow 0$ as $n\rightarrow \infty$ such that for each $n$, $\epsilon_n, \epsilon_n'>0$.  Since $P$ is a simple lottery, for $n$ large enough, we can guarantee that $P*(\delta_{\epsilon_n}, \delta_{\epsilon_n'})(x+\epsilon_n, y + {\epsilon_n'})=P(x,y)$ for all $x\in X_1^o, y\in X_2^o$. For such $n$, we have 
\begin{align*}
V^{BIB}(P*(\delta_{\epsilon_n}, \delta_{\epsilon_n'})) = &\sum_{x\in X_i^o} w(x+\epsilon_n,CE_{v_{2}}(P_{2|x}* \delta_{\epsilon_n'}))P_1(x)\\ 
&+ w(\cone,CE_{v_{2}}(P_{2|x}* \delta_{\epsilon_n'}))P_1(\cone) .
\end{align*}
If $\cone=+\infty$, then the second term is always $0$.  Notice that $P_{2|x}* \delta_{\epsilon_n'}\xrightarrow[]{w} P_{2|x}$ as $n$ goes to infinity. By continuity of $w$ and $v_2$, easy to see that $V^{BIB}(P*(\delta_{\epsilon_n}, \delta_{\epsilon_n'}))$ is continuous in $(\epsilon_n, \epsilon_n')$ and hence Axiom Continuity over Sure Gains holds for $\succsim$.

Then we will check Axiom CC. For each $P,Q,R,S\in\mathcal{P}$ and $\alpha\in(0,1)$, if $P_i=Q_i$ for $i=1,2$ and $\supp(P_1)\cap \supp(R_1)= \supp(P_1)\cap \supp(S_1)=\emptyset$, then 
$$V^{BIB}(\alpha P + (1-\alpha) R) = \alpha V^{BIB}(P) + (1-\alpha)V^{BIB}(R),$$
$$V^{BIB}(\alpha Q + (1-\alpha) S) = \alpha V^{BIB}(Q) + (1-\alpha)V^{BIB}(S).$$
Hence $P\succ Q, R\sim S$ implies that $\alpha P + (1-\alpha) R\succ \alpha Q + (1-\alpha) S$. 

Finally for Axiom M, suppose that $P$ dominates $(\delta_{y_1},\delta_{y_2})$, then for each $x\in \supp(P_1)$, $x\geq y_1$ and $P_{2|x} \succsim_{FOSD} \delta_{y_2}$ and there exists $x'\in \supp(P_1)$  with $x'>y_1$ or $P_{2|x'} \succ_{FOSD} \delta_{y_2}$. Then by regularity of $w$ and $v_2$, we know 
\begin{align*}
V^{BIB}(P) &= w(x',CE_{v_2}(P_{2|x'}))P_1(x') +  \sum_{x\neq x'} w(x,CE_{v_2}(P_{2|x}))P_1(x)\\
&>w(y_1,y_2)P_1(x') + \sum_{x\neq x'} w(y_1,y_2)P_1(x)\\
&=w(y_1,y_2)= V^{BIB}(\delta_{y_1}, \delta_{y_2}).
\end{align*}
This implies $P\succ (\delta_{y_1}, \delta_{y_2})$. Similarly, we can show that $P\prec (\delta_{x_1}, \delta_{x_2})$ if $P$ is dominated by $(\delta_{y_1},\delta_{y_2})$. This completes the proof for necessity of axioms.

\bigskip

\noindent {$i)\Rightarrow ii)$.} The proof of sufficiency is decomposed in the following steps. In {\bf Step 1}, we restrict our attention to the set of product lotteries $\hat{\cP}$ and apply \thmref{thm_cn}. {\bf Step 2} studies the implications of Axiom CC and Axiom M. In {\bf Step 3}, we derive a KP-style representation on the space of lotteries $\cP$. In {\bf Step 4}, we utilize the consistency of the two representations in {\bf Step 1} and {\bf Step 3} on $\hat{\cP}$ to finish the proof.

\bigskip

{\bf Step 1: We restrict $\succsim$ to $\hat{\cP}$.} 
For each preference $\succsim$ that satisfies the axioms stated in \thmref{thm_BIB}, we can define $\hat{\succsim}$ which satisfies Axiom CN and agrees with $\succsim$ on the set of product lotteries  $\hat{\cP}$. Easy to verify that $\hat{\succsim}$ also satisfies Axiom WO, Axiom M, Axiom WC and Axiom WI. By \thmref{thm_cn}, we know that $\hat{\succsim}$ admits one of the following representations: EU-CN, GBIB-CN and GFIB-CN. This implies that the restriction of $\succsim$ on $\hat{\cP}$ admits one of these representations. Furthermore, if Axiom CN holds, then ${\succsim}$ admits one of the following representations: EU-CN, GBIB-CN and GFIB-CN.

\bigskip

{\bf Step 2: We derive implications of Axiom CC and Axiom M.} 

Suppose, from now on, that Axiom CN does not hold, that is, there exists $P\succ \tilde{P}$ with $P_i =  \tilde{P}_i$ for $i=1,2$. For any $(p,q)\in \lxone\times \lxtwo$, denote $M(p,q)$ as the set of lotteries whose marginal lotteries are $p$ and $q$ respectively.  For any $P,R\in \cP$, we say $P$ and $R$ are {\it compatible}, or $P$ is {\it compatible} with $R$  if $\supp (P_1)\cap \supp (R_1)=\emptyset$. Easy to see that if $R$ is compatible with both $P$ and $Q$, then $R$ is also compatible with $\lambda P + (1-\lambda)Q$ for any $\lambda\in (0,1)$. Also, if $P$ is compatible with $Q$, then $P$ is compatible with all $Q'\in M(Q_1,Q_2)$.

One main difficulty is that betweenness does not hold, that is, for $\lambda\in (0,1)$, it is not guaranteed that $P\succ \lambda P +(1-\lambda)\tilde{P}\succ\tilde P$. However, we have the following weaker and local version of the betweenness property.

\begin{lemma}\label{lemma_betweenness}
For any $Q\succ Q'$, there exists $Q^*=\lambda^* Q+ (1-\lambda^*)Q'$ for some $\lambda^*\in [0,1]$ such that for any $\epsilon>0$, we can find $\lambda_{\epsilon}\in (\lambda^*-\epsilon, \lambda^*+\epsilon)\cap [0,1]$ with $Q^*\not\sim \lambda_{\epsilon} Q + (1-\lambda_{\epsilon})Q'$.
\end{lemma}

\begin{proof}[Proof of \lemmaref{lemma_betweenness}]
Suppose the result fails. Then for any $\lambda\in[0,1]$, there exists $\epsilon_{\lambda}>0$ such that for any $\lambda'\in (\lambda - \epsilon_{\lambda}, \lambda +\epsilon_{\lambda})\cap [0,1]$, $\lambda Q +(1-\lambda)Q'\sim \lambda' Q +(1-\lambda')Q'$. Notice that $\{(\lambda - \epsilon_{\lambda}, \lambda +\epsilon_{\lambda})\}_{\lambda\in [0,1]}$ forms an open cover of the compact set $[0,1]$. We can find a finite subcover of $[0,1]$. By transitivity of $\succsim$, we know that $\lambda Q +(1-\lambda)Q'\sim \lambda' Q +(1-\lambda')Q'$ for all $\lambda,\lambda'\in [0,1]$, which leads to $Q\sim Q'$ and a contradiction.\end{proof}

For $P\succ \tilde{P}$ with $\tilde{P}\in M(P_1,P_2)$, denote $P^*=\lambda^* P + (1-\lambda^*)\tilde{P}$ as the lottery found in \lemmaref{lemma_betweenness}. Clearly, either $P^*\not\sim P$ or $P^*\not\sim \tilde{P}$. Also, $P_1,P_2$ are not degenerate.  By \lemmaref{lemma_betweenness}, for any $n>0$, there exists $\lambda_n\in (\lambda^*-1/n, \lambda^*+1/n)\cap [0,1]$ with $P^*\not\sim \lambda_n P+(1-\lambda_n)\tilde{P}:=P^n$. By completeness, for each $n$, either $P^n\succ P^*$ or $P^*\succ P^n$. Then we can find a subsequence of $\{P^n\}$ (still denoted as $\{P^n\}$ when there is no confusion) such that either $P^n\succ P^*$ for all $n$ or $P^*\succ P^n$ for all $n$. Suppose that the former case holds. Take any $R\sim S$ and $R,S$ compatible with $P$. Axiom CC implies that for all $\alpha\in (0,1)$ and $n\geq 1$, $\alpha P^n + (1-\alpha) R\succ \alpha P^* + (1-\alpha) S$. By mixture continuity of $\succsim$ (the second part of Axiom WC), as $n$ goes to infinity, that is, $\lambda_n$ goes to $\lambda^*$, we have $\alpha P^* + (1-\alpha) R\succsim \alpha P^* + (1-\alpha) S$. This holds for all $R\sim S$ with $R,S$ compatible with $P$. By symmetry, we can just change the place of $R$ and $S$, and get $\alpha P^* + (1-\alpha) S\succsim \alpha P^* + (1-\alpha) R$. Thus, for all $\alpha\in (0,1)$ and $R\sim S$ with $R,S$ compatible with $P$, 
$$\alpha P^* + (1-\alpha) S\sim \alpha P^* + (1-\alpha) R.$$
If instead $P^*\succ P^n$ for all $n$, then the same result holds as the conclusion is an indifference relation. Without loss of generality, we assume that $P^n\succ P^*$ for all $n$ and $P^*\succ \tilde{P}$ from now on.

Fix any $Q$ compatible with $P$ and we know $Q$ is also compatible  with $\tilde{P}$,  $P^*$ and $P^n$ for each $n$. By Axiom CC, for any $\beta\in (0,1)$, $\beta P^* +(1-\beta)Q \succ \beta \tilde{P} +(1-\beta)Q$ and $\beta P^* +(1-\beta)Q, \beta \tilde{P} +(1-\beta)Q \in M(\beta P_1+(1-\beta)Q_1, \beta P_2+(1-\beta)Q_2)$. Similarly, as $P^n\succ P^*$ for all $n$, for any $\beta\in (0,1)$, $\beta P^n +(1-\beta)Q \succ \beta P^* +(1-\beta)Q$ and $\beta P^n +(1-\beta)Q\in M(\beta P_1+(1-\beta)Q_1, \beta P_2+(1-\beta)Q_2)$. For any $R\sim S$ with $R,S$ compatible with both $P,Q$, we know $R,S$ are also compatible with $\beta P^n +(1-\beta)Q$ and $\beta P^* +(1-\beta)Q$. With the same arguments as above, we can show that for any $\alpha\in (0,1)$ and $\beta\in (0,1)$,
$$\alpha[\beta P^* +(1-\beta)Q] + (1-\alpha)R\sim \alpha[\beta P^* +(1-\beta)Q] + (1-\alpha)S.$$
This can be rearranged as 
$$\beta[\alpha P^*+(1-\alpha)R]  + (1-\beta)[\alpha Q+(1-\alpha)R] \sim \beta[\alpha P^*+(1-\alpha)S]  + (1-\beta)[\alpha Q+(1-\alpha)S]  $$
Again by mixture continuity of $\succsim$, let $\beta\rightarrow 0^+$ and we have 
\begin{equation}\label{eq_thm2_e1}
    \alpha Q+(1-\alpha)R\sim    \alpha Q+(1-\alpha)S,
\end{equation}
for any $\alpha\in (0,1)$, $R\sim S$, $Q$  compatible with $P,R,S$ and $P$ compatible with $Q,R,S$.

Fix $P,\tilde{P}$ and $Q$ such that $P$ is compatible with $Q$, we want to strengthen property (\ref{eq_thm2_e1}) by discarding the constraint that $R,S$ are compatible with $P$. By the third part of Axiom WC, as $P\succ \tilde{P}$, we can find $\bar{\epsilon}>0$ such that for all $\epsilon\in (0,\bar{\epsilon})$ and $P_{\epsilon}= P *(\delta_{\epsilon},\delta_0)$,  $\tilde{P}_{\epsilon}= \tilde{P} *(\delta_{\epsilon},\delta_0)$, we have $P_{\epsilon}\succ \tilde{P}_{\epsilon}$. Note hat $\tilde{P}_{\epsilon},{P}_{\epsilon}\in M(P_1 *\delta_{\epsilon}, P_2)$. Since $\supp(P_1)\cup \supp(Q_1)$ is finite, we can make $\bar{\epsilon}$ small enough such that for all $\epsilon\in (0,\bar{\epsilon})$, $\supp(P_1 *\delta_{\epsilon})\cap \supp(Q_1)=\emptyset$ and $\tilde{P}_{\epsilon},{P}_{\epsilon}$ are compatible with $Q$. Then  any $Q$ compatible with $P$,
\begin{equation*}
    \alpha Q+(1-\alpha)R\sim    \alpha Q+(1-\alpha)S,
\end{equation*}
for any $\epsilon\in (0,\bar{\epsilon})$, $\alpha\in (0,1)$, $R\sim S$, $Q$  compatible with $R,S$ and $P_{\epsilon}$ compatible with $R,S$.

Now we show that by varying $\epsilon$, we can further get rid of the constraint that $R,S$ are compatible with $P_{\epsilon}$ for some $\epsilon\in (0, \up{\epsilon})$. This is again guaranteed by the fact that each lottery in $\cP$ has a finite support. Concretely, for any $R\sim S$ with $R,S$ compatible with $Q$, we can always find $\epsilon^*$ such that $R,S$ are compatible with $P_{\epsilon^*}$. Thus,
\begin{equation*}
    \alpha Q+(1-\alpha)R\sim    \alpha Q+(1-\alpha)S,
\end{equation*}
for any  $\alpha\in (0,1)$, $R\sim S$ and $Q$  compatible with $R,S,P$. 

The same argument can be applied to relax the requirement that $Q$ is compatible with $P$ and hence we end up with the result that for any $Q\in \cP$,
\begin{equation}\label{eq_thm2_e2}
    \alpha Q+(1-\alpha)R\sim    \alpha Q+(1-\alpha)S,
\end{equation}
for any $\alpha\in (0,1)$, $R\sim S$, $Q$  compatible with $R,S$.

\bigskip

For each $y\in X_2$, recall that 
$$\Gamma_{2,\delta_y}=\bigcup_{x_1\in X_1} \Gamma_{2,\delta_y}(x_1) =\bigcup_{x_1\in X_1} \bigcup_{\substack{P,Q\in  \Pi^1(x_1)\times \{\delta_y\},\\ P\succsim Q}}[Q,P]\subseteq \hat{\cP}.$$  

We define $\Phi_{2,\delta_y}\subseteq \cP$ such that for each $y\in X_2$,
$$\Phi_{2,\delta_y}= \big\{ P\in \cP:  \exists~ T,T'\in \Gamma_{2,\delta_y} ~s.t.~  T\succ P\succ T' \big\};$$

\begin{lemma}\label{lemma_achievable3}
(i). For each $P,Q,R\in \cP$ with $P\succ Q \succ R$, there exists $\lambda\in (0,1)$ such that $\lambda P +(1-\lambda)R\sim Q$.
\\
 (ii). For each $P\in \cP$, there exists $(x_1,x_2)\in X_1\times X_2$ such that $P\sim (\delta_{x_1}, \delta_{x_2})$. Moreover, if $P\in \Phi_{2,\delta_y}$ for some $y\in X_2$, then we can choose $x_2=y$. 
\end{lemma}

\begin{proof}[Proof of \lemmaref{lemma_achievable3}]
(i). Denote $A=\{\alpha\in (0,1): \alpha P + (1-\alpha)R\succ Q\}$ and $\lambda= inf A$. We claim that $\lambda P +(1-\lambda)R\sim Q$. Suppose by contradiction that $\lambda P +(1-\lambda)R\not\sim Q$. If $\lambda P +(1-\lambda)R\succ Q$, then $\lambda \in A$, which is open by mixture continuity of $\succsim$. Hence there exists $\lambda'<\lambda$ with $\lambda'\in A$, which contradicts with the definition of $\lambda$. If $\lambda P +(1-\lambda)R\prec Q$, then $\lambda \in \{\alpha\in (0,1): \alpha P + (1-\alpha)R\prec Q\}$, which is also open. We can find $\epsilon>0$ such that  $ [\lambda, \lambda+\epsilon)\subseteq (0,1)\backslash A$. Again a contradiction with the definition of $\lambda$. Hence $\lambda P +(1-\lambda)R\sim Q$.
\\
(ii). For each $P\in \cP$, denote $x_i=\max\supp(P_i),y_i=\min\supp(P_i)$ for $i=1,2$. By Axiom M, $(\delta_{x_1},\delta_{x_2})\succsim P \succsim (\delta_{y_1},\delta_{y_2})$. Then either $(\delta_{x_1},\delta_{x_2})\succsim P \succsim (\delta_{x_1},\delta_{y_2})$ or $(\delta_{x_1},\delta_{y_2})\succsim P \succsim (\delta_{y_1},\delta_{y_2})$. By symmetry, suppose the former case holds. Using the same argument as the proof of part (i) in \lemmaref{lemma_axiom implication}, we can find $\lambda\in [0,1]$ such that $P\sim (\delta_{x_1}, \lambda\delta_{y_1}+ (1-\lambda)\delta_{y_2})$. By \lemmaref{lemma_narrow pre}, there exists $x_2'\in X_2$ where $P\sim (\delta_{x_1}, \lambda\delta_{y_1}+ (1-\lambda)\delta_{y_2})\sim (\delta_{x_1},\delta_{x_2'})$. 

If further $P\in \Phi_{2,\delta_y}$ for some $y\in X_2$, then we can find $p_1,p_1'\in \lxone$ with $(p_1,\delta_y) \succ P\succ (p_1', \delta_y)$. By the same argument, we can find $x'\in X_1$ such that $P\sim (\delta_{x'}, \delta_{y})$. 
\end{proof}

The next lemma generalize Axiom CC on each $\Phi_{2,\delta_y}$ by relaxing the requirement that $P$ and $Q$ must agree on the marginal lotteries. 

\begin{lemma}\label{lemma_CC_strong}
Suppose that Axiom CN fails. For each $y\in X_2$ and $P,Q,R,S\in \Phi_{2,\delta_y}$,  the following properties hold:

i). $P\sim Q$ and $P$ is compatible with $Q$ $\Longrightarrow$ $\alpha P +(1-\alpha)Q\sim P\sim Q$ for all $\alpha\in (0,1)$;

ii). $P\succ Q$ and $P$ is compatible with $Q$ $\Longrightarrow$ $P\succ \alpha P +(1-\alpha)Q\succ Q$ for all $\alpha\in (0,1)$;

iii). $P\succ Q$, $R\sim S$, $P$ is compatible with $R$ and $Q$ is compatible with $S$ $\Longrightarrow$ $\alpha P +(1-\alpha)R\succ \alpha Q +(1-\alpha)S$ for all $\alpha\in (0,1)$;

iv). $P\sim Q$, $R\sim S$, $P$ is compatible with $R$ and $Q$ is compatible with $S$ $\Longrightarrow$ $\alpha P +(1-\alpha)R\sim \alpha Q +(1-\alpha)S$ for all $\alpha\in (0,1)$.
\end{lemma}

\begin{proof}[Proof of \lemmaref{lemma_CC_strong}]
We first prove (i) and (ii). Suppose $P,Q\in \Phi_{2,\delta_y}$ for some  $y\in X_2$ and $P,Q$ are compatible. By \lemmaref{lemma_achievable3}, there exists $x_P,x_Q\in X_1^o$ such that $P\sim (\delta_{x_P},\delta_y)$ and $Q\sim (\delta_{x_Q},\delta_y)$.  By \lemmaref{lemma_narrow pre}, we can find $\epsilon>0$ such that for all $z_P\in [x_P-\epsilon, x_P], z_Q\in [x_Q-\epsilon, x_Q]$, there exist $z'_P\geq x_P,z'_Q\geq x_Q$ such that $P\sim (1/2\delta_{z_P} +1/2\delta_{z_P'},\delta_y)$ and $Q\sim (1/2\delta_{z_Q} +1/2\delta_{z_Q'},\delta_y)$. Moreover, as $z_P,z_Q$ increases, $z_P',z_Q'$ will be decreasing continuously. Since $P,Q$ are simple, that is, $\supp(P_1)\cup\sup(Q_1)$ is finite, we can construct $z^*_P\neq z^*_Q$, $z^{*'}_P\neq z^{*'}_Q$ and  $z^*_P, z^*_Q, z^{*'}_P,z^{*'}_Q\not\in \supp(P_1)\cup\sup(Q_1)$. Denote $P'= (1/2\delta_{z^*_P} +1/2\delta_{z^{*'}_P},\delta_y)$, $Q'=(1/2\delta_{z^*_Q} +1/2\delta_{z^{*'}_Q},\delta_y)$. Then $P\sim P', Q\sim Q'$ and $P,Q,P',Q'$ are compatible with each other. Apply indifference relation (\ref{eq_thm2_e2}) twice and we get for any $\alpha\in (0,1)$,
$$\alpha P +(1-\alpha)Q\sim \alpha P + (1-\alpha)Q'\sim \alpha P' + (1-\alpha)Q'.$$
Again by \lemmaref{lemma_narrow pre} given marginal lottery in source 2 as $\delta_y$, we know 
$$P\sim Q\Longrightarrow P'\sim Q' \Longrightarrow \alpha P +(1-\alpha)Q\sim \alpha P' + (1-\alpha)Q'\sim Q'\sim Q,$$
$$P\succ Q\Longrightarrow P'\succ Q' \Longrightarrow \alpha P +(1-\alpha)Q\succ \alpha P' + (1-\alpha)Q'\succ Q'\sim Q.$$

Then we show  (iii) and (iii) in a similar way. For $P,Q,R,S\in \Phi_{2,\delta_y}$, we can construct $P'\sim P, Q'\sim Q, R'\sim R$ and $S'\sim S$ such that $P',Q',R',S'\in \hat{\cP}$, $P'_2=Q'_2=R'_2=S'_2=\delta_y$, $P,R,P',R'$ are compatible with each other and $Q,S,Q',S'$ are compatible with each other. Then for any $\alpha\in (0,1)$, 
$$\alpha P +(1-\alpha)R\sim \alpha P + (1-\alpha)R'\sim \alpha P' + (1-\alpha)R',$$
$$\alpha Q +(1-\alpha)S\sim \alpha Q + (1-\alpha)S'\sim \alpha Q' + (1-\alpha)S'.$$
By \lemmaref{lemma_narrow pre} given marginal lottery in source 2 as $\delta_y$, we know
$$P\sim Q,R\sim S \Longrightarrow \alpha P +(1-\alpha)R\sim \alpha P' +(1-\alpha)R'\sim \alpha Q' + (1-\alpha)S'\sim \alpha Q + (1-\alpha)S,$$
$$P\succ Q,R\sim S \Longrightarrow \alpha P +(1-\alpha)R\sim \alpha P' +(1-\alpha)R'\succ \alpha Q' + (1-\alpha)S'\sim \alpha Q + (1-\alpha)S.$$
\end{proof}

\begin{lemma}\label{lemma_CC_strong2}
Suppose that Axiom CN fails. For each $y\in X_2$ and $P,Q,R,S\in \cup_{y\in X_2}\Phi_{2,\delta_y}$,  the following properties hold:

i). $P\sim Q$ and $P$ is compatible with $Q$ $\Longrightarrow$ $\alpha P +(1-\alpha)Q\sim P\sim Q$ for all $\alpha\in (0,1)$;

ii). $P\succ Q$ and $P$ is compatible with $Q$ $\Longrightarrow$ $P\succ \alpha P +(1-\alpha)Q\succ Q$ for all $\alpha\in (0,1)$;

iii). $P\succ Q$, $R\sim S$, $P$ is compatible with $R$ and $Q$ is compatible with $S$ $\Longrightarrow$ $\alpha P +(1-\alpha)R\succ \alpha Q +(1-\alpha)S$ for all $\alpha\in (0,1)$;

iv). $P\sim Q$,  $R\sim S$, $P$ is compatible with $R$ and $Q$ is compatible with $S$ $\Longrightarrow$ $\alpha P +(1-\alpha)R\sim \alpha Q +(1-\alpha)S$ for all $\alpha\in (0,1)$.
\end{lemma}

\begin{proof}[Proof of \lemmaref{lemma_CC_strong2}]
First, for any $P,Q,R,S\in \cup_{y\in X_2}\Phi_{2,\delta_y}$, we claim that there exist finitely many $z_k\in X_2, k=1,...,K$ such that $z_1<z_2<...<z_K$ and $P,Q,R,S\in \cup_{k=1}^K\Phi_{2,\delta_{z_k}}$. Choose $P^1,P^2 \in \{P,Q,R,S\}$ such that $P^1\succsim P,Q,R,S \succsim P^2$. Suppose that $P^1\in \Phi_{2,\delta_{y_1}}$ and $P^2\in \Phi_{2,\delta_{y_2}}$ with $y_1\geq y_2$. If $y_1=y_2$, then $P,Q,R,S\in \Phi_{2,\delta_{y_1}}$ and we are done.

 Now suppose that $y_1>y_2$ and by \lemmaref{lemma_narrow pre}, we can find $t,t'\in \lxone$ with $(t,\delta_{y_1})\succ P^1\succ (t',\delta_{y_1}), (t,\delta_{y_2})\succ P^2\succ (t',\delta_{y_2})$.  Notice that for each $y\in [y_2,y_1]$, $H(y):=\{P\in \cP:(t,\delta_{y})\succ P\succ (t',\delta_{y}) \}\subseteq \Phi_{2,\delta_y}$. By Axiom WC, $\{P\in \cP:P^1\succsim P\succsim P^2 \}\subseteq \cup_{y_2\leq y\leq y_1}H(y)$ and for all $y\in [y_2,y_1]$, there exists $\epsilon_y>0$ such that $H(y)\cap H(y')\neq \emptyset$ for all $y'\in [y-\epsilon_y, y+\epsilon_y]\cap [y_2,y_1]$. By Finite Cover Theorem, we can find finitely many $z_1<z_2<...<z_K\in  [y_2,y_1]$ with $[y_2,y_1]\subseteq \cup_{k=1}^K [z_k-\epsilon_{z_k}, z_k+\epsilon_{z_k}]$. This implies 
$$P,Q,R,S\in\{P\in \cP:P^1\succsim P\succsim P^2 \}\subseteq \cup_{y_2\leq y\leq y_1}H(y)=\cup_{k=1}^K H(z_k)\subseteq \cup_{k=1}^K\Phi_{2,\delta_{z_k}}.$$

Then we use induction to show that the four properties stated in the lemma hold for $P,Q,R,S\in \cup_{k=1}^K\Phi_{2,\delta_{z_k}}$. The proof idea is similar to the proof of \lemmaref{lemma_source2_union}. By \lemmaref{lemma_CC_strong}, the four properties hold if $P,Q,R,S\in \Phi_{2,\delta_{z_1}}$. Suppose by induction that they also hold if $P,Q,R,S\in \cup_{k=1}^t\Phi_{2,\delta_{z_k}}$ for some $1\leq t<K$.  By our construction of $\{z_k\}$, $\Phi_{2,\delta_{z_t}}\cap \Phi_{2,\delta_{z_{t+1}}}$ has nonempty interior. Choose $T^1,T^2\in \Phi_{2,\delta_{z_t}}\cap \Phi_{2,\delta_{z_{t+1}}}$ with $T^1\succ T^2$.  By \lemmaref{lemma_achievable3} and  \lemmaref{lemma_narrow pre}, as $P_1,Q_1,R_1,S_1$ have finite supports, we can find $p_1,p_2,q_1,q_2\in \lxone$ such that $(p_1,\delta_{z_{t+1}})\sim (q_1,\delta_{z_{t}})\sim T^1$, $(p_2,\delta_{z_{t+1}})\sim (q_2,\delta_{z_{t}})\sim T^2$ and $(p_1,\delta_{z_{t+1}}), (q_1,\delta_{z_{t}}), (p_2,\delta_{z_{t+1}}), (q_2,\delta_{z_{t}})$ are compatible with $P,Q,R,S$.

For properties (i) and (ii), suppose $P\succsim Q$, $P$ is  compatible with $Q$ and $P,Q\in \cup_{k=1}^{t+1}\Phi_{2,\delta_{z_k}}$. If $P\sim Q$, then $P,Q\in \Phi_{2,\delta_{z_k}}$ for some $k=1,...,t+1$ and hence property (i) holds by the inductive hypothesis. 

If $P\succ Q$, then it suffices to consider the case where $P\in \Phi_{2,\delta_{z_{t+1}}}\backslash (\cup_{k=1}^{t}\Phi_{2,\delta_{z_k}})$ and $Q\in (\cup_{k=1}^{t}\Phi_{2,\delta_{z_k}})\backslash \Phi_{2,\delta_{z_{t+1}}}$. This implies $P\succ T^1\succ T^2\succ Q$. By \lemmaref{lemma_achievable3}, there exist $\lambda_1\neq \lambda_2\in (0,1)$ such that $T^1\sim \lambda_1 P+(1-\lambda_1)Q$ and $T^2\sim \lambda_2 P+(1-\lambda_2)Q$. Then property (ii) holds for $\lambda=\lambda_1,\lambda_2$. 

Notice that at the moment we cannot conclude that $\lambda_1>\lambda_2$. Suppose that $\lambda_i>\lambda_{-i}$ for some $i=1,2$. By \lemmaref{lemma_achievable3}, we can find $P', Q'\in \hat{\cP}$ with $Q'\sim Q, P'\sim P$ and $P,P',Q,Q', (p_1,\delta_{z_{t+1}}), (q_1,\delta_{z_{t}}), (p_2,\delta_{z_{t+1}}), (q_2,\delta_{z_{t}})$ compatible with each other. This guarantees 
$$T^1\sim \lambda_1 P+(1-\lambda_1)Q\sim \lambda_1 P'+(1-\lambda_1)Q\sim \lambda_1 P+(1-\lambda_1)Q'\sim \lambda_1 P'+(1-\lambda_1)Q',$$
$$T^2\sim \lambda_2 P+(1-\lambda_2)Q\sim \lambda_2 P'+(1-\lambda_2)Q\sim \lambda_2 P+(1-\lambda_2)Q'\sim \lambda_2 P'+(1-\lambda_2)Q'.$$

By property (i), for all $\beta,\beta'\in (0,1)$, $\beta P+ (1-\beta)P'\sim P, \beta' Q+ (1-\beta')Q'\sim Q$. Apply indifference relation (\ref{eq_thm2_e2}) twice and we have for each $\lambda, \beta,\beta'\in (0,1)$
\begin{equation}\label{eq_thm2_e3}
    \lambda P + (1-\lambda)Q \sim \lambda (\beta P+ (1-\beta)P') + (1-\lambda)(\beta' Q+ (1-\beta')Q').
\end{equation}

For any $\lambda \in (\lambda_{-i}, \lambda_{i})$, let $\beta= 1$,  $\beta' = \frac{\lambda_i-\lambda}{\lambda_i(1-\lambda)}$, and (\ref{eq_thm2_e3}) becomes
\begin{align*}
   \lambda P + (1-\lambda)Q \sim & \frac{\lambda}{\lambda_i}( \lambda_i P + (1-\lambda_i)Q') + (1-\frac{\lambda}{\lambda_i})Q\\
\sim    &\frac{\lambda}{\lambda_i}(q_i,\delta_{z_t}) +(1-\frac{\lambda}{\lambda_i})Q 
   \end{align*}
The second indifference comes from the fact that $\lambda_i P + (1-\lambda_i)Q'\sim T^i\sim (q_i,\delta_{z_t}) $ and (\ref{eq_thm2_e2}). Then by the inductive hypothesis on $\cup_{k=1}^{t}\Phi_{2,\delta_{z_k}}$, we have $$P\succ (q_i,\delta_{z_t})\succ \lambda P + (1-\lambda)Q \sim \frac{\lambda}{\lambda_i}(q_i,\delta_{z_t}) +(1-\frac{\lambda}{\lambda_i})Q \succ Q.$$

If $\lambda>\lambda_i$, then let $\beta= \frac{\lambda-\lambda_i}{\lambda(1-\lambda_i)}$,  $\beta'=0$ and (\ref{eq_thm2_e3}) becomes
\begin{align*}
   \lambda P + (1-\lambda)Q \sim & \frac{\lambda-\lambda_i}{1-\lambda_i}P + (1-\frac{\lambda-\lambda_i}{1-\lambda_i})(\lambda_i P' + (1-\lambda_i)Q)\\
\sim    &\frac{\lambda-\lambda_i}{1-\lambda_i}P + (1-\frac{\lambda-\lambda_i}{1-\lambda_i})(p_i,\delta_{z_{t+1}})
   \end{align*}
The second indifference comes from the fact that $\lambda_i P' + (1-\lambda_i)Q\sim T^i\sim (p_i,\delta_{z_{t+1}}) $ and (\ref{eq_thm2_e2}). Then by \lemmaref{lemma_CC_strong} on $\Phi_{2,\delta_{z_{t+1}}}$, we have $$P\succ  \lambda P + (1-\lambda)Q \sim \frac{\lambda-\lambda_i}{1-\lambda_i}P + (1-\frac{\lambda-\lambda_i}{1-\lambda_i})(p_i,\delta_{z_{t+1}}) \succ (p_i,\delta_{z_{t+1}}) \succ Q.$$
A symmetric proof applies for the case where $\lambda<\lambda_{-i}$. This completes the proof for property (ii)
on $\cup_{k=1}^{t+1}\Phi_{2,\delta_{z_k}}$. 

Now consider $P,Q,R,S\in \cup_{k=1}^{t+1}\Phi_{2,\delta_{z_k}}$ where $P$ is compatible with $R$ and $Q$ is compatible with $S$.  We first suppose $P\sim Q, R\sim S$ and prove property (iv). If $P\sim R$, then the result is trivial by property (i). Without loss of generality, we assume $P\succ R$. By the inductive hypothesis, if suffices to prove the case for $P,Q\in \Phi_{2,\delta_{z_{t+1}}}\backslash (\cup_{k=1}^{t}\Phi_{2,\delta_{z_k}})$ and $R,S\in (\cup_{k=1}^{t}\Phi_{2,\delta_{z_k}})\backslash \Phi_{2,\delta_{z_{t+1}}}$. Following the proof for property (ii), we construct $T^1,T^2, (p_1,\delta_{z_{t+1}}), (q_1,\delta_{z_{t}}), (p_2,\delta_{z_{t+1}}), (q_2,\delta_{z_{t}}), P',Q',R',S'$. Concretely, these lotteries are mutually compatible and each of them is compatible with $P,Q,R,S$ such that  
$$P\sim Q\sim P'\sim Q'~,~ R\sim S\sim R'\sim S';$$
$$ T^1\sim (p_1,\delta_{t_{t+1}})\sim (q_1,\delta_{z_{t}})~,~T^2\sim (p_2,\delta_{t_{t+1}})\sim (q_2,\delta_{z_{t}}).$$
By the inductive hypothesis, we can find $\lambda, \lambda'\in (0,1)$ such that
$$\lambda P + (1-\lambda)(p_2,\delta_{z_{t+1}}) \sim T^1\sim \lambda Q + (1-\lambda)(q_2,\delta_{z_{t}}),$$
$$\lambda'(p_1,\delta_{z_{t+1}})  + (1-\lambda')R \sim T^2\sim \lambda' (q_1,\delta_{z_{t}}) + (1-\lambda')S.$$

By  \lemmaref{lemma_achievable3}, there exist $\eta,\eta'\in (0,1)$ with 
$$\eta P+ (1-\eta)R\sim T^2\sim \eta' Q + (1-\eta')S.$$
We claim that we can choose $\eta=\eta'$. To see this, notice that $(p_1,\delta_{z_{t+1}})\sim T^1\sim \lambda P + (1-\lambda)(p_2,\delta_{z_{t+1}})$, all of which are compatible with $R$, we have 
\begin{align*}
    \lambda'(p_1,\delta_{z_{t+1}})  + (1-\lambda')R\sim  \lambda\lambda'{P} +  (1-\lambda)\lambda'(p_2,\delta_{z_{t+1}}) + (1-\lambda')R\sim (p_2,\delta_{z_{t+1}}).
\end{align*}
Again, as $(p_2,\delta_{z_{t+1}})$ is compatible with both $P$ and $R$, by property (i) and (ii), it must be the case that 
$$T^2\sim (p_2,\delta_{z_{t+1}})\sim \frac{\lambda\lambda'}{\lambda\lambda'+ (1-\lambda')}P +\frac{1-\lambda'}{\lambda\lambda'+ (1-\lambda')}R.  $$
Hence we can choose $\eta =\frac{\lambda\lambda'}{\lambda\lambda'+ (1-\lambda')}:=\eta^1$. Similarly we can show that $\eta'=\eta^1$ guarantees $T^2\sim \eta' Q + (1-\eta_2)S$. 

A symmetric argument shows that there exist $\eta^2= \frac{\lambda}{\lambda+(1-\lambda)(1-\lambda')}\in (\eta^1,1)$ with 
$$\eta^{2} P+ (1-\eta^{2})R\sim T^1\sim \eta^{2} Q + (1-\eta^{2})S.$$

Now we consider $\eta$ with $\eta^{1}< \eta< \eta^{2}$. By (\ref{eq_thm2_e3}), we can set $\beta=\frac{(\eta-\eta^{1})\eta^2}{(\eta^{2}-\eta^{1})\eta}$, $\beta'= \frac{(\eta-\eta^{1})(1-\eta^2)}{(\eta^{2}-\eta^{1})(1-\eta)}$ and then
\begin{align*}
    \eta P + (1-\eta) R&\sim  \frac{\eta-\eta^{1}}{\eta^{2}-\eta^{1}}[\eta^{2} P+ (1-\eta^{2})R] + \frac{\eta^{2}-\eta}{\eta^{2}-\eta^{1}}[\eta^{1} P'+ (1-\eta^{1})R']\\
    &\sim \frac{\eta-\eta^{1}}{\eta^{2}-\eta^{1}}(p_1,\delta_{z_{t+1}}) + \frac{\eta^{2}-\eta}{\eta^{2}-\eta^{1}}(p_2, \delta_{z_{t+1}}).
\end{align*}

Similarly, 
$$\eta Q + (1-\eta) S\sim\frac{\eta-\eta^{1}}{\eta^{2}-\eta^{1}}(p_1,\delta_{z_{t+1}}) + \frac{\eta^{2}-\eta}{\eta^{2}-\eta^{1}}(p_2, \delta_{z_{t+1}}).$$
Hence $\eta P + (1-\eta) R\sim \eta Q + (1-\eta) S$ for $\eta\in (\eta^1,\eta^2)$.

Suppose that $\eta^{2}<\eta<1$. By (\ref{eq_thm2_e3}), we can set $\beta=\frac{(\eta^2-\eta^{1})\eta}{(\eta-\eta^{1})\eta^2}$, $\beta'= \frac{(\eta^2-\eta^{1})(1-\eta)}{(\eta-\eta^{1})(1-\eta^2)}$ and then
\begin{align*}
   (p_1,\delta_{z_{t+1}})\sim \eta^{2} P + (1-\eta^{2}) R&\sim \frac{\eta^{2}-\eta^{1}}{\eta-\eta^{1}}[\eta P+ (1-\eta)R] + \frac{\eta-\eta^{2}}{\eta-\eta^{1}}[\eta^{1} P'+ (1-\eta^{1})R']\\
    &\sim \frac{\eta^{2}-\eta^{1}}{\eta-\eta^{1}}[\eta P+ (1-\eta)R] + \frac{\eta-\eta^{2}}{\eta-\eta^{1}}(p_2, \delta_{z_{t+1}}).
\end{align*}

Similarly, 
\begin{align*}
   (p_1,\delta_{z_{t+1}})\sim \eta^{2} Q + (1-\eta^{2}) S\sim \frac{\eta^{2}-\eta^{1}}{\eta-\eta^{1}}[\eta Q+ (1-\eta)S] + \frac{\eta-\eta^{2}}{\eta-\eta^{1}}(q_2, \delta_{z_{t}}).
\end{align*}

We claim that $\eta P+ (1-\eta)R\in \cup_{k=1}^{t+1}\Phi_{2,\delta_{z_k}}$. To see this, note that we can construct $P',R'$ such that $P'_2=\delta_{z_{t+1}}$, $R'_2=\delta_{z_{k}}$ for some $k\leq t$ and $\eta P + (1-\eta)R\sim \eta P' + (1-\eta)R'$. By Axiom M and definition of $\Phi_{2,\delta_{z_k}}$, there exist $x,x'\in X_1$ such that $ (\delta_x,\delta_{z_{t+1}})\succ \eta P + (1-\eta)R\sim\eta P' + (1-\eta)R'\succ (\delta_{x'},\delta_{z_{k}})$. This implies $\eta P+ (1-\eta)R\in \cup_{k=1}^{t+1}\Phi_{2,\delta_{z_k}}$. Similarly we know $\eta Q+ (1-\eta)S\in \cup_{k=1}^{t+1}\Phi_{2,\delta_{z_k}}$.

If $\eta P + (1-\eta)R\succ T^2 \succ \eta Q + (1-\eta)S$ or $\eta Q + (1-\eta)S\succ T^2 \succ \eta P + (1-\eta)R$, then $ (p_1,\delta_{z_{t+1}})\succ T^2\succ  (p_1,\delta_{z_{t+1}})$, a contradiction. Hence either $\eta P + (1-\eta)R,\eta Q + (1-\eta)S\in   \Phi_{2,\delta_{z_{t+1}}}$ or $\eta P + (1-\eta)R,\eta Q + (1-\eta)S\in   \cup_{k=1}^{t}\Phi_{2,\delta_{z_k}}$. By the inductive hypothesis, as $(p_2,\delta_{z_{t+1}})\sim (q_2,\delta_{z_{t}})\in \Phi_{2,\delta_{z_t}}\cap \Phi_{2,\delta_{z_{t+1}}}$, independence properties (iii) and (iv) hold for $(\eta P + (1-\eta)R,\eta Q + (1-\eta)S,(p_2,\delta_{z_{t+1}}),  (q_2,\delta_{z_{t}}))$. Thus we must have $\eta P + (1-\eta)R\sim \eta Q + (1-\eta)S$. 

The proof for the case with $\eta\in (0,\eta^1)$ is symmetric. Therefore for all $\eta\in (0,1)$, $\eta P + (1-\eta) R \sim \eta Q+ (1-\eta) S$ and property (iv) holds on $\cup_{k=1}^{t+1}\Phi_{2,\delta_{z_k}}$.

Before proving property (iii), we claim that for each $P,Q\in \cup_{k=1}^{t+1}\Phi_{2,\delta_{z_k}}$ with $P\succ Q$, $P$ compatible with $Q$ and $1>\lambda_1>\lambda_2>0$, we have $\lambda_1 P + (1-\lambda_1)Q\succ \lambda_2 P + (1-\lambda_2)Q$. To see this, by (\ref{eq_thm2_e3}), we can find $P'\sim P$ where $P'$ is compatible with both $P$ and $Q$ such that 
\begin{align*}
\lambda_1 P+ (1-\lambda_1)Q &\sim \frac{\lambda_1-\lambda_2}{1-\lambda_2}P' + \frac{1-\lambda_1}{1-\lambda_2}[\lambda_2 P+(1-\lambda_2)Q]  \succ \lambda_2 P+(1-\lambda_2)Q
\end{align*}
The second strict preference holds since $P\sim P'\succ \lambda_2 P+ (1-\lambda_2)Q$ by property (ii).

For property (iii), suppose $P\succ Q$, $R\sim S$. If $P\succsim R\succsim Q$, then by properties (i) and (ii), $\lambda P +(1-\lambda) R\succsim R\sim S\succsim \lambda Q +(1-\lambda) S$. As $P\succ Q$, at least one of the above weak preference rankings should be strict and we are done.  Then either $P\succ Q \succ R\sim S$ or $R\sim S\succ P \succ Q$. We start with the former case.

By \lemmaref{lemma_achievable3}, we can find $\alpha\in (0,1)$ and $P',R'$ where $P',R'$ are compatible and both of them are compatible with $P,Q,R,S$ such that  $$R\sim R' \hbox{~and~} \alpha P + (1-\alpha)R'\sim P' \sim Q.$$

Then by property (iv), for any $\lambda\in (0,1)$, 
$$\lambda Q+(1-\lambda)S\sim \lambda P'+(1-\lambda)R\sim \lambda\alpha P + (1-\lambda)R + (1-\alpha)\lambda R'.$$
Since $R\sim R'$, property (i) implies that $\frac{1-\lambda}{1-\lambda\alpha}R +\frac{(1-\alpha)\lambda}{1-\lambda\alpha} R'\sim R$. By indifference relation (\ref{eq_thm2_e2}), 
$$\lambda Q+(1-\lambda)S\sim \lambda\alpha P + (1-\lambda\alpha)R \prec \lambda P+(1-\lambda)R$$
by the previous claim and $\alpha\in (0,1)$. A symmetric proof applies if $R\sim S\succ P\succ Q$. This completes the proof for property (iii) on $\cup_{k=1}^{t+1}\Phi_{2,\delta_{z_k}}$. 

By induction, the four properties hold for  $P,Q,R,S\in \cup_{k=1}^K\Phi_{2,\delta_{z_k}}$ and hence arbitrary 
 $P,Q,R,S\in \cup_{y\in X_2}\Phi_{2,\delta_{y}}$.\end{proof}

It is worthwhile to notice that $\cup_{y\in X_2}\Phi_{2,\delta_{y}}$ might not be the same as $ \cP$. The next lemmas shows that they only possibly differ in the worst and the best possible lottery. Concretely, if $\cone,\ctwo<+\infty$, then $({\cone}, {\ctwo})\in \cP \backslash  (\cup_{y\in X_2}\Phi_{2,\delta_{y}})$; if $\lcone, \lctwo>-\infty$, then $({\lcone}, {\lctwo})\in \cP \backslash  (\cup_{y\in X_2}\Phi_{2,\delta_{y}})$.

\begin{lemma}\label{lemma_extreme}
Suppose that Axiom CN fails. $\cP  \backslash (\cup_{y\in X_2}\Phi_{2,\delta_{y}})=\{(\lcone,\lctwo),({\cone}, {\ctwo})\}\cap \bR^2.$
\end{lemma}

\begin{proof}[Proof of \lemmaref{lemma_extreme}]
We will focus on the case with $\lcone, \lctwo>-\infty$ and $\cone,\ctwo<+\infty$. The proof for the other case is simpler as it only involves the worst or the best possible lottery. 

First, for each $P\in\cP$ with $({\cone}, {\ctwo})\succ P\succ (\lcone,\lctwo)$, there exists $Q,Q'\in \cup_{y\in X_2}\Phi_{2,\delta_{y}}$ with $Q\succ P\succ Q'$, which implies $Q\in \cup_{y\in X_2}\Phi_{2,\delta_{y}}$. Hence 
$$\cP =(\cup_{y\in X_2}\Phi_{2,\delta_{y}})\cup \{P\in\cP: P\sim (\lcone,\lctwo) \hbox{~or~} P\sim ({\cone}, {\ctwo})\}.$$
It suffices to show that $P\sim (\lcone,\lctwo) $ if and only if $P=(\lcone,\lctwo)$,  $P\sim ({\cone}, {\ctwo}) $ if and only if $P=({\cone}, {\ctwo})$. This is trivial by Axiom M as for any $P\neq (\lcone,\lctwo), ({\cone}, {\ctwo})$, $P$ dominates $(\lcone,\lctwo)$ and is dominated by $({\cone}, {\ctwo})$. 
\end{proof}

Using the same proof as in \lemmaref{lemma_CC_strong}, we can easily show that the independence property holds for $P,Q,R,S\in \Phi_{2,\delta_{0}}\cup\{(\lcone,\lctwo)\}$ if $\lcone,\lctwo>-\infty$ or $P,Q,R,S\in \Phi_{2,\delta_{\ctwo}}\cup\{({\cone}, {\ctwo})\}$ if ${\cone}, {\ctwo}<+\infty$. Then a direct corollary of \lemmaref{lemma_extreme} follows. 

\begin{corollary}\label{coro_CCstrong}
Suppose that Axiom CN fails.  Then the following properties hold:

i). $P\sim Q$ and $P$ is compatible with $Q$ $\Longrightarrow$ $\alpha P +(1-\alpha)Q\sim P\sim Q$ for all $\alpha\in (0,1)$;

ii). $P\succ Q$ and $P$ is compatible with $Q$ $\Longrightarrow$ $P\succ \alpha P +(1-\alpha)Q\succ Q$ for all $\alpha\in (0,1)$;

iii). $P\succ Q$, $R\sim S$, $P$ is compatible with $R$ and $Q$ is compatible with $S$ $\Longrightarrow$ $\alpha P +(1-\alpha)R\succ \alpha Q +(1-\alpha)S$ for all $\alpha\in (0,1)$;

iv). $P\sim Q$, $R\sim S$, $P$ is compatible with $R$ and $Q$ is compatible with $S$ $\Longrightarrow$ $\alpha P +(1-\alpha)R\sim \alpha Q +(1-\alpha)S$ for all $\alpha\in (0,1)$.
\end{corollary}

We end this section by slightly relaxing the requirement of compatibility. For each $P,Q\in \cP$, we say $P$ and $Q$ are {\it weakly compatible} if the following properties hold: 
\begin{itemize}
    \item $\supp(P_1)\cap \supp(Q_1)\subseteq \{\lcone,\cone\}$;
    \item when $\lcone\in \supp(P_1)\cap \supp(Q_1)$, we have $P_{2|\lcone}=Q_{2|\lcone}=\delta_{\lctwo}$;
    \item  when $\cone\in \supp(P_1)\cap \supp(Q_1)$, we have $P_{2|{\cone}}=Q_{2|{\cone}}=\delta_{\ctwo}$.
\end{itemize}

In other words, for $P$ weakly compatible with $Q$, we allow outcome $\lcone$ or $\cone$ to be contained in the  overlapping supports of $P_1$ and $Q_1$ only if the conditional lotteries of $P$ and $Q$ given outcome $x$ are both $\delta_{\lctwo}$ or $\delta_{\ctwo}$.

\begin{lemma}\label{lemma_CC_strong3}
Suppose that Axiom CN fails.  Then the following properties hold:

i). $P\sim Q$ and $P$ is weakly compatible with $Q$ $\Longrightarrow$ $\alpha P +(1-\alpha)Q\sim P\sim Q$ for all $\alpha\in (0,1)$;

ii). $P\succ Q$ and $P$ is weakly compatible with $Q$ $\Longrightarrow$ $P\succ \alpha P +(1-\alpha)Q\succ Q$ for all $\alpha\in (0,1)$;

iii). $P\succ Q$, $R\sim S$, $P$ is weakly compatible with $R$ and $Q$ is weakly compatible with $S$ $\Longrightarrow$ $\alpha P +(1-\alpha)R\succ \alpha Q +(1-\alpha)S$ for all $\alpha\in (0,1)$;

iv). $P\sim Q$, $R\sim S$, $P$ is weakly compatible with $R$ and $Q$ is weakly compatible with $S$ $\Longrightarrow$ $\alpha P +(1-\alpha)R\sim \alpha Q +(1-\alpha)S$ for all $\alpha\in (0,1)$.
\end{lemma}

\begin{proof}[Proof of \lemmaref{lemma_CC_strong3}]
Suppose $P,Q$ are weakly compatible but not compatible, that is, $\emptyset\neq \supp(P_1)\cap\supp(Q_1)\subseteq \{\lcone,\cone\}$. We claim that we can find $\tilde{P}\sim P$ and $\tilde{Q}\sim Q$ such that $\tilde{P}$ is compatible with $\tilde{Q}$ and for any $\alpha\in (0,1)$, $\alpha P + (1-\alpha)Q\sim \alpha \tilde{P} + (1-\alpha)\tilde{Q}$, unless $P=Q=(\lcone,\lctwo)$ or $P=Q=({\cone},{\ctwo})$.

{\bf Case 1}: If $\supp(P_1)~ \cap~ \supp(Q_1)= \{\lcone,\cone\}$, then $P_{2|\lcone}=Q_{2|\lcone}=\delta_{\lctwo}$ and $P_{2|{\cone}}=Q_{2|{\cone}}=\delta_{\ctwo}$. First, we suppose that $P_1(\lcone)+P_1(\cone)< 1$ or $Q_1(\lcone)+Q_1(\cone)< 1$. By symmetry, it suffices to focus on the former case.  Denote $P^o=\sum_{x\neq \lcone,\cone}\frac{P_1(x)}{1-P_1(\lcone)-P_1(\cone)}(\delta_x,P_{2|x})$. Then 
$$P=P_1(\lcone)(\delta_{\lcone},\delta_{\lctwo}) + P_1(\cone) (\delta_{\cone},\delta_{\ctwo}) + (1-P_1(\lcone)-P_1(\cone))P^o.$$
We know $P^o$ and $Q$ are compatible and $\lcone,\cone\not\in \supp(P^o_1)$. We can similarly define $Q^o$ if $Q_1(\lcone)+ Q_1(\cone)< 1$. Otherwise, just choose an arbitrary $Q^o$  so long as $\lcone,\cone\not\in \supp(Q^o_1)$.

By Axiom M, $({\cone}, {\ctwo})\succ P^o\succ ({\lcone}, {\lctwo})$. Then we can find $\epsilon_P>0$ such that $\cone-\epsilon_P,\epsilon_P\not\in \supp(P_1)\cup\supp(Q_1)$, $\epsilon_P\neq \cone-\epsilon_P$ and 
$$({\cone-\epsilon_P},{\ctwo})\succ P^o\succ (\lcone+{\epsilon_P},\lctwo).$$

By \lemmaref{lemma_achievable3}, we can find $\lambda_P\in (0,1)$ such that $$P^o\sim \lambda_P (\delta_{\cone-\epsilon_P},\delta_{\ctwo})+ (1-\lambda_P)(\delta_{\lcone+ \epsilon_P},\delta_{\lctwo}):={P^o}'.$$
and ${P^o}'$ is compatible with $P,Q$. 

By \cororef{coro_CCstrong}, we know 
\begin{align*}
    P':= &P_1(\lcone)(\delta_{\lcone},\delta_{\lctwo}) + P_1(\cone) (\delta_{\cone},\delta_{\ctwo}) + (1-P_1(\lcone)-P_1(\cone)){P^o}' \sim P.
\end{align*}

Notice that 
\begin{align*}
    P'=& [P_1(\lcone)(\delta_{\lcone},\delta_{\lctwo})+(1-P_1(\lcone)-P_1(\cone))(1-\lambda_P)(\delta_{\lcone+ \epsilon_P},\delta_{\lctwo})] \\
    &+[P_1(\cone)(\delta_{\cone},\delta_{\ctwo})+(1-P_1(\lcone)-P_1(\cone))\lambda_P(\delta_{\cone-\epsilon_P},\delta_{\ctwo})].
\end{align*}
By \lemmaref{lemma_narrow pre} given $\delta_{\lctwo}$ or $\delta_{\ctwo}$ in source two, we can then find $\up{p},\dw{p}\in\lxone$ with $(\up{p},\delta_{\ctwo}), (\dw{p},\delta_{\lctwo}), P,P',Q,Q'$ are pairwise compatible and 
\begin{align*}
(\up{p},\delta_{\ctwo})&\sim \frac{P_1(\cone)(\delta_{\cone},\delta_{\ctwo})+ \lambda_P(1-P_1(\lcone)-P_1(\cone))(\delta_{\cone-\epsilon_P},\delta_{\ctwo})}{P_1(\cone)+\lambda_P(1-P_1(\lcone)-P_1(\cone))};\\
(\dw{p},\delta_{\lctwo})&\sim \frac{P_1(\lcone)(\delta_{\lcone},\delta_{\lctwo})+ (1-\lambda_P)(1-P_1(\lcone)-P_1(\cone))(\delta_{\lcone+ \epsilon_P},\delta_{\lctwo})}{P_1(\lcone)+(1-\lambda_P)(1-P_1(\lcone)-P_1(\cone))}.
\end{align*}
It is important to notice that $\up{p}\neq \delta_{\cone}$ and $\dw{p}\neq \delta_{\lcone}$. 

Again by \cororef{coro_CCstrong}, we have 
\begin{align*}
    \tilde{P}=& \lambda_P^*(\dw{p},\delta_{\lctwo}) + (1-\lambda_P^*)(\up{p},\delta_{\ctwo}) \sim P'\sim P,
\end{align*}
where $\lambda_P^*=P_1(\lcone)+(1-\lambda_P)(1-P_1(\lcone)-P_1(\cone))$. Easy to see that $\tilde{P}$ is compatible with $P$ and $Q$. 

We then want to show that for each $\alpha\in (0,1)$, $\alpha P + (1-\alpha)Q\sim \alpha \tilde{P} + (1-\alpha)Q$. To see this, notice that for each $\alpha\in (0,1)$,
\begin{align*}
    \alpha P + (1-\alpha)Q&= (\alpha P_1(\lcone) + (1-\alpha)Q_1(\lcone))(\delta_{\lcone},\delta_{\lctwo}) + (\alpha P_1(\cone) + (1-\alpha)Q_1(\cone)) (\delta_{\cone},\delta_{\ctwo})  \\
    &+ \alpha(1-P_1(\lcone)-P_1(\cone))P^o + (1-\alpha)(1-Q_1(\lcone)-Q_1(\cone))Q^o.
\end{align*}
Recall that $P^o\sim {P^o}'$, $P^o$ is compatible with $Q$ and $P$. Then \cororef{coro_CCstrong} implies that 
\begin{align*}
    \alpha P + (1-\alpha)Q\sim &  (\alpha P_1(\lcone) + (1-\alpha)Q_1(\lcone))(\delta_{\lcone},\delta_{\lctwo}) + (\alpha P_1(\cone) + (1-\alpha)Q_1(\cone)) (\delta_{\cone},\delta_{\ctwo})  \\
    &+ \alpha(1-P_1(\lcone)-P_1(\cone)){P^o}' + (1-\alpha)(1-Q_1(\lcone)-Q_1(\cone)){Q^o}\\
    =& \alpha P' + (1-\alpha)Q\\
    =& \alpha [P_1(\lcone)(\delta_{\lcone},\delta_{\lctwo})+(1-P_1(\lcone)-P_1(\cone))(1-\lambda_P)(\delta_{\lcone+ \epsilon_P},\delta_{\lctwo})] \\
    &+ (1-\alpha)Q_1(\lcone)(\delta_{\lcone},\delta_{\lctwo})\\
    &+\alpha[P_1(\cone)(\delta_{\cone},\delta_{\ctwo})+(1-P_1(\lcone)-P_1(\cone))\lambda_P(\delta_{\cone-\epsilon_P}, \delta_{\ctwo})]\\
    &+(1-\alpha)Q_1(\cone)(\delta_{\cone},\delta_{\ctwo})\\
    &+ (1-\alpha)(1-Q_1(\lcone)-Q_1(\cone))Q^o.
\end{align*}

Notice that the first two terms in the last equation have $\delta_{\lctwo}$ in source two, while the third and four term have $\delta_{\ctwo}$ in the source two. Apply \lemmaref{lemma_narrow pre} given $\delta_{\lctwo}$ or $\delta_{\ctwo}$ in source two and  \cororef{coro_CCstrong} sequentially, we know 
\begin{align*}
    \alpha P + (1-\alpha)Q\sim & \alpha\lambda_P^*(\dw{p},\delta_{\lctwo}) +  (1-\alpha)Q_1(\lcone)(\delta_{\lcone},\delta_{\lctwo})\\
    & +\alpha(1-\lambda_P^*)(\up{p},\delta_{\ctwo})  + (1-\alpha)Q_1(\cone)(\delta_{\cone},\delta_{\ctwo})\\
    & + (1-\alpha)(1-Q_1(\lcone)-Q_1(\cone))Q^o\\
    =& \alpha \tilde{P} + (1-\alpha){Q}.
\end{align*}

Now we suppose $P_1(\lcone)+P_1(\cone) = Q_1(\lcone)+Q_1(\cone)= 1$. As $\supp(P_1)~ \cap~ \supp(Q_1)= \{0,\cone\}$, we have $P_1(\lcone),Q_1(\lcone)\in (0,1)$. By \lemmaref{lemma_extreme}, $(\delta_{\cone},\delta_{\ctwo})\succ P,Q\succ (\delta_{\cone},\delta_{\ctwo})$. Then we can find $\tilde{P}\sim P$ such that $\lcone,\cone\not\in \supp(\tilde{P}_1)$. This implies $\tilde{P}$ is compatible with $P$ and $Q$. For any $\beta\in (0,1)$, by \cororef{coro_CCstrong}, $P\sim \beta P + (1-\beta)P':=P^{\beta}$. Clearly, $P^{\beta}(\lcone)+P^{\beta}(\cone)<1$. We can apply the previous result for $P^{\beta}$ and $Q$, that is, for each $\beta$, we can find $\tilde{P}^{\beta}\sim P^{\beta}$ with $\tilde{P}^{\beta}$ compatible with $Q$ such that for each $\alpha\in (0,1)$,  $ \alpha P^{\beta} + (1-\alpha)Q\sim  \alpha \tilde{P}^{\beta} + (1-\alpha)Q$. Again by \cororef{coro_CCstrong}, we can actually choose $\tilde{P}^{\beta}$ to be the same across all $\beta\in (0,1)$. Denote it as $\tilde{P}$. Hence, for each $\beta,\alpha\in (0,1)$,
$$\alpha \beta P +\alpha(1-\beta)P' + (1-\alpha)Q\sim \alpha \tilde{P} +  (1-\alpha) Q$$
By mixture continuity of $\sim$, let $\beta\rightarrow 1$ and we have for each $\alpha\in (0,1)$, 
$$\alpha  P  + (1-\alpha)Q\sim \alpha \tilde{P} +  (1-\alpha) Q.$$

{\bf Case 2}: If $\supp(P_1)~ \cap~ \supp(Q_1)= \{\lcone\}$, then $P_{2|\lcone}=Q_{2|\lcone}=\delta_{\lctwo}$. By our assumption, either $P_1(\lcone)<1$ or $P_2(\lcone)<1$. Without loss of generality, assume $P_1(\lcone)<1$. 

If $P_1(\cone)<1-P_1(\lcone)$ or $P_{2|\cone}\neq \delta_{\ctwo}$, denote $P^o=\sum_{x\neq \lcone}\frac{P_1(x)}{1-P_1(\lcone)}(\delta_x,P_{2|x})$. Then 
$$P=P_1(\lcone)(\delta_{\lcone},\delta_{\lctwo}) +  (1-P_1(\lcone))P^o.$$
We know $P^o$ and $Q$ are compatible and $0\not\in \supp(P^o_1)$. We can similarly define $Q^o$ if $Q_1(\lcone)< 1$. Otherwise, just choose an arbitrary $Q^o$ so long as $\lcone\not\in \supp(Q^o_1)$.

By \lemmaref{lemma_extreme}, $({\cone}, {\ctwo})\succ P^o\succ (\lcone,\lctwo)$. Then we can find $\epsilon_P>0$ such that $\cone-\epsilon_P,\epsilon_P\not\in \supp(P_1)\cup\supp(Q_1)$, $\epsilon_P\neq \cone-\epsilon_P$ and 
$$(\delta_{\cone-\epsilon_P},\delta_{\ctwo})\succ P^o\succ (\delta_{\lcone+ \epsilon_P},\delta_{\lctwo}).$$

By \lemmaref{lemma_achievable3}, we can find $\lambda_P\in (0,1)$ such that $$P^o\sim \lambda_P (\delta_{\cone-\epsilon_P},\delta_{\ctwo})+ (1-\lambda_P)(\delta_{\lcone+ \epsilon_P},\delta_{\lctwo}):={P^o}'.$$
and ${P^o}'$ is compatible with $P,Q$. 

By \cororef{coro_CCstrong}, we know 
\begin{align*}
    P':= &P_1(\lcone)(\delta_{\lcone},\delta_{\lctwo})  + (1-P_1(\lcone)){P^o}' \sim P.
\end{align*}

Notice that 
\begin{align*}
    P'=& [P_1(\lcone)(\delta_{\lcone},\delta_{\lctwo})+(1-P_1(\lcone))(1-\lambda_P)(\delta_{\lcone+ \epsilon_P},\delta_{\lctwo})] \\
    &+(1-P_1(\lcone))\lambda_P(\delta_{\cone-\epsilon_P},\delta_{\ctwo}).
\end{align*}
By \lemmaref{lemma_narrow pre} given $\delta_{\lctwo}$ or $\delta_{\ctwo}$ in source two, we can then find $\up{p},\dw{p}\in\lxone$ with $(\up{p},\delta_{\ctwo}), (\dw{p},\delta_{\lctwo}), P,P',Q,Q'$ are pairwise compatible and 
\begin{align*}
(\up{p},\delta_{\ctwo})&\sim (\delta_{\cone-\epsilon_P},\delta_{\ctwo});\\
(\dw{p},\delta_{\lctwo})&\sim \frac{P_1(\lcone)(\delta_{\lcone},\delta_{\lctwo})+ (1-\lambda_P)(1-P_1(\lcone))(\delta_{\lcone+ \epsilon_P},\delta_{\lctwo})}{P_1(\lcone)+(1-\lambda_P)(1-P_1(\lcone))}.
\end{align*}
It is important to notice that $\up{p}\neq \delta_{\cone}$ and $\dw{p}\neq \delta_{\lcone}$. 

Again by \cororef{coro_CCstrong}, we have 
\begin{align*}
    \tilde{P}=& \lambda_P^*(\dw{p},\delta_{\lctwo}) + (1-\lambda_P^*)(\up{p},\delta_{\ctwo}) \sim P'\sim P,
\end{align*}
where $\lambda_P^*=P_1(\lcone)+(1-\lambda_P)(1-P_1(\lcone))$. Easy to see that $\tilde{P}$ is compatible with $P$ and $Q$. 

We then want to show that for each $\alpha\in (0,1)$, $\alpha P + (1-\alpha)Q\sim \alpha \tilde{P} + (1-\alpha)Q$. To see this, notice that for each $\alpha\in (0,1)$,
\begin{align*}
    \alpha P + (1-\alpha)Q&= (\alpha P_1(\lcone) + (1-\alpha)Q_1(\lcone))(\delta_{\lcone},\delta_{\lctwo})   \\
    &+ \alpha(1-P_1(\lcone))P^o + (1-\alpha)(1-Q_1(\lcone))Q^o.
\end{align*}
Recall that $P^o\sim {P^o}'$, $P^o$ is compatible with $Q$ and $P$. Then \cororef{coro_CCstrong} implies that 
\begin{align*}
    \alpha P + (1-\alpha)Q\sim &  (\alpha P_1(\lcone) + (1-\alpha)Q_1(\lcone))(\delta_{\lcone},\delta_{\lctwo})  \\
    &+ \alpha(1-P_1(\lcone)){P^o}' + (1-\alpha)(1-Q_1(\lcone)){Q^o}\\
    =& \alpha P' + (1-\alpha)Q\\
    =& \alpha [P_1(\lcone)(\delta_{\lcone},\delta_{\lctwo})+(1-P_1(\lcone))(1-\lambda_P)(\delta_{\lcone+ \epsilon_P},\delta_{\lctwo})] \\
    &+ (1-\alpha)Q_1(\lcone)(\delta_{\lcone},\delta_{\lctwo})\\
    &+\alpha(1-P_1(\lcone))\lambda_P(\delta_{\cone-\epsilon_P}, \delta_{\ctwo})\\
    &+ (1-\alpha)(1-Q_1(\lcone))Q^o.
\end{align*}

 Apply \lemmaref{lemma_narrow pre} given $\delta_{\lctwo}$  in source two and  \cororef{coro_CCstrong} sequentially, we know 
\begin{align*}
    \alpha P + (1-\alpha)Q\sim & \alpha\lambda_P^*(\dw{p},\delta_{\lctwo}) +  (1-\alpha)Q_1(\lcone)(\delta_{\lcone},\delta_{\lctwo})\\
    & +\alpha(1-\lambda_P^*)(\up{p},\delta_{\ctwo}) \\
    & + (1-\alpha)(1-Q_1(\lcone))Q^o\\
    =& \alpha \tilde{P} + (1-\alpha){Q}.
\end{align*}

If $P_1(\cone)=1-P_1(\lcone)$ or $P_{2|\cone}= \delta_{\ctwo}$, then the result can be proved by the same continuity  argument in Case 1. 

{\bf Case 3}: If $\supp(P_1)~ \cap~ \supp(Q_1)= \{\cone\}$, then the proof is symmetric to the proof of Case 2 and hence omitted.  

As an intermediate summary, for each $P,Q$ weakly compatible, we can find $\tilde{P}\sim P$ and $\tilde{Q}\sim Q$ such that $\tilde{P}$ is compatible with $\tilde{Q}$ and for any $\alpha\in (0,1)$, $\alpha P + (1-\alpha)Q\sim \alpha \tilde{P} + (1-\alpha)\tilde{Q}$, unless $P=Q=(\delta_{\lcone},\delta_{\lctwo})$ or $P=Q=(\delta_{\cone},\delta_{\ctwo})$.

Now we are ready to prove the four properties. 

For (i) and (ii), if  $P=Q=(\delta_{\lcone},\delta_{\lctwo})$ or $P=Q=(\delta_{\cone},\delta_{\ctwo})$, then the result is trivial. Otherwise, there exist $\tilde{P}\sim P$ and $\tilde{Q}\sim Q$ such that for any $\alpha\in (0,1)$, 
$$P\sim Q\Longrightarrow P'\sim Q' \Longrightarrow \alpha P +(1-\alpha)Q\sim \alpha P' + (1-\alpha)Q'\sim P,$$
$$P\succ Q\Longrightarrow P'\succ Q' \Longrightarrow \alpha P +(1-\alpha)Q\succ \alpha P' + (1-\alpha)Q'\sim P.$$

For (iii) and (iv), if $P=R=(\delta_{\lcone},\delta_{\lctwo})$ or $Q=S=(\delta_{\lcone},\delta_{\lctwo})$ or $P=R=(\delta_{\cone},\delta_{\ctwo})$ or $Q=S=(\delta_{\cone},\delta_{\ctwo})$, then by \lemmaref{lemma_extreme}, the primitives of (iii) or (iv) hold only if $P=Q=R=S$, in which case the result holds trivially. By excluding those cases, we can construct $\tilde{P}\sim P$, $\tilde{Q}\sim Q$,  $\tilde{R}\sim P$ and $\tilde{S}\sim S$ such that $\tilde{P}$ is compatible with $\tilde{R}$, $\tilde{Q}$ is compatible with $\tilde{S}$  and for any $\alpha\in (0,1)$, $\alpha P + (1-\alpha)R\sim \alpha \tilde{P} + (1-\alpha)\tilde{R}$, $\alpha Q + (1-\alpha)S\sim \alpha \tilde{Q} + (1-\alpha)\tilde{S}$.

By \cororef{coro_CCstrong}, we know
$$P\sim Q,R\sim S \Longrightarrow \alpha P +(1-\alpha)R\sim\alpha \tilde{P} + (1-\alpha)\tilde{R}\sim  \alpha \tilde{Q} + (1-\alpha)\tilde{S}\sim \alpha Q + (1-\alpha)S,$$
$$P\succ Q,R\sim S \Longrightarrow \alpha P +(1-\alpha)R\sim \alpha \tilde{P} + (1-\alpha)\tilde{R}\succ  \alpha \tilde{Q} + (1-\alpha)\tilde{S}\sim \alpha Q + (1-\alpha)S.$$
This completes the proof.

\end{proof}

\bigskip

{\bf Step 3: Then we show that $\succsim$ admits a  KP-style representation.}

Recall that the general (history-dependent) KP representation in two periods is given by $V^{KP}$ as 	$$V^{KP}(d)  = \sum_{(x,p)} w(x,CE_{v_{x}}(p))d(x,p)$$
 This next lemma introduces a KP-style representation in the space of lotteries $\cP$. Notice that the difference from the BIB model is that the conditional preference in source 2 is allowed to depend on the outcome in source 1.

\begin{lemma}\label{lemma_kp}
Suppose that Axiom CN fails. Then $\succsim$ on $\cP$ admits the a representation $U$ where for each $P\in\cP$, 
$$U(P)= \sum_x \hat{u}_D(x, CE_{v_x}(P_{2|x}))P_1(x)$$
with regular $\hat{u}_D$ and $v_x$ for all $x\in X_1$. 
\end{lemma}

Before proving \lemmaref{lemma_kp}, we introduce a mapping from the space of lotteries to the space of temporal lotteries. Denote $\mathcal{D}^*:=\mathcal{L}^0(X_1\times \mathcal{L}^0(X_2))$ as the set of temporal lotteries and  $\hat{\mathcal{D}}^*$ as a subset of $\mathcal{D}^*$ such that 
$$\hat{\mathcal{D}}^*:=\big\{d\in \mathcal{D}^*: d(x,p)d(x,p')=0,\forall x\in X_1, p\neq p'\in \lxtwo \big\}.$$
Notice that for $i=1,2$, $X_i$ with the standard topology is separable. By \cite{KP1978}, we know that the  $\cP$ and $\cD^*$ with weak topology can be metrizable by the Prokhorov metric. Endow $\hat{\cP}$ with the relative topology with respect to the weak topology on $\cP$ and $\hat{\cD}^*$ with the relative topology with respect to the weak topology on ${\cD}^*$.

Define a mapping $f:\cP\rightarrow \hat{\cD}^*$ as follows: for $P\in \cP$, denote $f[P]=d\in {\cD}^*$ such that for any $(x,q)\in X_1\times \mathcal{L}^0(X_2)$, $f[P](x,q) = P_1(x)$ if $q=P_{2|x}$ and $f[P](x,q) = 0$ if $q\neq P_{2|x}$. Clearly, for all $q'\neq P_{2|x}$, $f[p](x,q)=0$. Hence $f[P]\in \hat{\cD}^*$ and $f$ is well-defined. Inversely, $f^{-1}:\hat{\cD}^* \rightarrow\cP$ such that $f^{-1}[d](x,y) = \sum_{q\in \lxtwo}d(x,q)q(y)$. This is also well-defined as for each $x\in X_1$ there exists at most one $q\in \lxtwo$ with $d(x,q)>0$ for any $d\in \hat{\cD}^*$. Thus, $f$ is a bijective mapping between $\cP$ and  $\hat{\cD}^*$. It is worth noting that $f$ is not a homeomorphism as $f$ is not continuous, although $f^{-1}$ is continuous.

Now we define a binary relation $\succsim'$ on $\hat{\cD}^*$ by $d\succsim' d'$ if and only if $f^{-1}[d]\succsim f^{-1}(d')$. $\succ'$ and $\sim'$ are defined correspondingly. We have the following corollary of \lemmaref{lemma_CC_strong3}.

\begin{corollary}\label{coro_CC_strong2}
Suppose $\lambda_i>0$ for all $i$ and $\sum_{i=1}^n \lambda_i=1$.  For $d^1= \sum_{i=1}^n \lambda_i \delta_{(x_i,p_i)}$, $d^2= \sum_{i=1}^n \lambda_i \delta_{(y_i,q_i)}$ with $x_i\neq x_j, y_i\neq y_j$ for all $i\neq j$ and $ \delta_{(x_i,p_i)}\sim' \delta_{(y_i,q_i)}$ for all $i$, then $d^1\sim' d^2$. 
\end{corollary}
\begin{proof}[Proof of \cororef{coro_CC_strong2}]
By \lemmaref{lemma_CC_strong3}, as $x_i\neq x_j, y_i\neq y_j$ for all $i\neq j$ and $ \delta_{(x_i,p_i)}\sim' \delta_{(y_i,q_i)}$ for all $i$, we have 
$$\frac{\lambda_1}{\lambda_1+\lambda_2} (\delta_{x_1},p_1) + \frac{\lambda_2}{\lambda_1+\lambda_2} (\delta_{x_2},p_2)\sim \frac{\lambda_1}{\lambda_1+\lambda_2} (\delta_{y_1},q_1) + \frac{\lambda_2}{\lambda_1+\lambda_2} (\delta_{y_2},q_2).$$
Then by induction, we can get 
$$\sum_{i=1}^n\lambda_i(\delta_{x_i},p_i)\sim \sum_{i=1}^n\lambda_i(\delta_{y_i},q_i).$$
By the definition of $\succsim'$ and note that $d^1=f^{-1}(\sum_{i=1}^n\lambda_i(\delta_{x_i},p_i))$, $d^2=f^{-1}(\sum_{i=1}^n\lambda_i(\delta_{y_i},q_i))$, we conclude that $d^1\sim' d^2$.
\end{proof}

Then we extend $\succsim'$ to $\succsim^*$ on the entire space of temporal lotteries $\cD^*$. For any $d\in \cD^*\backslash \hat{\cD}^*$, denote $d_1$ as the marginal lottery in source 1 and $d_{2|x}$ as the lottery over marginal lotteries conditional on outcome $x$ in source 1. Denote $\supp(d_1)= \{x_1,...,x_N\}$ with $x_1<\cdots<x_N$ and for each $k=1,...,N$, $\supp(d_{2|x_k})=\{p_{k,1},...,p_{k,t_k}\}\subseteq \lxtwo$ with $t_k\geq 1$. Since $d\not\in  \hat{\cD}^*$, there exists some $k'$ with $t_{k'}>1$. We will construct $\succsim^*$ by relating $d$ to some temporal lottery in $\hat{\cD}^*$ as follows.

{\bf Stage 1. $i=1$}. If $t_1=1$, then define $d^1$ such that for all $x\neq x_1$ and $q\in \lxtwo$, $d^1(x,q)=d(x,q)$. Note that $(\delta_x,q)\succsim (\delta_x,\delta_{\lctwo})$ if $\lctwo>-\infty$. Denote $z_{1,1}=x_1$. By \lemmaref{lemma_narrow pre}, we can find $\hat{z}_{1,1}$ with $(\delta_{x_1},\delta_{\hat{z}_{1,1}})\sim (\delta_{x_1},p_{1,1})$. Denote $d^1(x_1,\delta_{\hat{z}_{1,1}})=d(x_1, p_{1,1})$ and $d^1(x_1,q)=0$ for all $q\neq \delta_{\hat{z}_{1,1}}$. 

If $t_1>1$ and $d(\lcone,\delta_{\lctwo})>0$, that is, $x_1=\lcone$ and $\delta_{\lctwo}\in \supp(d_{2|\delta_{\lcone}})$, then $\lcone,\lctwo>-\infty$ and we can reorder lotteries in $\supp(d_{2|\delta_{\lcone}})$ such that $(\delta_{x_1},p_{1,1})\precsim (\delta_{x_1},p_{1,2})\precsim\cdots\precsim (\delta_{x_1},p_{1,t_1})$. By Axiom M, we know $(\delta_{x_1},p_{1,1})=(\delta_{\lcone},\delta_{\lctwo})\prec (\delta_{x_1},p_{1,i})$ for each $i>1$. 

By continuity of $\succsim$ on $\hat{\cP}$, we can find $\frac{x_1+x_2}{2}>\up{z}_1>\lcone=x_1$ such that $(\delta_{x_1},p_{1,2})\succ (\delta_{\up{z}_1},\delta_{\lcone})$. Also, by Axiom M, for each $i>1$ and $x_1\leq z\leq \up{z}_1$, we have $(\delta_{x_1},p_{1,i})\in \Gamma_{1,\delta_{z}}$. Define  $z_{1,1}=x_1=\lcone$ and 
$z_{1,i}= \frac{i-1}{t_1}\up{z}_1$ for all $i=2,...,t_1$. Clearly, $z_{1,1}<z_{1,2}<\cdots<z_{1,t_1}< \up{z}_1$. By \lemmaref{lemma_achievable} and \lemmaref{lemma_narrow pre}, we can find $\hat{z}_{1,i}$ for $i\geq 1$ with $(\delta_{z_{1,i}}, \delta_{\hat{z}_{1,i}})\sim (\delta_{x_1},p_{1,i})$. Then we define $d^1\in \cD^*$ such that for any $x\not\in \{z_{1,i}\}_{i=1}^{t_1}$, $q\in \lxtwo$, $d^1(x,q)=d(x,q)$, and for each $i=1,...,t_1$, $d^1(z_{1,i},\delta_{\hat{z}_{1,i}})=d(x_1,p_{1,i})$ and $d^1(z_{1,i},q)=0$ for $q\neq \delta_{\hat{z}_{1,i}}$.

If $t_1>1$ and $d(\lcone,\delta_{\lctwo})=0$, then we can apply a similar construction method by choosing $z_{1,1}>\lcone$.

{\bf Stage 2. $i\geq 2$}. Consider $x_i> x_{i-1}\geq x_1\geq \lcone$. If $t_i=1$, then define $d^i\in \cD^*$ such that for all $x\neq x_i$ and $q\in \lxtwo$, $d^i(x,q)=d^{i-1}(x,q)$. Denote $z_{i,1}=x_i$ and by  \lemmaref{lemma_narrow pre}, there exists $\hat{z}_{i,1}$ with $(\delta_{x_i}, \delta_{\hat{z}_{i,1}})\sim (\delta_{x_i}, p_{i,1})$. Define $d^i(x_i,\delta_{\hat{z}_{i,1}})= d(x_i, p_{i,1})$ and $d^i(x_i,q)=0$ for $q\neq \delta_{\hat{z}_{i,1}}$.

If instead $t_i>1$, again we assume that without loss of generality, $(\delta_{x_i},p_{i,1})\precsim \cdots\precsim (\delta_{x_i},p_{i,t_i})$. As $x_i>x_1\geq \lcone$, $(\delta_{x_i},p_{i,1})\succ (\delta_{\frac{x_i+x_{i-1}}{2}},\delta_{\lctwo}) \succ (\delta_{x_1},\delta_{\lctwo}).$  

If $p_{i,1}=\delta_{\lctwo}$, then $(\delta_{x_i},p_{i,1})\neq (\delta_{x_i},p_{i,2})$ implies that $(\delta_{x_i},p_{i,1})\prec (\delta_{x_i},p_{i,2})$.  Again, we can find  $\frac{x_i+x_{i+1}}{2}>\up{z}_i>x_i$ such that $(\delta_{x_i},p_{1,2})\succ (\delta_{\up{z}_i},\delta_{\lctwo})$. By Axiom M, for each $j>1$ and $x_i\leq z\leq \up{z}_i$, we have $(\delta_{x_i},p_{i,j})\in \Gamma_{1,\delta_{z}}$. 

Denote $z_{i,j}=x_i + \up{z}_i\frac{j-1}{t_i}$ for $j=1,2,...,t_i$. Clearly, $x_i=z_{i,1}<z_{i,2}<\cdots<z_{i,t_i}< \up{z}_i$. By \lemmaref{lemma_achievable} and \lemmaref{lemma_narrow pre}, we can find  $\hat{z}_{i,j}$ for $j\geq 1$ with $(\delta_{z_{i,j}}, \delta_{\hat{z}_{i,j}})\sim (\delta_{x_i},p_{i,j})$.  Then we define $d^i\in \cD^*$ such that for any $x\not\in \{z_{i,j}\}_{j=1}^{t_i}$, $q\in \lxtwo$, $d^i(x,q)=d^{i-1}(x,q)$, and for each $j=1,...,t_i$, $d^i(z_{i,j},\delta_{\hat{z}_{i,j}})=d(x_i,p_{i,j})$ and $d^i(z_{i,j},q)=0$ for $q\neq \delta_{\hat{z}_{i,j}}$.

The algorithm ends at $i=N<+\infty$. We know that $\supp(d^N_1)=\bigcup_{k=1}^N \{z_{k,1},\cdots, z_{k,t_k}\}$. The discussion with $i=N$ is similar to the discussion with $i=1$ as we need to consider the cases where $t_N>1$,  and $d(\cone,\delta_{\ctwo})>0$ or $d(\cone,\delta_{\ctwo})=0$. For each $z_{k,i}\in \supp(d_1^N)$, we have $\hat{z}_{k,i}$ with $(\delta_{z_{k,i}}, \delta_{\hat{z}_{k,i}})\sim (\delta_{x_k},p_{k,i})$ and $d^k(z_{k,i}, \delta_{\hat{z}_{k,i}})=d(x_k,p_{k,i})$ for all  $1\leq i\leq t_k$ and $1\leq k\leq N$.  Also, $\{z_{k,i}\}$ admits a lexicographic order, that is, $z_{k,i}<z_{k',i'}$ if $i<i'$ or $i=i'$ and $j<j'$. This implies $d^N\in \hat{\cD}^*$. In this way, we have defined a mapping $h:\cD^*\backslash \hat{\cD}^*\rightarrow \hat{\cD}^*$ where $h(d)=d^N$. When there is no confusion, we can also use the same technique to derive $h(d)$ for $d\in \hat{\cD}^*$. Although it might be the case that $h(d)\neq d$, we must have $h(d)\sim' d$ as is shown in the next paragraph. Easy to show that $h(h(d))=h(d)$ for all $d\in \cD^*$.

Now we can define $\succsim^*$ on ${\cD}^*$ such that $\succsim^*$ agrees with $\succsim'$ on $\hat{\cD}^*$ and $d\sim^* h(d)$ for $d\in \cD^*\backslash \hat{\cD}^*$. Then we know $d\succsim^*d'$ if and only if $h(d)\succsim'h(d')$ for all $d,d'\in \cD^*$. To verify that $\succsim^*$ is well-defined, we need to argue that the arbitrary choice of $\{z_{k,i}\}$ does not affect the definition of $\succsim^*$. Consider two constructions $h$ and $\hat{h}$. For each $d=\sum_{k,i}d(x_k,p_{k,i})\delta_{(x_k,p_{k,i})}$, $h(d) = \sum_{k,i}d(x_k,p_{k,i})\delta_{(z_{k,i},\delta_{\hat{z}_{k,i}})}$ and $\hat{h}(d) = \sum_{k,i}d(x_k,p_{k,i})\delta_{(z'_{k,i},\delta_{\hat{z}'_{k,i}})}$ such that $z_{k,i}\neq z_{k',i'}$, $z'_{k,i}\neq z'_{k',i'}$ for all $(k,i)\neq (k',i')$ and for all $k,i$, $(\delta_{z'_{k,i}}, \delta_{\hat{z}'_{k,i}})\sim (\delta_{z_{k,i}}, \delta_{\hat{z}_{k,i}})\sim (x_k,p_{k,i})$. By \cororef{coro_CC_strong2}, $h(d)\sim' \hat{h}(d)$. Hence the definition of $\succsim^*$ is not affected by the specific construction of  $h$.

The next lemma extends \cororef{coro_CC_strong2} to temporal lotteries in $\cD^*\backslash \hat{\cD}^*$.

\begin{lemma}\label{lemma_CC_strong4}
Suppose $\lambda_i>0$ for all $i$ and $\sum_{i=1}^n \lambda_i=1$.  For $d^1= \sum_{i=1}^n \lambda_i \delta_{(x_i,p_i)}$, $d^2= \sum_{i=1}^n \lambda_i \delta_{(y_i,q_i)}$ with  $ \delta_{(x_i,p_i)}\sim' \delta_{(y_i,q_i)}$ for all $i$, then $d^1\sim^* d^2$. 
\end{lemma}

\begin{proof}[Proof of \lemmaref{lemma_CC_strong4}]
We first prove the result for the case that $(x_i,p_i)\neq (x_j,p_j), (y_i,q_i)\neq (y_j,q_j)$ for all $i\neq j$. Notice that $h(d^1)=\sum_{i=1}^n \lambda_i\delta_{(z^1_i,\delta_{\hat{z}^1_i})}$ with $z^1_i\neq z^1_j$ for $i\neq j$ and $(\delta_{z^1_i}, \delta_{\hat{z}^1_i})\sim (\delta_{x_i},p_i)$ for each $i=1,...,n$. Similarly, $h(d^2) = \sum_{i=1}^n \lambda_i\delta_{(z^2_i,\delta_{\hat{z}^2_i})}$ with $z^2_i\neq z^2_j$ for $i\neq j$ and $(\delta_{z^2_i}, \delta_{\hat{z}^2_i})\sim (\delta_{x_i},p_i)$ for each $i=1,...,n$. We know that $h(d^1)\sim^* d^1$ and $h(d^2)\sim^* d^2$. By \cororef{coro_CC_strong2}, we have $h(d^1)\sim^* h(d^2)$ and hence $d^1\sim^*d^2$. 

Now we consider the general case. If $i=1$, then the result is trivial. Suppose $i\geq 2$. 

First, we can reorder the subscripts so that $(x_i,p_i)=(\delta_{\lcone},\delta_{\lctwo})$ for $i\leq k$ for some $k\geq 0$ and $(x_i,p_i)\succ (\delta_{\lcone},\delta_{\lctwo})$ for $i>k$. Since $(\delta_{x_i},p_i)\sim (\delta_{y_i},q_i)$, we know that $(y_i,q_i)=(\delta_{\lcone},\delta_{\lctwo})$ for $i\leq k$ and $(y_i,q_i)\succ (\delta_{\lcone},\delta_{\lctwo})$ for $i>k$. Then we can write 
$$d^1=[\sum_{i=1}^k\lambda_i](\delta_{\lcone},\delta_{\lctwo}) + \sum_{i=k+1}^n \lambda_i \delta_{(x_i,p_i)}~,~d^2=[\sum_{i=1}^k\lambda_i](\delta_{\lcone},\delta_{\lctwo}) + \sum_{i=k+1}^n \lambda_i \delta_{(y_i,q_i)}.$$
This implies that we can assume that $(\delta_{x_i},p_i)\neq (\delta_{\lcone},\delta_{\lctwo})$ for all $i\geq 2$. By a similar argument, we can assume $(\delta_{x_i},p_i)\neq (\delta_{\cone},\delta_{\ctwo})$ for all $i< n$.

Without loss of generality, we can further assume $\cone>x_i>\lcone$ for all $2\leq i\leq n-1$ as we can always replace $(\delta_{\lcone},p_i)\neq (\delta_{\lcone},\delta_{\lctwo})$ with $(\delta_{a},q)$ for some $a>\lcone$, and  $(\delta_{\cone},p_i)\neq (\delta_{\cone},\delta_{\ctwo})$ with $(\delta_{b},q)$ for some $\lcone<a,b<\cone$ without changing the preference ranking of $d^1$. 

i). Suppose that $(\delta_{x_1},p_1)\succ (\delta_{\lcone},\delta_{\lctwo})$ and $(\delta_{x_n},p_n)\prec (\delta_{\cone},\delta_{\ctwo})$. By reordering, there exists a partition of $\{1,...,n\}$ as $\{1,...,t_1\}, \cdots, \{t_{k-1}+1,...,n\}$ such that $(x_i,p_i)=(x_j,p_j)$ for all $t_l+1\leq i,j\leq t_{l+1}$ with $0\leq l\leq k-1$ and $t_0=0, t_k=n$. 

For $l=0$, that is, $1\leq i\leq t_1$, by continuity on $\hat{\cP}$ and \lemmaref{lemma_achievable}, we can construct $z_i>\lcone$, $\hat{z}_i\geq \lctwo$ with $(\delta_{z_i},\delta_{\hat{z}_i})\sim (\delta_{x_i},p_i)= (\delta_{x_1},p_1)$ for all $i=1,...,t_1$ and $z_i\neq z_j$ for all $i\neq j$. By applying \lemmaref{lemma_CC_strong3} repeatedly,  we derive
$$\sum_{i=1}^{t_1} \frac{\lambda_i}{\sum_{j=1}^{t_1}\lambda_j}(\delta_{z_i}, \delta_{\hat{z}_i})\sim (\delta_{x_1}, p_1). $$
The same result holds for $l=1,...,k-1$.

Now recall that $d^1= \sum_{i=1}^n \lambda_i \delta_{(x_i,p_i)}=\sum_{l=0}^{k-1}(\sum_{i=t_l+1}^{t_{l+1}}\lambda_i)\delta_{(x_i,p_i)}$.  By definition of $h$, we can find $h(d^1)\sim^* d^1$ with $h(d_1)\in \hat{\cD}^*$. Denote $h(d^1) = \sum_{l=0}^{k-1}\hat{\lambda}_{l+1}\delta_{(x'_{l+1}, \hat{x}'_{l+1})}$, where $\hat{\lambda}_{l+1}= \sum_{i=t_l+1}^{t_{l+1}}\lambda_i$ and $(\delta_{x'_{l+1}}, \delta_{\hat{x}'_{l+1}})\sim (\delta_{x_{t_l+1}}, p_{t_{l}+1})\sim \sum_{i=t_l+1}^{t_{l+1}} \frac{\lambda_i}{\sum_{j=t_l+1}^{t_{l+1}}\lambda_j}(\delta_{z_i}, \delta_{\hat{z}_i})$ for each $l$. Denote $R_{l+1}=\sum_{i=t_l+1}^{t_{l+1}} \frac{\lambda_i}{\sum_{j=t_l+1}^{t_{l+1}}\lambda_j}(\delta_{z_i}, \delta_{\hat{z}_i})$. Note that $R_l$ and $R_{l'}$ are compatible, $(\delta_{x'_{l}}, \delta_{\hat{x}'_{l}})$ and $(\delta_{x'_{l'}}, \delta_{\hat{x}'_{l'}})$ are compatible for all $l\neq l'$. By \lemmaref{lemma_CC_strong3} and the definition of $\succsim'$, 
$$d^1\sim^* h(d^1)= \sum_{l=1}^k \hat{\lambda}_l \delta_{( x'_{l}, \delta_{\hat{x}'_{l}})}\sim' \sum_{l=1}^k h( \hat{\lambda}_l R_l)=\sum_{i=1}^n \lambda_i \delta_{(z_i,\delta_{\hat{z}_i})}\in \hat{\cD}^*.$$

Similarly, we can find $z'_i,\hat{z}'_i$ for $i=1,...,n$ such that $z'_i\neq z'_j$ for all $i\neq j$, $(\delta_{z'_i}, \delta_{\hat{z}'_i})\sim (\delta_{y_i},q_i)$ for each $i$ and $$d^2\sim^* \sum_{i=1}^n \lambda_i \delta_{(z'_i,\delta_{\hat{z}'_i})}\in \hat{\cD}^*.$$
By \cororef{coro_CC_strong2}, we have 
$$d^2\sim^* \sum_{i=1}^n \lambda_i \delta_{(z'_i,\delta_{\hat{z}'_i})}\sim^* \sum_{i=1}^n \lambda_i \delta_{(z_i,\delta_{\hat{z}_i})}\sim^* d^1.$$

(ii). Now we turn to the case where $(\delta_{x_1},p_1)= (\delta_{\lcone},\delta_{\lctwo})$ or $(\delta_{x_n,p_n})= (\delta_{\cone},\delta_{\ctwo})$ or both. This implies $(\delta_{y_1},q_1)=(\delta_{\lcone},\delta_{\lctwo})$ or $(\delta_{y_n,q_n})= (\delta_{\cone},\delta_{\ctwo})$ or both. If $\lambda_1+\lambda_n=1$, then the result is trivial as $n=2$ and we are back to the special case where $(x_i,p_i)\neq (x_j,p_j), (y_i,q_i)\neq (y_j,q_j)$ for all $i\neq j$.

 Recall that $(\delta_{\cone},\delta_{\ctwo})\succ (\delta_{x_i},p_i)\succ (\delta_{\lcone},\delta_{\lctwo}), (\delta_{\cone},\delta_{\ctwo})\succ (\delta_{y_i},q_i)\succ (\delta_{\lcone},\delta_{\lctwo}) $ for all $2\leq i\leq n-2$. Define 
$$\hat{d}^1= \frac{1}{1-\lambda_1-\lambda_n}\sum_{i=2}^n\lambda_i\delta_{(x_i,p_i)}~,~\hat{d}^2= \frac{1}{1-\lambda_1-\lambda_n}\sum_{i=2}^n\lambda_i\delta_{(y_i,q_i)}. $$
Then we are back to case (i) and $\hat{d}^1\sim^* \hat{d}^2.$ 
By \lemmaref{lemma_CC_strong3} and the definition of $\succsim'$, we know $h(d^1)=\lambda_1 \delta_{(0,\delta_0)}+ \lambda_n \delta_{(\cone,\delta_{\ctwo})}+ (1-\lambda_1-\lambda_n)h(\hat{d}^1)\sim \lambda_1 \delta_{(0,\delta_0)}+ \lambda_n \delta_{(\cone,\delta_{\ctwo})}+ (1-\lambda_1-\lambda_n)h(\hat{d}^2)=h(d^2)$. Thus, by definition of $\succsim^*$, we conclude that $d^1\sim^*d^2$. \end{proof}

As a summary, we have defined a preference relation $\succsim^*$ on ${\cD}^*$, which is a mixture space. For $d,d'\in {\cD}^*$ and $\alpha\in (0,1)$, we define the $\alpha$-mixture of $d$ and $d'$ as 
$$[\alpha d + (1-\alpha) d'](x,q) = \alpha d(x,q) + (1-\alpha)d'(x,q), \forall~\alpha\in (0,1), (x,q)\in X_1\times \lxtwo.$$

We claim that $\succsim^*$ satisfies the vNM independence property and mixture continuity. 

For the independence property, fix $d,d''\in \cD^*$ and $\alpha \in (0,1)$. Denote $d=\sum_{i=1}^n \lambda_i\delta_{(x_i,p_i)}$ and $d''=\sum_{j=1}^m \eta_j\delta_{(y_j,q_j)}$. As $\supp(d_1)\cup\supp(d''_1)$ is finite, for any $i$ such that $(x_i,p_i)\neq (\delta_{\lcone},\delta_{\lctwo})$ and $ (\delta_{\cone},\delta_{\ctwo})$, we can find $(\delta_{z_i^d}, \delta_{\hat{z}_i^d})\sim (\delta_{x_i},p_i)$; for any $j$ such that $(y_j,q_j)\neq (\delta_{\lcone},\delta_{\lctwo}) $ and $ (\delta_{\cone},\delta_{\ctwo})$, we can find $(\delta_{z_j^{d''}}, \delta_{\hat{z}_j^{d''}})\sim (\delta_{y_j},q_j)$. Moreover, we require that $\cone>z_i^d,z_j^{d''}>\lcone$ and $z_i^d\neq z_j^{d''}$  for all $i,j$. For any $i,j$ with $(x_i,p_i)=(\delta_{\lcone},\delta_{\lctwo})$ or $(y_j,q_j)=(\delta_{\lcone},\delta_{\lctwo})$, denote $z_i^d=z_j^{d''}=\lcone$ and $\hat{z}_i^d=\hat{z}_j^{d''}=\lctwo$. For any $i,j$ with $(x_i,p_i)=(\delta_{\cone},\delta_{\ctwo})$ or $(y_j,q_j)=(\delta_{\cone},\delta_{\ctwo})$, denote $z_i^d=z_j^{d''}=\cone$ and $\hat{z}_i^d=\hat{z}_j^{d''}=\ctwo$.

By  \lemmaref{lemma_CC_strong4}, we know 
$$\alpha d +(1-\alpha)d'' \sim^* \alpha \hat{d} +(1-\alpha)\hat{d}''.$$
where $\hat{d}= \sum_{i=1}^n \lambda_i \delta_{(z_i^d, \delta_{\hat{z}_i^d})}\sim^* d$
and $\hat{d}''= \sum_{j=1}^m\eta_j \delta_{(z_i^{d''}, \delta_{\hat{z}_i^{d''}})}\sim^* d''$. As $\hat{d},\hat{d}''\in \hat{\cD}^*$, we can denote $P=f^{-1}[\hat{d}]$ and $R=f^{-1}[\hat{d}'']$. Easy to see that $P$ and $R$ are weakly compatible, and  $\alpha d +(1-\alpha)d''\sim^* \alpha f(P) +(1-\alpha)f(R)$ for any $\alpha\in (0,1)$.

Now we consider $d'$ with $d\succ^* d'$. Using the same argument, we can find $Q$ and $S$ such that $d'\sim^*f(Q)$, $d''\sim^* f(S)$, $Q$ and $S$ are weakly compatible and $\alpha d' +(1-\alpha)d''\sim^* \alpha f(Q) +(1-\alpha)f(S)$ for any $\alpha\in (0,1)$. Note that $d\succ^*d'$ if and only if $P\succ Q$ and $d''=d''$ implies that $R\sim S$. By \lemmaref{lemma_CC_strong3}, we know $\alpha P +(1-\alpha)R\succ \alpha Q +(1-\alpha)S$ for any $\alpha\in (0,1)$. It is easy to verify that $f(\alpha P+(1-\alpha R))= \alpha f(P) + (1-\alpha)f(R)$ and $f(\alpha Q+(1-\alpha S))= \alpha f(Q) + (1-\alpha)f(S)$ since $P$ and $R$ are weakly compatible and  $Q$ and $S$ are weakly compatible. Thus for each $\alpha\in (0,1)$,
\begin{align*}
    &\alpha P +(1-\alpha)R\succ \alpha Q +(1-\alpha)S\\
    \Longleftrightarrow & f(\alpha P +(1-\alpha)R) \succ^* f(\alpha Q +(1-\alpha)S)\\
      \Longleftrightarrow & \alpha f(P) + (1-\alpha)f(R) \succ^* \alpha f(Q) + (1-\alpha)f(S)\\
       \Longleftrightarrow & \alpha d + (1-\alpha)d'' \succ^* \alpha d' + (1-\alpha)d''.
\end{align*}
Hence $\succsim^*$ satisfies the vNM independence property on $\cD^*$. 

Next we show the mixture continuity of $\succsim^*$ on $\cD^*$.  For any $d,d',d''\in \cD^*$, by the above proof for independence, we can find $P,Q,R\in \cP$ such that $f(P)\sim^* d, f(Q)\sim^*d', f(R)\sim^*d''$ and for each $\alpha\in (0,1)$, $\alpha P +(1-\alpha)Q\in \cP$, $f(\alpha P + (1-\alpha)Q)= \alpha f(P) + (1-\alpha)f(Q)$ and $\alpha d + (1-\alpha)d'\sim^* \alpha f(P) +(1-\alpha)f(Q)$. Then 
\begin{align*}
   A & = \big\{ \alpha\in[0,1]: \alpha d + (1-\alpha)d'\succ^* d'' \big\}\\
   & = \big\{ \alpha\in[0,1]: \alpha f(P) + (1-\alpha)f(Q)\succ^* f(R) \big\}\\
      & = \big\{ \alpha\in[0,1]: f(\alpha P + (1-\alpha)Q)\succ^* f(R) \big\}\\
      &= \big\{ \alpha\in[0,1]: \alpha P + (1-\alpha)Q\succ R \big\}.
\end{align*}
By mixture continuity of $\succsim$ on $\cP$, we know $A$ is open in $[0,1]$. Similarly, $\big\{ \alpha\in[0,1]: \alpha d + (1-\alpha)d'\prec^* d'' \big\}$ is also open in $[0,1]$. Thus, $\succsim^*$ satisfies mixture continuity on $\cD^*$. 

We are now prepared to finish the proof of \lemmaref{lemma_kp}.

\begin{proof}[Proof of \lemmaref{lemma_kp}]
Since $\cD^*$ is a mixture space  and the preference relation $\succsim^*$ satisfies mixture continuity and the independence axiom, by the Mixture Space Theorem, $\succsim^*$ on $\cD^*$ admits an EU representation $U$ with a utility index $w_D:X_1\times \lxtwo \rightarrow \bR$. That is,  the expected utility of $d\in \cD^*$ is given by 
$$U(d)=\sum_{x,p}w_D(x,p)d(x,p).$$
We also know that $w_D$ is unique up to a positive affine transformation. 

Recall that $\succsim^*$ extends $\succsim'$ from $\hat{\cD}^*$ to $\cD^*$ and $d\succsim' d'$ if and only if $f^{-1}(d)\succsim f^{-1}(d')$ for all $d,d'\in \hat{\cD}^*$. Then the utility function 
$$V(P)=\sum_x P_1(x)w_D(x,P_{2|x}),\forall~P\in \cP$$
represents $\succsim$ on $\cP$. 

Then we derive the implications of Axiom CI. Recall that $\succsim$ restricted to $\{\delta_x\}\times \lxtwo$ for each $x\in X_1$ admits an EU representation with some regular utility index $v_x$. Then by uniqueness up to a positive affine transformation, there exists a continuous and monotone function $\phi_x$ such that for all $p\in \lxtwo$, 
$$V(\delta_x,p)=w_D(x,p) = \phi_x(CE_{v_x}(p))$$
Define $\hat{u}_D:X_1\times X_2\rightarrow \bR$ as $\hat{w}(x,y)= \phi_x(y)$ for all $(x,y)\in X$. Then the representation can be rewritten as 
$$V(P)=\sum_x\hat{u}_D(x,CE_{v_x}(P_{2|x})) P_1(x),\forall~P\in \cP$$
where $v_x$ is a regular function for each $x\in X_1$. This is exactly the functional form stated in \lemmaref{lemma_kp}. The final step is to verify that $\hat{u}_D$ is a regular function. 

Monotonicity can be guaranteed by Axiom M. WLOG, let $\hat{u}_D(0,0)=0$. To see why $\hat{u}_D$ is bounded, notice that $\succsim$ satisfies mixture continuity on $\hat{\cP}$. Suppose by contradiction that $\hat{u}_D$ is unbounded from above. Then $\cone >0$.  Denote $\hat{u}_D(\cone/2,0)=a>0$. For any $n$, we can find $z_n>z_{n-1}>0$ and $z_n'\geq 0$ such that $\hat{u}_D(z_n,z_n')>n^2a$, which implies $V((\frac{1}{n}\delta_{z_n}+ \frac{n-1}{n}\delta_0,\frac{1}{n}\delta_{z_n'}+ \frac{n-1}{n}\delta_0))>a$ and hence $(\frac{1}{n}\delta_{z_n}+ \frac{n-1}{n}\delta_0,\frac{1}{n}\delta_{z_n'}+ \frac{n-1}{n}\delta_0)\succ (\delta_{\cone/2},\delta_0)$ for each $n$.  However, by continuity of $\succsim$ over product lotteries, as $n\rightarrow \infty$, we must have $(\delta_0,\delta_0)\succsim  (\delta_{\cone/2},\delta_0)$, a contradiction.

Then we show that $\hat{u}_D$ is continuous. Again, we normalize $\hat{u}_D(0,0)=0$. Suppose by contradiction that $\hat{u}_D$ is not continuous, then we can find $(x,y)\in X_1\times X_2$ and a sequence $(x_n,y_n)\rightarrow (x,y)$ such that $\lim_{n\rightarrow \infty} \hat{u}_D(x_n,y_n)\neq \hat{u}_D(x,y)$. Then there exists a bounded subsequence of $\{(x_n,y_n)\}$ (still denoted as $\{(x_n,y_n)\}$ given there is no confusion) such that either $\hat{u}_D(x_n,y_n)\leq \hat{u}_D(x,y)$ for all $n$ or $\hat{u}_D(x_n,y_n)\geq \hat{u}_D(x,y)$ for all $n$. By symmetry, we will focus on the former case. 
Since $\hat{u}_D$ is bounded, $\{\hat{u}_D(x_n,y_n)\}_{n\geq 1}$ admits a convergent subsequence (again we still denote the subsequence as the sequence itself). Then it must be the case that $\lim_{n\rightarrow \infty}\hat{u}_D(x_n,y_n)=a<b = \hat{u}_D(x,y)$.

We claim that we can find some $P\in \cP$ with $V(P)=\frac{a+b}{2}$. To see this, first suppose that $x_n=x$ for all $x$ large enough. Without loss of generality, we can assume $x_n=x$ for all $x$. If $x\neq \lcone$, then fix any $\lcone<x'<x$, we have $\hat{u}_D(x',y_m)<\frac{a+b}{2}$ for some $m$ large enough since $\hat{u}_D$ is monotone and $\lim_{n\rightarrow \infty}\hat{u}_D(x_n,y_n)=a$. Then there exists $\eta\in (0,1)$ such that $V(\eta(\delta_{x'},\delta_{y_m})+(1-\eta)(\delta_{x},\delta_{y}))=\eta \hat{u}_D(x',y_m)+(1-\eta)b=\frac{a+b}{2}$. If $x=\lcone$, then fix any $x<x'<\cone$, we have $\hat{u}_D(x',y)>b$ and $\hat{u}_D(x,y_m)<\frac{a+b}{2}$ for some $m$ large enough. Again, there exists $\eta\in (0,1)$ such that $V(\eta(\delta_{x},\delta_{y_m})+(1-\eta)(\delta_{x'},\delta_{y}))=\eta \hat{u}_D(x,y_m)+(1-\eta)\hat{u}_D(x',y)=\frac{a+b}{2}$. Now suppose that we can find $m$ large enough such that $x_m\neq x$ and $\hat{u}_D(x_m,y_m)<\frac{a+b}{2}$. Then there exists $\eta\in (0,1)$ such that $V(\eta(\delta_{x_m},\delta_{y_m})+(1-\eta)(\delta_{x},\delta_{y}))=\eta \hat{u}_D(x_m,y_m)+(1-\eta)b=\frac{a+b}{2}$.

As $ \lim_{n\rightarrow \infty}\hat{u}_D(x_n,y_n)=a<\frac{a+b}{2}$, for $n$ large enough, we have $\hat{u}_D(x_n,y_n)<\frac{a+b}{2}<b$, that is, $(\delta_{x_n},\delta_{y_n})\prec P$. Let $n$ goes to infinity and we know $(\delta_{x_n},\delta_{y_n})\xrightarrow[]{w} (\delta_x,\delta_y)$. By Axiom Topological Continuity over Product Lotteries (the second part of Axiom WC), $(\delta_x,\delta_y)\precsim P$, that is, $b<\frac{a+b}{2}$, a contradiction. 
Hence $\hat{u}_D$ is continuous and this completes the proof.\end{proof}

\bigskip

{\bf Step 4: Finally we check the consistency of the previous representations on $\hat{\cP}$.} 

Now we have two representations on $\hat{\cP}$: the EU-CN, GBIB-CN and GFIB-CN representations in {\bf Step 1} and the KP-style representations in \lemmaref{lemma_kp}. Both of them represent $\succsim$ on $\hat{\cP}$ and we will explore the implications of such consistency. 

By \lemmaref{lemma_kp}, we know that $\succsim$ on $\hat{\cP}$ can be represented by 
$$U(P_1,P_2)= \sum_x \hat{u}_D(x, CE_{v_x}(P_{2}))P_1(x), \forall (P_1,P_2)\in \hat{\cP}$$
where $\hat{u}_D$ and $v_x$ for all $x\in X_1$ are regular.

First, suppose that $\succsim$ admits an EU-CN representation $w$ on $\hat{\cP}$. Assume that $\cone>0$. The case with $\cone=0$ is symmetric. Fix $0<a<\cone$. As $w$ and $\hat{u}_D$ are unique up to a positive affine transformation, we can normalize $w(0,0)=\hat{u}_D(0,0)=0$ and $w(a,0)=\hat{u}_D(a,0)=b>0$.  There exists a continuous and monotone function $\phi: w(X_1,X_2)\rightarrow \bR$ such that for all $P\in \hat{\cP}$,
$$U^{KP}(P_1,P_2) = \phi\circ U^{EU-CN}(P_1,P_2).$$

Now focus on $\lxone\times \{\delta_y\}$ for some $y\in X_2$. We know for all $p\in \lxone$,
$$U^{KP}(p,\delta_y)= \sum_{x}\hat{u}_D(x,y)p(x) = \phi\circ U^{EU-CN}(p,\delta_y) = \phi[\sum_xw(x,y)p(x)].$$
Then we know $\phi$ must be linear on $w(X_1,y)$ for each y. By continuity of $w$, by ranging over $y\in X_2$, we know  $\phi$ must be linear on its domain $w(X_1,X_2)$. That is, for all $t_1,t_2\in w(X_1,X_2)$ and $\alpha\in (0,1)$, $\phi(\alpha t_1 + (1-\alpha)t_2) = \alpha \phi(t_1) + (1-\alpha)\phi(t_2)$. Also, by our normalization,  $\phi(0)=0$ and $\phi(b)=b$. Then for any $t\in (0,b)$, $\phi(t)=\phi(\frac{t}{b}\cdot b + (1-\frac{t}{b})\cdot 0) = \frac{t}{b} \phi(b) + (1-\frac{t}{b})\phi(0)= t$. We show also that $h(t)=t$ for $t>b$ or $t<0$. Thus $w\equiv \hat{u}_D$.

Second, suppose that $\succsim$ admits a GBIB-CN representation $(w,v'_1,v'_2,H_2)$. Then for each $x\in  X_1$, on $\{\delta_x\}\times \lxtwo$, 
$$U^{GBIB-CN}(\delta_x,p)= w(x,CE_{v'_2}(p))$$
Similarly, $U^{KP}(\delta_x,p) = \hat{u}_D(x,CE_{v_x}(p))$. Since $v_x$ and $v'_2$ are regular, consistency of BIB-CN and KP on $\{\delta_x\}\times \lxtwo$ requires that $v_x$ is a positive affine transformation of $v'_2$, that is, $CE_{v_x}(p)= CE_{v'_2}(p)$ for all $x\in X_1$. Hence  for any $P\in \cP$,  \begin{align*}
   U^{KP}(P)&= \sum_x \hat{u}_D(x,CE_{v_x}(P_{2|x}))P_1(x) \\
  &= \sum_x \hat{u}_D(x,CE_{v'_2}(P_{2|x}))P_1(x)
\end{align*}
Thus,  $\succsim$ admits an BIB representation $(\hat{u}_D,v'_2)$.

Finally, suppose that $\succsim$ admits a GFIB-CN representation $(w,v'_1,v'_2,H_1)$. That is, 
$$V^{GFIB-CN}(P_1,P_2)= \begin{cases}
w(CE_{v'_1}(P_1),CE_{v'_2}(P_2)), \hbox{~if~} CE_{v'_1}(P_1)\not\in H_1\\
\sum w(CE_{v'_1}(P_1),y)P_2(y), \hbox{~if~} CE_{v'_1}(P_1)\in H_1\end{cases}$$

 As both utility functions represent $\succsim$ on $\hat{\cP}$, we can find a monotone and continuous function $\phi:w(X_1,X_2)\rightarrow \bR$ with 
$$U^{KP}(P_1,P_2)= \phi\circ U^{GFIB-CN}(P_1,P_2),\forall P\in\hat{\cP}.$$

We first focus on $\lxone \times \{p_2\}$ for some $p_2\in \lxtwo$. Based on the GFIB-CN, the preference $\succsim_{1|p_2}$ admits an EU representation with the index $v'_1$. Also, note that $U^{KP}$ is linear in the first source for fixed $p_2$. Hence for each $p_2\in \lxtwo$, $w(\cdot, CE_{v_x}(p_2))$ must be a positive affine transformation of $v'_1$. That is, there exists functions $\hat{\alpha}$ and $\hat{\beta}$ defined on $\lxtwo$ such that $\hat{a}(p_2)>0$, $\hat{b}(p_2)\in\bR$ for all $p_2\in \lxtwo$ and 
\begin{equation*}
    w(x,CE_{v_x}(p_2)) = \hat{a}(p_2)v'_1(x) + \hat{b}(p_2), \forall x\in X_1.
\end{equation*}
Specifically, if $p_2=\delta_y$ for some $y\in X_2$, then we know 
$$w(x,y)=\hat{a}(\delta_y)v'_1(x) + \hat{b}(\delta_y), \forall (x,y)\in X_1\times X_2.$$
Define $a,b$ as functions on $X_2$ by $a(y)=\hat{a}(\delta_y)$ and $b(y)=\hat{b}(\delta_y)$. This implies 
\begin{align*}
     w(x,CE_{v_x}(p_2)) &= \hat{a}(p_2)v'_1(x) + \hat{b}(p_2)\\
     &= a(CE_{v_x}(p_2)) v'_1(x) + b(CE_{v_x}(p_2)), \forall x\in X_1, q\in \lxtwo
\end{align*}

If $H_1=\emptyset$, then GFIB-CN reduces to NB, which is a special case of GBIB-CN and it has been covered in the second case. From now on, assume $H_1\neq\emptyset$. For any fixed $x_1,x_2\in X_1$ with $x_1<x_2$, we can always normalize $v'_1(x_1)=0$ and $v'_1(x_2)=1$. Plug  the two values into the previous equation, we get for any $p_2\in \lxtwo$,
$$\hat{b}(p_2)= b(CE_{v_{x_1}}(p_2))~,~ \hat{a}(p_2)= a(CE_{v_{x_2}}(p_2))+ b(CE_{v_{x_2}}(p_2)) - b(CE_{v_{x_1}}(p_2)).$$

First, suppose that there exists $x_1<x_2$ with $x_1,x_2\not\in H_1$. Then $\hat{u}_D(x_i, CE_{v_{x_i}}(p_2))= \phi(  w(x_i,CE_{v'_2}(p_2)))$ for $i=1,2$ and hence  for any $x\in X_1$, $p_2\in\lxtwo$, 
$$\hat{u}_D(x, CE_{v_{x}}(p_2)) = \phi(  w(x_2,CE_{v'_2}(p_2)))\cdot v'_1(x) + \phi(  w(x_1,CE_{v'_2}(p_2)))\cdot (1-v'_1(x)).$$

Clearly, the RHS depends on $p_2$ only through its certainty equivalent under $v'_2$. Then $v_x$ must be a positive affine transformation of $v'_2$ for all $x\in X_1$ and the KP representation reduces to a BIB representation $(\hat{u}_D,v'_2)$. 

Second, suppose that there exists $x_1<x_2$ with $x_1,x_2\in H_1$ and $w(x_1,\cdot)$ is a positive affine transformation of $w(x_2,\cdot)$. As 
$x_1,x_2\in H$, we know $\hat{u}_D(x_i, CE_{v_{x_i}}(p_2))= \phi( \sum_y w(x_i,y)p_2(y))$ for $i=1,2$, which further implies for any $x\in X_1$, $p_2\in\lxtwo$, 
$$\hat{u}_D(x, CE_{v_{x}}(p_2)) = \phi( \sum_y w(x_2,y)p_2(y))\cdot v'_1(x) + \phi( \sum_y w(x_1,y)p_2(y))\cdot (1-v'_1(x)).$$
Since $w(x_1,\cdot)$ is a positive affine transformation of $w(x_2,\cdot)$., the above equation can be rewritten as 
$$\hat{u}_D(x, CE_{v_{x}}(p_2)) = g(x, \sum_y w(x_2,y)p_2(y))$$
for some function $g$. Notice that $g$ depends on $p_2$ only through its expected value under $w(x_2,\cdot)$. Hence the $v_x$ must be a positive affine transformation of  $w(x_2,\cdot)$ for all $x\in X_1$. Denote $w(x_2,\cdot)$ as $w(\cdot)$, then the KP representation reduces to a BIB representation $(\hat{u}_D,w)$.

Finally, suppose $cl(H_1)=X_1$ and for all $x_1\neq x_2$, $w(x_1,\cdot)$ is not a positive affine transformation of $w(x_2,\cdot)$. Without loss of generality, we suppose $\cone>0$ and set $w(0,0)=0$, $x_1=0,x_2=\cone/2$. Denote $w(x_1,\cdot)=w_1(\cdot)$ and $w(x_2,\cdot)=w_2(\cdot)$. Then we know $w_1$ is not a positive affine transformation of $w_2$. We denote the preference represented by EU index $w_i$ as $\succsim_{w_i}$. Then $\succsim_{w_1}\neq \succsim_{w_2}$. Since $\phi$ is strictly increasing and continuous, $\phi$ is almost everywhere differentiable. Recall our normalization that $v'_1(x_1)=0$ and $v'_1(x_2)=1$.
Then the previous argument applies and we can show that for any $x\in X_1$, $p_2\in\lxtwo$, 
$$\hat{u}_D(x, CE_{v_{x}}(p_2)) = \phi( w_2(p_2)) v'_1(x) + \phi( w_1(p_2)) (1-v'_1(x)).$$
Denote
$$\mathcal{H}= \big\{ p\in \lxone: \exists~q\in \lxtwo, ~s.t.~\big(w_1(p)-w_1(q)\big)\cdot\big(w_2(p)-w_2(q)\big)<0 \big\}$$
That is, $\mathcal{H}$ is the set of single-source lotteries such that there exists another lottery where the two preferences represented by $w_1$ and $w_2$ disagree on the ranking of the two lotteries. We claim that $\mathcal{H}=\lxone \backslash\{\delta_{\cone},\delta_{\lcone}\}$. 
Clearly $\delta_{\cone},\delta_{\lcone}\not\in \mathcal{H}$. 

For any $q\neq \delta_{\cone},\delta_{\lcone}$ with $q\sim_{w_i}\delta_y$ for $i=1,2$, we know $y\in (\lcone,\cone)$. Since $\succsim_{w_1}\neq \succsim_{w_2}$, we can find $q'\sim_{w_1}\delta_y$ and $q'\not\sim_{w_2}\delta_y$, otherwise $\succsim_{w_1}$ and $\succsim_{w_2}$ share the same indifference curve with certainty equivalent $y\in (\lcone,\cone)$ and they should be the same EU preference. Then we can assume $q'\not\sim_{w_2}q$, otherwise, we can take a mixture between $q'$ and $\delta_y$. Without loss of generality, assume $q'\succ_{w_2}q$. Then we can choose $q''$ slightly dominated by $q'$ and by continuity, we have $q''\succ_{w_2}q$ and $q\succ_{w_1}q''$. 

Now for any $q\in \mathcal{H}$, take $p$ such that, without loss of generality, $p\succ_{w_2}q$ and $q\succ_{w_1}p$. 
Then we can find $x\in (0,1)$ such that
\begin{equation*}
    \frac{1-v'_1(x)}{v'_1(x)} = \frac{\phi(w_2(p))-\phi(w_2(q))}{\phi(w_1(q)-\phi(w_1(p))}.
\end{equation*}
Such $x$ exists as the RHS is strictly positive and the the range of $\frac{1-v'_1(x)}{v'_1(x)}$ for $x\in(0,1)$ is $(0,+\infty)$. Rearrange the above equation, we can get 
$$\phi(w_2(q))v'_1(x) + \phi(w_1(q))(1-v'_1(x)) = \phi(w_2(p))v'_1(x) + \phi(w_1(p))(1-v'_1(x)) $$
that is, $(\delta_x,p)\sim (\delta_x, q)$. By conditional independence, for any $\beta\in [0,1]$, $(\delta_x, \beta p +(1-\beta)q)\sim (\delta_x,p)\sim (\delta_x, q)$. This implies that for all $\beta\neq \beta'$
$$\frac{\phi(\beta w_2(p)+(1-\beta)w_2(q))-\phi(\beta' w_2(p)+(1-\beta')w_2(q))}{\phi(\beta' w_1(p)+(1-\beta')w_1(q))-\phi(\beta w_1(p)+(1-\beta)w_1(q))}= \frac{1-v'_1(x)}{v'_1(x)}= \frac{\phi(w_2(p))-\phi(w_2(q))}{\phi(w_1(q)-\phi(w_1(p))}.$$

Take $\beta'=0$ and let $\beta\rightarrow 0^+$, we have 
$$\frac{\partial_+\phi(w_2(q)) \slash \partial{x} }{\partial_-\phi(w_1(q))\slash \partial{x} }= \frac{\phi(w_2(p))-\phi(w_2(q))}{\phi(w_1(q)-\phi(w_1(p))} \frac{w_1(q)-w_1(p)}{w_2(p)-w_2(q)}. $$
We argue that the two semi-derivatives are well-defined. The RHS is always well-defined. Notice that by continuity of $\phi$, we can change $q$ slightly to change $w_2(q)$ without changing $w_1(q)$. This is possible as $\succsim_{w_1}\neq \succsim_{w_2}$. If the semi-derivative $\partial_-\phi(w_1(q))\slash \partial{x}$ does not exist, then $\partial_+\phi(x)\slash \partial{x}$ does not exist for $x$ in a open interval, which contradicts with the fact that $\phi$ is almost everywhere differentiable. A similar proof can show that $\partial_+\phi(w_2(q)) \slash \partial{x}$ is well-defined.

Let $\beta=1$ and $\beta'\rightarrow 1^-$, we can get 
\begin{equation}\label{eq_thm2_constant}
    \frac{\partial_-\phi(w_2(p)) \slash \partial{x} }{\partial_+\phi(w_1(p))\slash \partial{x} }= \frac{\phi(w_2(p))-\phi(w_2(q))}{\phi(w_1(q)-\phi(w_1(p))} \frac{w_1(q)-w_1(p)}{w_2(p)-w_2(q)} = \frac{\partial_+\phi(w_2(q)) \slash \partial{x} }{\partial_-\phi(w_1(q))\slash \partial{x} }.
\end{equation}
Again, all the semi-derivatives are well-defined.

Fix $p$ and choose $q'\sim_{w_2}q$. we can find $q'\not\sim_{w_1} q$ and $q'\succ_{w_1}p$. Without loss of generality, let $w_1(q')>w_1(q)$. The other case can be proved symmetrically. Then for any $\alpha\in (0,1)$, we can redo the above calculation for $\alpha q' +(1-\alpha)q \succ_{w_1} p$ and $p\succ_{w_2} q' +(1-\alpha)q\sim_{w_2}q $. Then the left equality of equation (\ref{eq_thm2_constant}) becomes: 
\begin{equation*}
    \frac{\partial_-\phi(w_2(p)) \slash \partial{x} }{\partial_+\phi(w_1(p))\slash \partial{x} }=\frac{\phi(w_2(p))-\phi(w_2(q))}{w_2(p)-w_2(q)} \frac{\alpha w_1(q) + (1-\alpha)w_1(q')-w_1(p)}{\phi(\alpha w_1(q) + (1-\alpha)w_1(q'))-\phi(w_1(p))}.
\end{equation*}
This implies that 
$$\frac{\phi(\alpha w_1(q) + (1-\alpha)w_1(q'))-\phi(w_1(p))}{\alpha w_1(q) + (1-\alpha)w_1(q')-w_1(p)}$$
 is a constant as $\alpha$ varies in $(0,1)$.
Hence $\phi'(z)=C$  for all $z\in (w_1(q),w_1(q'))$, where $C>0$ is a constant.

Moreover, for any $z_1,z_2\in w(X_1,X_2)$  with $w(\cone,\ctwo)>z_1>z_2>w(\lcone,\lctwo)$. 
Given ${y}\in (\lcone,\cone)$, for any $\beta\in (0,1)$, we know $\beta p+(1-\beta)\delta_{{y}}\succsim_{w_2} \beta q+(1-\beta)\delta_{{y}}\sim_{w_2}\beta q'+(1-\beta)\delta_{{y}}$ and $\beta q'+(1-\beta)\delta_{{y}}\succsim_{w_1} \beta q+(1-\beta)\delta_{{y}}\sim_{w_1}\beta p+(1-\beta)\delta_{{y}}$. By the above argument, we know that $\phi'(z)$ is a constant for $z\in (\beta w_1(q) + (1-\beta)w_1(\delta_y), \beta w_1(q') + (1-\beta)w_1(\delta_y))$. By continuity of $w_1$, those open intervals are intersecting with each other.  We make ${y}$ large enough and small enough respectively, so that we can get an open cover of $[z_2,z_1]$. Then there exists a finite subcover can we have $\phi'(z_1)=\phi'(z_2)=C$ for any $w(\cone,\ctwo)>z_1>z_2>w(\lcone,\lctwo)$ with $z_1,z_2\in w(X_1,X_2)$. This implies $\phi(z)=Cz+ b$ for $z\in (w(\lcone,\lctwo),w(\cone,\ctwo))\cap w(X_1,X_2)$ and by continuity of $\phi$, $\phi(z)=Cz+ b$ holds for all $z\in w(X_1,X_2)$.

Since $w$ and $\hat{u}_D$ are unique up to positive affine transformation, we can set $\phi(z)=z$ for all $z$ without loss of generality. Then we know that for all $x\in X_1,p_2\in \lxtwo$
$$\hat{u}_D(x,CE_{v_x}(p_2))=\sum_y w_2(y)p_2(y)v'_1(x) + \sum_y w_1(y)p_2(y) (1-v'_1(x))$$
Thus the representation on $\cP$ is given by 
\begin{align*}
   U^{KP}(P)&=\sum_x \hat{u}_D(x,CE_{v_x}(P_{2|x}))P_1(x)\\
   &= \sum_{x,y}\big( [w_2(y)-w_1(y)]v'_1(x) + w_1(y)\big)P(x,y), 
\end{align*}
for all $P\in\cP$. This is of course an EU representation. 

As a summary of {\bf Step 4}, in all possible cases, $\succsim$ on $\cP$ can be represented by either an EU or a BIB representation. 

Combining the results in {\bf Step 1} and {\bf Step 4}, given the axioms stated in \thmref{thm_BIB}, the relation $\succsim$ admits one of the following representations: EU, BIB, EU-CN, GBIB-CN and GFIB-CN.\end{proof}

\end{document}